\documentclass[12pt]{article}
\usepackage{amsmath,amsthm,amssymb,amsfonts}
\usepackage{mathtools}
\usepackage{graphicx}
\usepackage{booktabs,multirow,array}
\usepackage{enumitem}
\usepackage{natbib}
\usepackage[T1]{fontenc}
\usepackage{xcolor}
\usepackage[colorlinks,linkcolor=red,anchorcolor=blue,citecolor=blue,urlcolor=blue]{hyperref}
\usepackage{bookmark}
\usepackage[margin=1in]{geometry}
\usepackage[protrusion=false,expansion=true]{microtype}
\usepackage[title]{appendix}
\usepackage{tikz}
\usetikzlibrary{arrows.meta,positioning}
\usepackage[parfill]{parskip}
\usepackage{algorithm}
\usepackage[noend]{algpseudocode}

\theoremstyle{plain}
\newtheorem{theorem}{Theorem}[section]
\newtheorem{lemma}[theorem]{Lemma}
\newtheorem{proposition}[theorem]{Proposition}
\newtheorem{corollary}[theorem]{Corollary}
\theoremstyle{definition}

\newtheorem{assumption}[theorem]{Assumption}

\theoremstyle{remark}
\newtheorem{remark}[theorem]{Remark}

\newcommand{\E}{\mathbb{E}}
\newcommand{\R}{\mathbb{R}}
\newcommand{\Pn}{\mathbb{P}_n}
\newcommand{\Pstar}{\mathbb{P}^\star}
\newcommand{\ip}[2]{\left\langle #1,\,#2\right\rangle}
\newcommand{\norm}[1]{\left\lVert #1\right\rVert}

\newcommand{\infnorm}[1]{\left\lVert #1\right\rVert_{\infty}}
\newcommand{\Fnorm}[1]{\left\lVert #1\right\rVert_{F}}
\newcommand{\oneNorm}[1]{\left\lVert #1\right\rVert_1}
\newcommand{\twoninfnorm}[1]{\left\lVert #1\right\rVert_{2,\infty}}
\newcommand{\Thetastar}{\Theta^{\star}}
\newcommand{\Thetahat}{\widehat{\Theta}}
\newcommand{\dt}{d_{\mathrm{t}}}
\newcommand{\dm}{d_{\mathrm{m}}}

\newcommand{\dstar}{d^{\star}}
\newcommand{\rk}{\mathrm{rk}}
\newcommand{\TopK}{\mathcal{S}^{\star}_K}

\newcommand{\cL}{\mathcal{L}}
\newcommand{\rank}{\mathrm{rank}}
\newcommand{\Var}{\mathrm{Var}}

\newif\ifshowrevisions
\showrevisionsfalse
\definecolor{revisionred}{RGB}{190,25,25}
\newcommand{\revisioncolor}{\ifshowrevisions\color{revisionred}\fi}

\title{\Large{\textbf{Low Rank for Rank: Uncertainty-Aware Task-Specific LLM Ranking under Sparse Pairwise Comparisons}}}
 \author{
 Jiachun Li\thanks{Laboratory for Information and Decision Systems, MIT. Email: jiach334@mit.edu.}\qquad
 David Simchi-Levi\thanks{Laboratory for Information and Decision Systems, MIT. Email: dslevi@mit.edu.}\qquad
 Will Wei Sun\thanks{Daniels School of Business, Purdue University. Email: sun244@purdue.edu.}
 }

\date{}
\begin{document}
\maketitle

\begin{abstract}
%\begin{revisionblock}
Pairwise human-preference platforms such as Chatbot Arena have become central to large language model evaluation, yet reliable task-specific ranking remains challenging. Global leaderboards can mask substantial task heterogeneity, while independently estimating rankings for each fine-grained task is statistically unstable under sparse and imbalanced comparisons. We propose a low-rank framework for task-specific LLM ranking from sparse pairwise comparisons, assuming that the task-by-model ability matrix is low-rank so that information can be shared across related tasks while preserving task-specific differences. We introduce a ranking-recovery algorithm that combines convex initialization with pairwise-logistic row and column refinement, obtains max-norm accuracy, and yields task-wise top-$K$ recovery guarantees. Beyond point estimation, we construct cross-fitted one-step estimators for task-specific score gaps, characterize their full joint covariance, and show that they attain the semiparametric efficiency lower bound. For the high-dimensional ranking regime, where ranks and top-$K$ membership depend on many correlated gap signs, we apply Gaussian multipliers to the estimated influence matrix and invert the resulting simultaneous bands into rank intervals and three-way top-$K$ certificates. Experiments on synthetic data and real Arena data show that low-rank sharing improves sample efficiency over independent task-wise Bradley--Terry estimation and enables uncertainty-aware task-specific certification.
%\end{revisionblock}
\end{abstract}

\noindent%
{\it Keywords:} Large language model evaluation; pairwise comparisons; low-rank models; task-specific ranking; uncertainty quantification.
%\vfill

\newpage
\baselineskip=25pt % JASA allows at most 26 lines per page
\section{Introduction}\label{sec:intro}

Pairwise human-preference evaluation has become a central tool for comparing large
language models (LLMs). Platforms such as Chatbot Arena
\citep{chiang2024chatbot} collect side-by-side comparisons of model responses and
aggregate them into public leaderboards, providing a scalable alternative to fixed
benchmark scores. At the same time, modern LLM evaluation is increasingly
task-specific, with benchmarks and evaluation platforms reporting performance separately
for coding, mathematical reasoning, instruction following, multilingual tasks, creative
writing, and other fine-grained categories
\citep{white2025livebench,frick2025prompt,moslem2026dynamic}. This task-specific view
is essential because model strengths are heterogeneous: a model that is strong overall
may not be the best choice for a particular task, user group, or deployment domain.

The central statistical problem is therefore not only to estimate model abilities, but
to make reliable task-specific ranking decisions. Practitioners often want to know
which models can be trusted as top performers for a task, whether an apparent difference
between two models is statistically significant, and which leaderboard claims remain
valid after many comparisons are considered simultaneously \citep{avelar2026prompt, su2026large}. A point leaderboard alone
cannot answer these questions. When comparisons are sparse or imbalanced, especially
within fine-grained task categories, small estimated score gaps near the top-\(K\)
boundary may simply reflect sampling noise. Thus task-specific LLM evaluation calls for
uncertainty-aware ranking, with valid uncertainty quantification for ranks, top-\(K\) membership, and
task-specific model comparisons.

Existing approaches leave a gap. Independent task-wise Bradley-Terry-Luce (BTL)
estimation \citep{bradley1952rank,luce1959individual} respects task heterogeneity, but
can be statistically unstable when each task receives limited or uneven comparisons.
Fully pooling data across tasks reduces variance, but erases the task-specific
heterogeneity that motivates fine-grained evaluation in the first place. Recent
low-rank approach \citep{li2026llmevaluation} provides an attractive compromise by sharing information across
related tasks, but smooth score estimation or inference for fixed functionals does not
by itself solve the ranking problem. Ranks and top-\(K\) membership are nonsmooth
functionals determined by many dependent score-gap signs. Valid task-specific
leaderboards therefore require new tools for boundary-sensitive top-\(K\) recovery,
multiple testing over correlated score gaps, and simultaneous rank certification.

We develop a statistical framework for \emph{certified task-specific leaderboards} from
sparse pairwise comparisons. Let
\(\Theta^\star\in\mathbb{R}^{d_t\times d_m}\) denote the latent task-by-model ability
matrix, where \(\Theta^\star_{t,m}\) is the score of model \(m\) on task \(t\). Each
observation consists of a task \(t\), two models \(m,m'\), and a binary preference
outcome following a BTL-type model that depends on the score difference
\(\Theta^\star_{t,m}-\Theta^\star_{t,m'}\). We assume that \(\Theta^\star\) has
rank \(r\), reflecting shared latent capabilities, such as reasoning,
instruction following, or style sensitivity, across tasks. This setting is more
challenging than standard low-rank matrix completion
\citep{candes2009exact,negahban2012restricted,davenport20141bit}, because each
observation is a binary, within-task comparison depending only on a score difference,
rather than a noisy observation of an individual matrix entry. More importantly, the
main target is not merely low-rank score estimation, but statistically valid ranking and
certification under sparse, non-uniform, and dependent pairwise comparisons.

Moving from score estimation to ranking certification introduces additional
difficulties. First, top-\(K\) accuracy
requires entrywise control of task-model scores, rather than only Frobenius or
prediction-error accuracy. Second, rank and top-\(K\) decisions are determined by the
signs of many score gaps, so uncertainty quantification naturally becomes a multiple
hypothesis testing problem. Third, these score-gap statistics are strongly dependent:
they share tasks, models, comparisons, and low-rank latent factors. As a result,
pointwise confidence intervals or independent testing corrections are not sufficient for
valid leaderboard-level guarantees. 

\begin{figure}[h!]
\centering
\begin{tikzpicture}[
    node distance=0.8cm and 0.8cm,
    box/.style={
        rectangle, rounded corners,
        draw=black!70, thick,
        align=center,
        minimum width=2.55cm,
        minimum height=0.9cm,
        font=\small
    },
    widebox/.style={
        rectangle, rounded corners,
        draw=black!70, thick,
        align=center,
        minimum width=2.9cm,
        minimum height=0.9cm,
        font=\small
    },
    smallbox/.style={
        rectangle, rounded corners,
        draw=black!60,
        align=center,
        minimum width=2.55cm,
        minimum height=0.85cm,
        font=\small
    },
    arrow/.style={->, thick}
]

\node[box] (data) {
Sparse task-specific\\
pairwise comparisons\\
\(\{t_i,m_i,m_i',Y_i\}_{i=1}^n\)
};

\node[box, right=of data] (lowrank) {
Shared low-rank\\
ability matrix\\
\(\Theta^\star\in\mathbb{R}^{d_t\times d_m}\)
};

\node[widebox, right=of lowrank] (entrywise) {
Entrywise error control\\
\(\|\widehat\Theta-\Theta^\star\|_\infty\)
};

\node[smallbox, below=of entrywise] (testing) {
Multiple testing for\\
score-gap hypotheses\\
Sec.~4
};

\node[smallbox, left=0.3cm of testing] (topk) {
Task-wise top-\(K\)\\
accuracy\\
Sec.~3
};

\node[smallbox, right=0.3cm of testing] (certify) {
Simultaneous rank and \\top-$K$ certification\\
Sec.~5
};

\draw[arrow] (data) -- (lowrank);
\draw[arrow] (lowrank) -- (entrywise);
\draw[arrow] (entrywise) -- (testing);
\draw[arrow] (entrywise) -- (topk);
\draw[arrow] (entrywise) -- (certify);
\draw[arrow] (testing) -- (certify);

\end{tikzpicture}
\caption{Paper outline: from sparse pairwise comparisons to certified task-specific leaderboards.}
\label{fig:method-overview}
\end{figure}
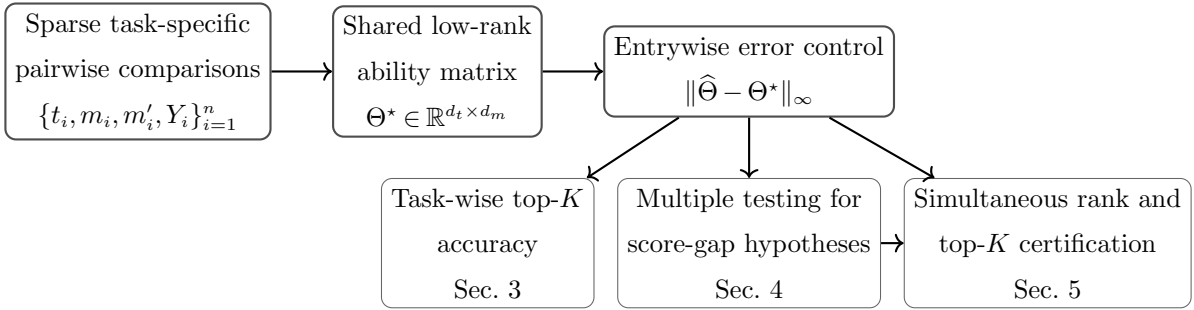

Figure~\ref{fig:method-overview} summarizes how we turn low-rank preference estimates
into certified task-specific leaderboard decisions. Our contributions can be summarized as follows.

\begin{itemize}
%\begin{revisionblock}
    \item \textbf{Global-to-local refinement for ranking recovery.}
    We develop a multi-stage ranking procedure that first recovers the shared low-rank geometry through convex initialization and then performs model-row and task-column pairwise-logistic refinements to achieve the entrywise score accuracy. The refinement exploits low-dimensional score equations and model-side recentering for identifiability to upgrade global Frobenius accuracy to uniform max-norm control, which in turn yields task-wise top-\(K\) Hamming and exact-recovery guarantees.
    %We combine a convex low-rank initializer with model-row and task-column pairwise-logistic refinements, exact model centering, and a final top-\(K\) decision rule.  The accompanying analysis explains how local score equations upgrade global Frobenius accuracy to max-norm accuracy, and then converts this resolution into task-wise top-\(K\) Hamming and exact-recovery guarantees.

    \item \textbf{Efficient score-gap inference with an explicit covariance interface.}
    We construct cross-fitted one-step estimators for task-specific score gaps using the
    centered low-rank tangent space and restricted Fisher operator.  We prove the
    semiparametric efficiency bound, show that the estimator attains it jointly, and
    express the covariance of all correlated ranking hypotheses through a common matrix
    of estimated influence values.

   \item \textbf{Dependence-aware simultaneous inference for ranking certification.} We develop a multiplier-bootstrap procedure that propagates the full dependence structure of the estimated influence functions across many correlated score-gap comparisons, while avoiding explicit inversion or factorization of a large covariance matrix. The resulting simultaneous confidence bands are translated directly into rank intervals and three-way top-\(K\) certificates: certified top-\(K\), certified non-top-\(K\), or statistically unresolved.
    
    %\item \textbf{Covariance-adaptive multiple testing and rank certification.}
    %We give an implementable multiplier algorithm that applies common random weights to the influence matrix, thereby retaining the complete empirical dependence among studentized gaps without inverting a large covariance matrix.  Simultaneous bands are inverted into rank intervals and three-way top-\(K\) decisions: certified top-\(K\), certified non-top-\(K\), or statistically unresolved.
%\end{revisionblock}

    \item \textbf{Reliable ranking from sparse real-world comparisons.} Synthetic experiments validate the statistical gains predicted by our theory, while an LM Arena study shows that these gains persist in heterogeneous real-world leaderboard data. By sharing information across ten task categories and thirty leading models, the proposed method substantially improves top-\(K\) recovery when task-level comparisons are scarce and produces markedly more informative rank and top-\(K\) certificates than independent task-wise BTL estimation.

    %\item \textbf{Empirical validation on task-specific leaderboards.}
    %Synthetic and Arena experiments show that low-rank sharing improves top-\(K\) recovery and produces shorter, better-calibrated rank and membership certificates than independent task-wise BTL estimation.
\end{itemize}

\subsection{Related work and positioning.}

\textbf{Arena-style LLM evaluation.}
Arena-style platforms have made pairwise human-preference comparison a standard tool
for evaluating LLMs \citep{chiang2024chatbot,arena2026howitworks,arena2026max}.
Modern leaderboards increasingly report category-specific results for coding, math,
creative writing, and instruction following \citep{arena2026max},
reflecting substantial task-level heterogeneity. Recent work further studies
prompt-dependent leaderboards and routing: \cite{frick2025prompt} learn
prompt-dependent Bradley-Terry coefficients, while \cite{avelar2026prompt, neuhof2026quantifying} develop uncertainty quantification for LLM rankings across prompts or benchmark tasks. These works motivate
fine-grained and uncertainty-aware evaluation, but independent category- or prompt-wise
ranking does not share information across related tasks and can be inefficient under
sparse comparisons. Our framework uses low-rank structure to share information across
tasks while preserving task-specific rankings and providing simultaneous uncertainty
quantification for ranks and top-\(K\) membership.

\textbf{Pairwise ranking, uncertainty quantification, and low-rank inference.}
Our theory builds on pairwise ranking and top-\(K\) recovery under BTL-type models
\citep{bradley1952rank,luce1959individual,hunter2004mm,
chen2015spectral,chen2022partial}, as well as recent uncertainty quantification for
sparse, covariate-assisted, and heterogeneous ranking models
\citep{fan2024care,fan2025uncertainty,fan2026spectral, liu2025uncertainty, lee2026low}. More broadly, a growing literature develops statistical inference for low-rank models under incomplete or sequential observations
\citep{duan2026online,ma2026statistical,wen2026online}.
Related semiparametric ideas have also begun to appear in LLM evaluation, for example
through the use of auxiliary comparison signals to improve efficiency
\citep{dong2026evaluating}.  Closest to our work,
\cite{li2026llmevaluation} study low-rank LLM evaluation and semiparametric efficiency
for smooth functionals of a latent score tensor. In contrast, our paper focuses on
task-specific ranks and top-\(K\) membership, which are nonsmooth ranking functionals
determined by many dependent score-gap signs. Thus the main new ingredients are
ranking-specific: boundary-based top-\(K\) recovery, multiple testing for correlated
score-gap hypotheses, and simultaneous rank/top-\(K\) certification. These ingredients
allow us to move from efficient score estimation to certified task-specific
leaderboards.

\section{Problem Setup}\label{sec:setup}

We formalize task-specific LLM evaluation as a sparse pairwise-comparison problem.
There are \(\dt\) task categories and \(\dm\) candidate models. For each task-model
pair, let \(\Thetastar_{t,m}\) denote the latent ability score of model \(m\) on task
\(t\), and collect these scores in the matrix
$
    \Thetastar \in \R^{\dt\times\dm}.
$
The row index \(t\in[\dt]\) represents a task category and the column index
\(m\in[\dm]\) represents an LLM. We assume that \(\Thetastar\) has rank
\(r\ll \min\{\dt,\dm\}\) and singular value decomposition
\(
    \Thetastar = U^\star \Sigma^\star (V^\star)^\top .
\)
Let \(\sigma_1^\star\ge \cdots \ge \sigma_r^\star>0\) denote the nonzero singular
values collected in \(\Sigma^\star\), and define the
condition number $\kappa := \sigma_1^\star / \sigma_r^\star$.
Throughout the asymptotic theory, \(r\) is fixed, and we impose the signal-strength
condition
\(
    \|\Thetastar\|_F\ge c_{\rm sig}\sqrt{\dt\dm}
\)
for a constant \(c_{\rm sig}>0\). Together with bounded \(\kappa\), this keeps the
smallest nonzero singular value separated at the scale required for estimating the
rank-\(r\) singular subspaces.
This low-rank structure allows information to be shared across related tasks while preserving
task-specific model rankings. Let \(e_t\in\R^{\dt}\) and
\(e_m\in\R^{\dm}\) denote the standard basis vectors for task \(t\)
and model \(m\), respectively. We assume the singular vectors are \(\mu\)-incoherent:
\[
    \max_{t\in[\dt]}\|e_t^\top U^\star\|_2^2
    \le \frac{\mu r}{\dt},
    \qquad
    \max_{m\in[\dm]}\|e_m^\top V^\star\|_2^2
    \le \frac{\mu r}{\dm}.
\]
Such signal-strength and incoherence assumptions are standard in
low-rank estimation and inference, ensuring stable singular-subspace recovery and
preventing the latent factors from concentrating on a small number of coordinates
under incomplete observations
\citep{candes2009exact,negahban2012restricted,xia2021statistical,
duan2026online}.

\textbf{Observation model.}
For \(i=1,\ldots,n\), a task \(t_i\in[\dt]\) is sampled from a distribution \(\nu\).
Conditional on \(t_i\), a pair of distinct models \((m_i,m_i')\) is sampled from a task-dependent
distribution \(\pi_{t_i}\). We let \(Y_i=1\) indicate that model \(m_i\) is preferred to model \(m_i'\). Define the signed comparison design matrix
\(    X_i := e_{t_i}(e_{m_i}-e_{m_i'})^\top.
\)
It has a single nonzero row, corresponding to task \(t_i\), with \(+1\) in
column \(m_i\) and \(-1\) in column \(m_i'\). Hence
\(\langle X_i,\Thetastar\rangle=\Thetastar_{t_i,m_i}-\Thetastar_{t_i,m_i'}\). Conditional on \(X_i\), the preference follows a BTL model \citep{bradley1952rank,luce1959individual}
\[
    \Pr(Y_i=1\mid X_i)
    =
    \sigma(\langle X_i,\Thetastar\rangle)
    =
    \sigma\!\left(\Thetastar_{t_i,m_i}-\Thetastar_{t_i,m_i'}\right),
    \qquad
    \sigma(x)=(1+e^{-x})^{-1}.
\]

Because pairwise comparisons depend only on score differences, the matrix
\(\Thetastar\) is identifiable only up to task-specific additive shifts. We resolve this non-identifiability by imposing the row-centering constraint
\(
    \Thetastar \mathbf 1_{\dm}=0 .
\)

\textbf{Sampling design.}
We allow the task and model-pair sampling distributions to be non-uniform, reflecting
the uneven traffic patterns of real evaluation platforms. To obtain clean theoretical
rates, we assume this imbalance is controlled: there exist constants
\(0<c_\nu\le C_\nu<\infty\) and \(0<c_\pi\le C_\pi<\infty\), independent of
\(\dt,\dm,n\), such that for all \(t\in[\dt]\) and unordered pairs
\(\{m,m'\}\subset[\dm]\),
\[
    \frac{c_\nu}{\dt}\le \nu_t\le \frac{C_\nu}{\dt},
    \qquad
    \frac{c_\pi}{\binom{\dm}{2}}
    \le \pi_t(\{m,m'\})
    \le
    \frac{C_\pi}{\binom{\dm}{2}} .
\]
Together with the bounded-signal condition \(\|\Thetastar\|_\infty\le B\), the above
sampling assumptions ensure that every task and every model pair receives comparable statistical information. Indeed, writing
$
    \eta=\Thetastar_{t,m}-\Thetastar_{t,m'}
$
for a generic within-task score difference, the bounded-signal condition implies
\(|\eta|\le 2B\). Hence the BTL Fisher information
\(
    I(\eta):=\sigma(\eta)\{1-\sigma(\eta)\}
\)
is bounded away from zero and infinity, as in standard bounded dynamic-range assumptions
for BTL-type ranking models
\citep{chen2015spectral,chen2022partial,fan2024care,fan2025uncertainty}. The lower
bounds on \(\nu_t\) and \(\pi_t(\{m,m'\})\) prevent any task or pairwise comparison
direction from being asymptotically unobserved, in the same spirit as standard sampling conditions in low-rank matrix completion and one-bit matrix estimation
\citep{candes2009exact,negahban2012restricted,davenport20141bit}.

\textbf{Ranking targets.}
Built on the entrywise error bound for the latent score matrix, Section~\ref{sec:estimation}
studies task-wise top-\(K\) accuracy. To ensure that each model has a unique rank even
when two latent scores are exactly equal, we fix an arbitrary deterministic ordering
\(\prec_0\) of the model labels and use it only to break such ties. We then define
%Fix once and for all a deterministic total order \(\prec_0\) on the model labels, used only to break exact score ties, and define
\[
    \rk_t(m)
    :=1+\sum_{m'\ne m}\mathbf 1\{\Thetastar_{t,m'}>\Thetastar_{t,m}\}
      +\sum_{m'\ne m}\mathbf 1\{\Thetastar_{t,m'}=\Thetastar_{t,m},\ m'\prec_0m\}.
\]
The task-specific top-\(K\) set is
\(
    \TopK(t):=\{m:\rk_t(m)\le K\},
\)
and hence has exactly \(K\) elements. Section~\ref{sec:fdim}
develops efficient inference and multiple testing tools for task-specific score gaps
\(\Thetastar_{t,m}-\Thetastar_{t,m'}\), and Section~\ref{sec:ranking} converts these
gap inferences into simultaneous rank confidence sets and top-\(K\) membership
certificates.

%\begin{revisionblock}
\section{Ranking Recovery: Initialization, Refinement, and
\texorpdfstring{Top-\(K\)}{Top-K} Decisions}\label{sec:estimation}

This section establishes the first step toward certified task-specific leaderboards:
task-wise top-\(K\) accuracy from sparse pairwise comparisons.  We give the complete
procedure in Algorithm~\ref{alg:ranking-recovery}.  It separates three statistical
tasks that require different tools: locating the shared low-rank signal, correcting
row- and column-specific errors, and translating score resolution into a ranking
decision.  The output here is a point estimate of the leaderboard.  Sections
\ref{sec:fdim}-\ref{sec:ranking} subsequently attach uncertainty to the same ranking
claims.

\subsection{Global initialization and the need for local refinement}

\textbf{Low-rank score initializer.}
Let \(X_i=e_{t_i}(e_{m_i}-e_{m_i'})^\top\) and
\(\ell(y,\eta)=\log(1+\exp(\eta))-y\eta\). We compute a nuclear-norm penalized BTL
initializer
\[
\resizebox{0.95\textwidth}{!}{$ 
\widehat\Theta_0\in
\arg\min_{\Theta\in\mathcal C_B}
\left\{
\frac1{|I_0|}\sum_{i\in I_0}
\ell(Y_i,\langle X_i,\Theta\rangle)+\lambda\|\Theta\|_*
\right\},
\quad
\mathcal C_B=\{\Theta:\Theta\mathbf 1_{\dm}=0,\ \|\Theta\|_\infty\le B\}.$}
\]
This convex optimization provides a Frobenius-accurate initializer. Theorem~\ref{thm:convex-main-app} in Appendix shows that, with high probability,
$\|\widehat\Theta_0-\Thetastar\|_F    \lesssim    \sqrt{{r\,\dt\dm\bar d\,\mathrm{polylog}(n_0\bar d)}/{n_0}},
$ where $\bar d:=\max\{\dt,\dm\}$ and $n_0 = |I_0|$.
The nuclear norm shares information across tasks and the centering constraint fixes the
BTL shift indeterminacy. However, Frobenius control is an average guarantee over all \(\dt\dm\) entries and does not rule out substantially larger errors in a small number of model rows or task columns. This is insufficient for ranking, where a single poorly estimated score near the \(K\)-versus-\((K+1)\) boundary can alter the top-\(K\) set. 

\textbf{Local refinement.}
To obtain the entrywise accuracy required for ranking, we construct
\[\Thetahat=\mathsf{Refine}_r(\widehat\Theta_0)\]from the Frobenius-accurate initializer $\widehat\Theta_0$ followed by a three-split row-wise refinement. Specifically, we first split the comparison data into three independent, comparable subsets,
\(
\mathcal D_1,\mathcal D_2,\mathcal D_3,
\)
used respectively for low-rank initialization, model-side refinement, and
task-side refinement. For the refinement, it is convenient to work with the
transposed score matrix
\[
M^\star=(\Thetastar)^\top
       =L^\star(A^\star)^\top\in\R^{\dm\times\dt},
\qquad
(A^\star)^\top A^\star=I_r,
\qquad
\mathbf 1_{\dm}^\top L^\star=0,
\]
where model-specific latent vectors form the rows of \(L^\star\), task-specific
latent vectors form the rows of \(A^\star\), and the singular values are absorbed
into \(L^\star\). This factorization suggests a global-to-local refinement strategy. Starting from the nuclear-norm initializer on \(\mathcal D_1\) (Algorithm~\ref{alg:ranking-recovery}, line~2), we first estimate the shared task-side factor from the leading right singular vectors of \(\widehat M_0=\widehat\Theta_0^\top\) and project it onto an incoherent Stiefel set (line~3)
\begin{equation}
\mathcal A_\mu = \left\{A\in\mathbb R^{\dt\times r}:A^\top A=I_r, \ \max_{t\le\dt}\|A[t]\|_2 \le C_{\rm inc}\sqrt{\frac{\mu r}{\dt}}\right\}.
\label{eqn:stiefelset}
\end{equation}
Using the independent split
\(\mathcal D_2\), we then refine each model-side latent vector through a
low-dimensional pairwise-logistic score equation, with the initializer supplying
the opponent offset (lines~4--6). We subsequently recenter the resulting model factor using
$
R_{\rm m}
=
I_{\dm}-\dm^{-1}\mathbf1_{\dm}\mathbf1_{\dm}^\top,
$
thereby restoring the BTL identification constraint (line~7). Here when a comparison involves model \(u\), we reorient it so that \(u\) appears first. Write \(Z_{iu}=Y_i\) if \(u=m_i\), and \(Z_{iu}=1-Y_i\) if \(u=m_i'\);
let \(w_{iu}\) denote the opponent.  This convention uses every comparison involving
\(u\) in a common logistic score equation.

On the remaining split
\(\mathcal D_3\), we refine each task-side latent vector conditional on the centered
model factor (lines~8--10). Finally, we reconstruct the score matrix 
\(
\Thetahat=(\overline L\widetilde A^\top)^\top
\) (line~11)
and rank the models
within each task to obtain the estimated ranks and top-\(K\) sets
(lines~12--13). Thus, the refinement uses the initializer to recover the shared
low-rank structure and then performs local latent-coordinate updates to achieve the
entrywise accuracy needed for task-specific ranking. The resulting procedure is summarized in Algorithm~\ref{alg:ranking-recovery}, and its entrywise error guarantee is proved in Appendix~\ref{app:entrywise-proof}.

\begin{algorithm}[h!]
\caption{Ranking Recovery by Low-Rank Initialization and Pairwise Refinement}
\label{alg:ranking-recovery}
\begin{algorithmic}[1]
%\revisioncolor
\Require Comparisons \(\mathcal D=\{(t_i,m_i,m_i',Y_i)\}_{i=1}^n\), rank \(r\),
top-set size \(K\), bound \(B\), penalty \(\lambda\), and tie order \(\prec_0\).
\Ensure Refined matrix \(\Thetahat\), ranks \(\widehat{\rk}_t(m)\), and sets
\(\widehat{\TopK}(t)\).
\State Split \(\mathcal D\) into independent comparable parts
\(\mathcal D_1,\mathcal D_2,\mathcal D_3\).
\State On \(\mathcal D_1\), solve the nuclear-norm program above; set
\(\widehat M_0=\widehat\Theta_0^\top\).
\State Let \(Q_0\) contain the top-\(r\) right singular vectors of
\(\widehat M_0\); set \(\widehat A=\operatorname{Proj}_{\mathcal A_\mu}(Q_0)\).
\For{each model \(u\in[\dm]\)}
  \State Let \(\mathcal Q_u=\{i\in\mathcal D_2:u\in\{m_i,m_i'\}\}\) and
  reorient its observations so that \(u\) is first.
  \State Solve locally for \(\ell\in\mathbb R^r\):
  \[
  \sum_{i\in\mathcal Q_u}\widehat A[t_i]
  \left[Z_{iu}-\sigma\!\left\{
  \widehat A[t_i]^\top\ell-\widehat M_{0,w_{iu},t_i}
  \right\}\right]=0,
  \]
  and store the solution as row \(\widetilde L[u]\).
\EndFor
\State Restore the identification gauge by setting \(\overline L=R_{\rm m}\widetilde L\).
\For{each task \(t\in[\dt]\)}
  \State For \(i\in\mathcal I_t:=\{i\in\mathcal D_3:t_i=t\}\), set
  \(x_i=\overline L[m_i]-\overline L[m_i']\).
  \State Solve locally
  \(\sum_{i\in\mathcal I_t}x_i\{Y_i-\sigma(x_i^\top a)\}=0\),
  and store the solution as row \(\widetilde A[t]\).
\EndFor
\State Set \(\Thetahat=(\overline L\widetilde A^\top)^\top\).
\For{each task \(t\)}
  \State Sort \(\{\Thetahat_{t,m}:m\in[\dm]\}\), breaking ties by \(\prec_0\),
  and return \(\widehat{\rk}_t(m)\) and
  \(\widehat{\TopK}(t)=\{m:\widehat{\rk}_t(m)\le K\}\).
\EndFor
\end{algorithmic}
\end{algorithm}

\subsection{Refinement intuition and contribution}

\paragraph{The initializer estimates geometry; the row equations estimate local
coordinates.}
The rank-\(r\) singular vectors of \(\widehat M_0\) identify a common task-feature
coordinate system.  Projecting them onto \(\mathcal A_\mu\) prevents a small number of
tasks from acquiring excessive leverage while preserving the estimated subspace.  Once
this geometry is fixed, the unknown coefficient for any model has only \(r\) coordinates.
The model-row update is therefore a collection of low-dimensional likelihood equations,
not another ambient matrix optimization.  This dimension reduction is what makes a
uniform row-wise analysis possible.

\paragraph{The opponent offset removes first-order confounding from opponent quality.} A comparison involving model \(u\) depends on both \(u\)'s score and its opponent's
score.  The term \(\widehat M_{0,w_{iu},t_i}\) carries the estimated opponent ability
into the row equation.  Without this offset, variation in opponent quality would be
mistaken for variation in \(u\)'s latent coefficient.  Reorienting comparisons allows
the same equation to use observations in which \(u\) appeared on either side, and the
independent split makes the plug-in offset conditionally fixed in the row-wise
concentration argument.

\paragraph{Centering aligns the pairwise information geometry.}
BTL probabilities depend only on score differences, so the model-side factor is
identified only after fixing a centering constraint. Multiplication by
\(R_{\rm m}\) restores \(\mathbf1_{\dm}^\top\overline L=0\) while leaving every
pairwise difference \(\overline L[m]-\overline L[m']\) unchanged. This centering
also has an important statistical consequence. Under uniform pair sampling, the
population Gram matrix of the pairwise differences
\(
(\overline L[m]-\overline L[m'])
(\overline L[m]-\overline L[m'])^\top
\)
is proportional to the centered factor Gram matrix
\(\overline L^\top\overline L\). Under the near-uniform sampling design assumed
here, the population pairwise Gram matrix remains uniformly comparable to
\(\overline L^\top\overline L\), and its empirical counterpart concentrates around
this population quantity. Together with the bounded-score condition, which keeps
the logistic curvature weights bounded away from zero, this yields uniform
curvature for the subsequent task-side refinement.

\paragraph{The second update balances both sides of the factorization.}
After all model rows have been refined in a common coordinate system, data from a fixed
task form an ordinary \(r\)-dimensional logistic regression with feature
\(x_i=\overline L[m_i]-\overline L[m_i']\).  Solving the task equation removes the
remaining first-order subspace error.  The product
\((\overline L\widetilde A^\top)^\top\) is automatically low rank and exactly
row-centered. Uniform rowwise \(2,\infty\) control of both the model factor
\(\overline L\) and the task factor \(\widetilde A\) then combines through the
bilinear product decomposition to yield a max-norm bound for \(\Thetahat-\Thetastar\).

%The two rowwise \(2,\infty\) bounds then combine through an entrywise product decomposition to yield max-norm accuracy.

\begin{theorem}[Uniform entrywise estimation]\label{thm:entrywise}
Under the model and near-uniform sampling assumptions in
Section~\ref{sec:setup}, suppose \(r\) is fixed, \(\Thetastar\) has rank \(r\), is
\(\mu\)-incoherent, satisfies \(\|\Thetastar\|_\infty\le B\) and
\(\|\Thetastar\|_F\ge c_{\rm sig}\sqrt{\dt\dm}\), and has condition
number \(\kappa\). Denote $\bar d:=\max\{\dt,\dm\}$. If the pairwise comparison sample size
\(
    n\gtrsim \mathrm{poly}(\mu,r,\kappa,B)\,\bar d\,\log^c(n\bar d),
\)
then, for some large constant $a$, with probability at least \(1-n^{-a}\),
\[
  \resizebox{0.68\textwidth}{!}{$   \|\Thetahat-\Thetastar\|_\infty
    \le
    \varepsilon_n,
    \textrm{~~with~~}
    \varepsilon_n
    :=
    C\,\mathrm{poly}(a,\mu,r,\kappa,B)
    \sqrt{\frac{\bar d\,\log^c(n\bar d)}{n}}.$}
\]
\end{theorem}
Theorem~\ref{thm:entrywise} is the main estimation result.  The nuclear-norm optimization provides a Frobenius-accurate
initializer, while the refinement step upgrades this global accuracy to uniform
max-norm control over all task-model scores.  For constant rank and condition number, the required sample size is
near-linear in \(\bar d\), so the method learns all task-specific scores jointly rather
than fitting \(d_t\) unrelated BTL models. In the balanced regime \(\dt\asymp\dm\asymp d\), this is
$
    \|\Thetahat-\Thetastar\|_\infty
    \lesssim
    \sqrt{d\,\mathrm{polylog}(nd)/n} .
$

%\paragraph{What is new relative to generic low-rank estimation.}
While the rate in Theorem~\ref{thm:entrywise} is consistent with the near-linear
sample complexity familiar from low-rank matrix completion \citep{xia2021statistical}, obtaining uniform
entrywise accuracy in the present pairwise-comparison model requires substantially
more than a standard Frobenius-norm oracle inequality.
The refinement is not a cosmetic post-processing of a nuclear-norm estimator. Its proof
combines four ingredients tailored to ranking from comparisons: spectral alignment of
the shared task geometry; local existence and stability of nonlinear row score roots;
uniform rowwise coverage and concentration; and an exact centered pairwise Gram identity
for the task update. A Frobenius oracle inequality alone contains none of these local
guarantees. The contribution is precisely the conversion from average low-rank accuracy
to the entrywise resolution required by rank and top-\(K\) decisions. See Appendix
\ref{app:entrywise-proof} for detailed proofs.

% \paragraph{What is new relative to generic low-rank estimation.}
% The refinement is not a cosmetic post-processing of a nuclear-norm estimator.  Its proof
% combines four ingredients tailored to ranking from comparisons: spectral alignment of
% the shared task geometry; local existence and stability of nonlinear row score roots;
% uniform rowwise coverage and concentration; and an exact centered pairwise Gram identity
% for the task update.  A Frobenius oracle inequality alone contains none of these local
% guarantees.  The contribution is precisely the conversion from average low-rank accuracy
% to the entrywise resolution required by rank and top-\(K\) decisions.  Appendix
% \ref{app:entrywise-proof} proves this chain directly under the assumptions of the present
% paper.

\subsection{From max-norm resolution to \texorpdfstring{top-\(K\)}{top-K} recovery}

We next translate the entrywise error bound into a task-specific top-\(K\) recovery
guarantee. Denote the estimated top-\(K\) set for task \(t\) as
\(
    \widehat{\TopK}(t)=\{m:\widehat{\rk}_t(m)\le K\},
\)
where the rank \(\widehat{\rk}_t(m)\)
is computed from the refined estimate \(\Thetahat_{t,m}\) using the same
deterministic tie-breaking order \(\prec_0\). To measure the
discrepancy between the estimated and true top-\(K\) sets, define the normalized Hamming
error
\[
    \mathsf{Ham}_{K,t}
    :=
    \frac{1}{2K}
    |\widehat{\TopK}(t)\triangle \TopK(t)|.
\]
Here \(\triangle\) denotes symmetric difference: it counts models that are included in
one top-\(K\) set but not the other. The normalization by \(2K\) makes
\(\mathsf{Ham}_{K,t}\in[0,1]\).

The difficulty of top-\(K\) recovery depends on how many models have scores close to
the top-\(K\) cutoff. Let
$
    \Thetastar_{t,(1)}\ge\cdots\ge\Thetastar_{t,(\dm)}
$
be the sorted true scores for task \(t\), and define the midpoint between the \(K\)-th
and \((K+1)\)-st scores as
$
    \tau_K(t)
    :=
    (\Thetastar_{t,(K)}+\Thetastar_{t,(K+1)})/2.
$
We call \(\tau_K(t)\) the top-\(K\) decision boundary. For a resolution level
\(\delta>0\), define the boundary profile
\[
    \mathcal R_{K,t}(\delta;\Thetastar)
    :=
    \frac{1}{2K}
    \bigl|\{m:|\Thetastar_{t,m}-\tau_K(t)|\le\delta\}\bigr|.
\]
This quantity measures the fraction of models lying within distance \(\delta\) of the
top-\(K\) boundary. If many models lie near \(\tau_K(t)\), then the task is intrinsically
hard to rank because small estimation errors can swap models across the top-\(K\)
cutoff. If few models lie near the boundary, top-\(K\) recovery is easier.

\begin{proposition}[Task-wise top-\(K\) Hamming accuracy]\label{prop:topk-hamming}
On the event \(\|\Thetahat-\Thetastar\|_\infty\le\varepsilon_n\), for every
\(t\in[\dt]\),
$
    \mathsf{Ham}_{K,t}
    \le
    \mathcal R_{K,t}(2\varepsilon_n;\Thetastar).
$
Therefore, under Theorem~\ref{thm:entrywise}, the above Hamming bound holds
simultaneously for all tasks with probability at least \(1-n^{-a}\).
\end{proposition}

Proposition~\ref{prop:topk-hamming} shows that top-\(K\) mistakes can only occur for
models whose true scores lie within the estimation resolution \(2\varepsilon_n\) of
the top-\(K\) boundary. This extends the boundary-resolution principle from single-task
BTL top-\(K\) ranking \citep{chen2015spectral,chen2022partial} to many task-specific
rankings coupled through shared low-rank structure. Consequently, exact recovery follows under a task-specific margin condition. Define the \(K\)-gap
$
    \Delta_K(t):=\Thetastar_{t,(K)}-\Thetastar_{t,(K+1)} .
$
If \(\Delta_K(t)>2\varepsilon_n\), then no model can cross the top-\(K\) boundary, so
\(
    \widehat{\TopK}(t)=\TopK(t).
\)
Thus, with high probability, exact top-\(K\) recovery holds
simultaneously for every task whose \(K\)-gap exceeds \(2\varepsilon_n\). Hence, low-rank sharing enables simultaneous task-wise top-\(K\) recovery at the same
entrywise resolution as score estimation, without requiring each task to be estimated
independently.

Two interpretations are useful.  First, the boundary profile gives a graded guarantee:
even without a global minimum gap, all errors are localized to models that are
within the estimation resolution of the cutoff.  Second, the exact-recovery condition
is task adaptive.  Easy tasks with a large \(K\)-gap can be recovered even when another
task has a dense cluster of nearly tied models.  Algorithm~\ref{alg:ranking-recovery}
therefore shares estimation strength globally without forcing all task-specific
leaderboards to have the same difficulty.  When a boundary gap is below \(2\varepsilon_n\), a point top-\(K\) set should not be interpreted as certified; the inferential procedures developed in later sections instead provide
inner, outer, and unresolved top-\(K\) sets.

\section{Efficient Score-Gap Inference}
\label{sec:fdim}

Task-specific ranking decisions are built from score-gap signs. For example, deciding
whether model \(m\) outranks model \(m'\) on task \(t\) requires inference on
$
    \Theta^\star_{t,m}-\Theta^\star_{t,m'} .
$
Similarly, ranks and top-\(K\) membership are determined by many such pairwise gaps.
We therefore
first develop efficient inference for score gaps and characterize their joint dependence,
which provides the multiple-testing foundations used later for simultaneous ranking
inference.

\subsection{Score-gap targets and the efficient direction}

We write a generic linear contrast as
\(
    \psi_\Gamma(\Thetastar)=\ip{\Gamma}{\Thetastar}.
\)
The canonical example is
\(\Gamma=e_t(e_m-e_{m'})^\top\), for which
$
    \psi_\Gamma(\Thetastar)=\Thetastar_{t,m}-\Thetastar_{t,m'} .
$
Testing whether model \(m\) is better than model \(m'\) on task \(t\) is therefore a
test on the sign of \(\psi_\Gamma(\Thetastar)\). Multiple ranking claims correspond to
testing many such contrasts jointly.

\paragraph{Low-rank tangent geometry.}
In the low-rank model, local perturbations of \(\Thetastar\) must lie in the tangent
space \(\mathbb T\) of the exact-rank-\(r\), row-centered manifold. Since
\((V^\star)^\top\mathbf 1_{\dm}=0\), this tangent space and its orthogonal
projector are
\[
    \mathbb T
    :=
    \left\{
      U^\star C^\top+D(V^\star)^\top:
      C^\top\mathbf 1_{\dm}=0,\ 
      C\in\R^{\dm\times r},\ D\in\R^{\dt\times r}
    \right\},
\]
\[
    P_{\mathbb T}(H)
    =
    P_U H R_{\rm m}+(I_{\dt}-P_U)H P_V,
    \qquad
    R_{\rm m}:=I_{\dm}-\frac{\mathbf 1_{\dm}\mathbf 1_{\dm}^\top}{\dm},
\]
where \(P_U=U^\star(U^\star)^\top\) and
\(P_V=V^\star(V^\star)^\top\). In particular,
\(P_{\mathbb T}(H)\mathbf 1_{\dm}=0\). Therefore only the
projected contrast direction \(P_{\mathbb T}\Gamma\) is locally identifiable. Define
the Fisher operator \(G\) by
\[
    \ip{G H_1}{H_2}
    =
    \E\!\left[
        I(\ip{X}{\Thetastar})\ip{H_1}{X}\ip{H_2}{X}
    \right],
    \qquad
    I(\eta)=\sigma(\eta)\{1-\sigma(\eta)\}.
\]
It measures how much information the pairwise-comparison design carries about
directions \(H_1,H_2\). Closely matched comparisons have larger \(I(\eta)\), while
lopsided comparisons carry less information.

Define the restricted Fisher information operator \(A\), the efficient direction
\(H_\Gamma^\star\), and the corresponding efficient variance \(V_{\rm eff}(\Gamma)\) by
\[
    A:=\left.P_{\mathbb T}GP_{\mathbb T}\right|_{\mathbb T}:
       \mathbb T\to\mathbb T,
    \qquad
    H_\Gamma^\star:=A^{-1}P_{\mathbb T}\Gamma,
    \qquad
    V_{\rm eff}(\Gamma)
    :=
    \ip{P_{\mathbb T}\Gamma}{A^{-1}P_{\mathbb T}\Gamma}.
\]
Here and below \(A^{-1}\) means the inverse only on \(\mathbb T\); whenever it is
viewed as an ambient operator, it means the extension
\(
    A^\dagger:=P_{\mathbb T}A^{-1}P_{\mathbb T}.
\)
The near-uniform design and bounded Fisher weights make \(A\) positive definite on
\(\mathbb T\), even though \(G\) is singular on the unidentifiable row-shift
directions in the ambient space.
The operator \(A\) describes how much information the observed pairwise comparisons
carry about locally admissible low-rank perturbations of \(\Thetastar\). The efficient
direction \(H_\Gamma^\star\) is the optimal weighting direction for converting
comparison residuals into an estimate of the target contrast \(\psi_\Gamma\). The
variance \(V_{\rm eff}(\Gamma)\) is the resulting semiparametric efficiency bound: it
is the smallest achievable asymptotic variance for regular estimators of
\(\psi_\Gamma\), accounting for the sampling design, the BTL Fisher information, and
the low-rank constraint.

The equation \(A H_\Gamma^\star=P_{\mathbb T}\Gamma\) has the same interpretation as a
weighted least-squares normal equation.  The projected contrast is the local target
direction, while \(A\) measures which tangent perturbations are visible through the
comparison design.  Applying \(A^{-1}\) produces the residual weight that estimates the
target with minimum variance.  Directions outside \(\mathbb T\) are not admissible local
changes of a rank-\(r\) matrix, and the task-specific row-shift directions are not identifiable;
both are removed before inversion.  This is why an ambient inverse Fisher matrix would
be either singular or inefficient.

For the ranking contrasts used below, let
\(\mathcal F_{\rm gap}:=\{e_t(e_m-e_{m'})^\top:t\in[\dt],\ m\ne m'\}\).
We impose the uniform alignment condition
\[
    \inf_{\Gamma\in\mathcal F_{\rm gap}}
    \frac{\|P_{\mathbb T}\Gamma\|_F}
         {\sqrt{\bar d/(\dt\dm)}\,\|\Gamma\|_F}
    \ge \alpha_0>0.
\]
This is a separate structural condition ensuring that no target score gap is
asymptotically orthogonal to the tangent space, so that inference is possible.

\subsection{Cross-fitted efficient score-gap algorithm}

%\begin{revisionblock}
We implement the one-step estimator in Algorithm~\ref{alg:efficient-gaps}
by fixed-\(K_{\rm cf}\) cross-fitting. Partition
\([n]\) into folds \(I_1,\ldots,I_{K_{\rm cf}}\), where \(K_{\rm cf}\ge3\) is fixed and
\(n_k:=|I_k|\asymp n\). For each \(k\), use the complement of \(I_k\)
to construct the exact-rank-\(r\), row-centered estimator
\(\widehat\Theta^{(-k)}\) and its tangent projector
\(\widehat P_{\mathbb T}^{(-k)}\).  %As is standard for a randomized comparison design, the task and pair sampling law \((\nu,\pi)\) is known.
For the inference theory below, we consider a randomized comparison design in which
the task and pair sampling law \((\nu,\pi)\) is specified by the randomization protocol
and hence known by design
\citep{imbens2015causal,chen2015spectral}.
Define the design-averaged plug-in Fisher operator
\[
    \ip{\widehat G_k H_1}{H_2}
    :=
    \E_X\!\left[
    I\!\left(\ip{X}{\widehat\Theta^{(-k)}}\right)
    \ip{H_1}{X}\ip{H_2}{X}\right],
    \qquad
    X\sim(\nu,\pi).
\]
This finite expectation can be evaluated by summing over the task-pair
support. %It is not an empirical covariance computed from the evaluation fold.
Section~\ref{app:notation-CA} shows that this design-averaged construction allows the plug-in inverse bounds to be derived from the population \(C_A\) condition.
An unknown observational design can be accommodated by estimating the design law from an auxiliary sample and plugging it into the Fisher operator, with an additional design-estimation error term in the analysis \citep{newey1994asymptotic,chernozhukov2018double}.
%An unknown observational design may instead be estimated on an external design-only sample, but then its relative operator error must be included explicitly.
%\end{revisionblock}
Then define
\[
    \widehat A_k
    :=
    \left.
    \widehat P_{\mathbb T}^{(-k)}
    \widehat G_k
    \widehat P_{\mathbb T}^{(-k)}
    \right|_{\widehat{\mathbb T}^{(-k)}} ,
    \qquad
    \widehat H_{\Gamma,k}
    :=
    \widehat A_k^{-1}
    \widehat P_{\mathbb T}^{(-k)}\Gamma .
\]
All inverses in this display are taken only on the indicated estimated tangent
space. With \(s(y,\eta)=y-\sigma(\eta)\), set
\[
    \widehat\psi_\Gamma^{(k)}
    :=
    \ip{\Gamma}{\widehat\Theta^{(-k)}}
    +
    \frac1{n_k}\sum_{i\in I_k}
    s\!\left(Y_i,\widehat\eta_i^{(-k)}\right)
    \ip{\widehat H_{\Gamma,k}}{X_i},
    \qquad
    \widehat\psi_\Gamma
    :=
    \sum_{k=1}^{K_{\rm cf}}\frac{n_k}{n}
    \widehat\psi_\Gamma^{(k)} .
\]

\begin{algorithm}[h!]
\caption{Cross-Fitted Efficient Score-Gap Estimation}
\label{alg:efficient-gaps}
\begin{algorithmic}[1]
%\revisioncolor
\Require Data \(\mathcal D\), rank \(r\), contrast family
\(\mathcal G=\{\Gamma_1,\ldots,\Gamma_q\}\), and folds
\(I_1,\ldots,I_{K_{\rm cf}}\).
\Ensure One-step estimates \(\widehat\psi\), influence matrix \(\widehat\Phi\),
covariance \(\widehat\Sigma\), and marginal standard errors.
\For{each fold \(k=1,\ldots,K_{\rm cf}\)}
  \State Apply Algorithm~\ref{alg:ranking-recovery} on \(I_k^c\) to obtain
  \(\widehat\Theta^{(-k)}\), its tangent space
  \(\widehat{\mathbb T}^{(-k)}\), and projector
  \(\widehat P_{\mathbb T}^{(-k)}\).
  \State Average \(I(\langle X,\widehat\Theta^{(-k)}\rangle)X\otimes X\)
  to form \(\widehat G_k\), and set
  \(\widehat A_k=\widehat P_{\mathbb T}^{(-k)}\widehat G_k
  \widehat P_{\mathbb T}^{(-k)}|_{\widehat{\mathbb T}^{(-k)}}.\)
  \For{each target \(j\in[q]\)}
    \State Solve
    \(\widehat A_k\widehat H_{j,k}
    =\widehat P_{\mathbb T}^{(-k)}\Gamma_j\) on
    \(\widehat{\mathbb T}^{(-k)}\).
    \For{each \(i\in I_k\)}
      \State Set \(\widehat\eta_i^{(-k)}=\langle X_i,\widehat\Theta^{(-k)}\rangle\) and
      \(\widehat\phi_{j,i}=\{Y_i-\sigma(\widehat\eta_i^{(-k)})\}
      \langle\widehat H_{j,k},X_i\rangle\).
    \EndFor
    \State Set
    \(\widehat\psi_j^{(k)}=\langle\Gamma_j,\widehat\Theta^{(-k)}\rangle
    +n_k^{-1}\sum_{i\in I_k}\widehat\phi_{j,i}\).
  \EndFor
\EndFor
\State Aggregate \(\widehat\psi_j=\sum_k(n_k/n)\widehat\psi_j^{(k)}\) and
center \(\widehat\Phi_{ij}=\widehat\phi_{j,i}-\overline\phi_j\).
\State Compute
\(\widehat\Sigma=n^{-1}\widehat\Phi^\top\widehat\Phi\) and
\(\widehat{\operatorname{SE}}_j=\sqrt{\widehat\Sigma_{jj}/n}\).
\State \Return \(\widehat\psi,\widehat\Phi,\widehat\Sigma\), and
\(\{\widehat{\operatorname{SE}}_j\}_{j=1}^q\).
\end{algorithmic}
\end{algorithm}

%\begin{revisionblock}
The score-matrix estimate, tangent space, and information direction are therefore
independent of the evaluation fold \(I_k\).
Conditional on the training sample, observations in \(I_k\) remain i.i.d.,
and the oracle scores \(s(Y_i,\eta_i^\star)\) are centered.  The plug-in scores
\(s(Y_i,\widehat\eta_i^{(-k)})\) need not be centered and are handled by the exact
population Taylor decomposition in Appendix~\ref{app:fdim-proof-decomp}.
%\end{revisionblock}

\paragraph{Why the one-step correction works.}
The plug-in contrast
\(\langle\Gamma,\widehat\Theta^{(-k)}\rangle\) inherits a first-order error from
estimating a high-dimensional low-rank matrix.  Expanding the conditional mean of the
score around \(\Thetastar\) produces the opposite first-order term
\(-G(\widehat\Theta^{(-k)}-\Thetastar)\).  Because
\(\widehat H_{\Gamma,k}\) solves the restricted information equation, its inner product
with this score drift cancels the tangent component of the plug-in error.  Exact-rank
geometry makes the remaining normal component second order.  Cross-fitting prevents the
same observations from choosing the nuisance direction and evaluating its score, so the
leading term becomes an average of conditionally centered influence values.

This cancellation is the bridge between recovery and inference.  Max-norm refinement
makes the Taylor remainder uniformly small over many localized score-gap targets, but
the first-order distribution no longer depends on the convergence rate of every nuisance
coordinate.  Instead, it is governed by the efficient influence function
\[
\phi_\Gamma(W)
=
\{Y-\sigma(\langle X,\Thetastar\rangle)\}
\langle H_\Gamma^\star,X\rangle.
\]
The construction is thus both debiased and design adaptive: comparisons that are more
informative for a target receive weights determined by the inverse restricted Fisher
operator.

%\begin{revisionblock}
Finally, we impose inverse-information stability \emph{on the population
restricted Fisher operator}.  Namely, for a structural constant \(C_A\ge1\),
\[
    \frac{\|A^\dagger\|_{\infty\to\infty}}{\dt\dm}\le C_A.
\]
Proposition~\ref{prop:CA-range-app} shows
\(
    c\ \le\ C_A\ \le\ C\sqrt{\bar d},
\)
and, under uniform task and pair sampling, \(C_A=O(1)\) for constant or near-constant Fisher weights.  The
inferential rates require only \(C_A r_n\to0\), where
\(r_n=\sqrt{\bar d\log^c(n\bar d)/n}\).
%\end{revisionblock}

\subsection{Joint covariance, efficiency, and multiple hypotheses}

For ranking, we need joint inference for many score gaps. For a fixed
collection \(\Gamma_1,\ldots,\Gamma_q\), let
\(\psi_j=\psi_{\Gamma_j}(\Thetastar)\) and define the population efficient
covariance matrix \(\Sigma_n\in\mathbb R^{q\times q}\) by
\[
    (\Sigma_n)_{jk}
    :=
    \ip{P_{\mathbb T}\Gamma_j}
       {A^{-1}P_{\mathbb T}\Gamma_k}.
\]
% The efficient covariance between the
% corresponding one-step estimators is
% \[
%     \Sigma_{jk}
%     =
%     \ip{P_{\mathbb T}\Gamma_j}{A^{-1}P_{\mathbb T}\Gamma_k}.
% \]
This covariance is generally non-diagonal because two score gaps may share a task, a
model, observed comparisons, or low-rank latent factors. Capturing this dependence is
essential for multiple testing: treating correlated score-gap tests as independent can
miscalibrate leaderboard-level uncertainty.

There are two equivalent computational views of this dependence.  Let
\(\widehat\Phi\in\mathbb R^{n\times q}\) be the centered influence matrix returned by
Algorithm~\ref{alg:efficient-gaps}.  One may explicitly form
\(\widehat\Sigma=n^{-1}\widehat\Phi^\top\widehat\Phi\), estimate the correlation
matrix, and simulate a Gaussian vector with that covariance.  Alternatively, one may
draw a single multiplier vector \(\xi\in\mathbb R^n\) and compute
\(n^{-1/2}\widehat\Phi^\top\xi\) directly.  Conditional on the observations,
\begin{equation}\label{eq:multiplier-covariance-identity}
\operatorname{Cov}_\xi\!\left(
n^{-1/2}\widehat\Phi^\top\xi\mid\mathcal D
\right)
=
\frac1n\widehat\Phi^\top\widehat\Phi
=\widehat\Sigma.
\end{equation}
Thus the multiplier bootstrap does not bypass covariance estimation; it realizes the
same estimated joint law without requiring a separate factorization of a potentially
large or singular \(q\times q\) matrix.  Explicit formation of \(\widehat\Sigma\) is
valuable for reporting correlations and checking efficiency.  Direct multiplier
products are often preferable when the hypothesis family is large because each draw
requires only a matrix--vector product and no covariance inversion.

\begin{theorem}[Efficient joint score-gap inference]\label{thm:fdim-clt}
Under the assumptions of
Theorem~\ref{thm:entrywise},
{the uniform alignment condition and the population
inverse-information condition above}, and the score and design regularity conditions
detailed in Appendix~\ref{app:fdim-proof}, there exists
\(Z_{\Gamma,n}\sim N(0,\Sigma_n)\) such that
\[
   \resizebox{0.95\textwidth}{!}{$  \sup_{B\in\mathcal R_q}
    \left|
    \Pr\!\left\{
        \sqrt n(\widehat\psi_1-\psi_1,\ldots,\widehat\psi_q-\psi_q)\in B
    \right\}
    -
    \Pr\{Z_{\Gamma,n}\in B\}
    \right|
    \lesssim
    C_A
    \sqrt{\frac{\bar d\,\log^c(n\bar d)}{n}}
    +
    \left\{\frac{\bar d\,\log^7(nq)}{n}\right\}^{1/6} .$}
\]
%\begin{revisionblock}
Here \(\mathcal R_q\) is the class of rectangles in \(\R^q\), and \(C_A\) is the
population inverse-information stability factor.
%\end{revisionblock}
Consequently, if
\(C_A\sqrt{\bar d\,\log^c(n\bar d)/n}\to0\), the displayed
triangular-array Gaussian approximation error tends to zero, with
\(Z_{\Gamma,n}\sim N(0,\Sigma_n)\) at sample size \(n\).
Indeed, Assumption~\ref{ass:sample-app} fixes \(c\ge c_0\ge7\), so the
displayed condition also makes the oracle CCK term
\(\{\bar d\log^7(nq)/n\}^{1/6}\) vanish for fixed \(q\).
Moreover, let \(k(i)\) denote the fold containing observation \(i\), and define
\[
    \widehat\phi_{j,i}
    :=
    s\!\left(Y_i,\widehat\eta_i^{(-k(i))}\right)
    \ip{\widehat H_{\Gamma_j,k(i)}}{X_i},
    \qquad
    \overline\phi_j:=\frac1n\sum_{i=1}^n\widehat\phi_{j,i}.
\]
The following centered empirical influence-function covariance is consistent for
\(\Sigma_{jk}\):
\[
    \widehat\Sigma_{jk}
    =
    n^{-1}\sum_{i=1}^n
    (\widehat\phi_{j,i}-\overline\phi_j)
    (\widehat\phi_{k,i}-\overline\phi_k).
\]
\end{theorem}

Theorem~\ref{thm:fdim-clt} yields confidence intervals and joint tests for fixed
collections of score-gap hypotheses. For a single contrast \(\Gamma_j\), the estimated
standard error is
\(
\widehat{\operatorname{SE}}_j
:=
\sqrt{\frac{\widehat\Sigma{jj}}{n}},
\)
so a pointwise \((1-\alpha)\)-confidence interval for
\(\psi_j=\psi_{\Gamma_j}(\Thetastar)\) is
\(
\widehat\psi_j
\pm
z_{1-\alpha/2}\widehat{\operatorname{SE}}_j .
\)
% For a single gap, with
% \(\widehat{\operatorname{SE}}_\Gamma=(\widehat V_\Gamma/n)^{1/2}\), a pointwise
% \((1-\alpha)\)-confidence interval is
% \(
%     \widehat\psi_\Gamma\pm z_{1-\alpha/2}\widehat{\operatorname{SE}}_\Gamma .
% \)
For multiple gaps, \(\widehat\Sigma\) captures their dependence, which is crucial for
valid joint testing of overlapping ranking claims. The covariance \(\Sigma\) is also the
semiparametric efficiency lower bound: any regular asymptotically Gaussian estimator
has limiting covariance no smaller than \(\Sigma\) in the positive-semidefinite order,
and the proposed one-step estimator attains this bound. 

The positive-semidefinite ordering is important.  Efficiency is not merely a collection
of optimal marginal variances: it states that every fixed linear combination of the
reported gaps has the smallest regular asymptotic variance.  Consequently the same
influence representation that gives efficient confidence intervals also supplies the
correct off-diagonal covariance for joint testing.  Redundant gap systems may make
\(\Sigma\) singular---for example, three pairwise gaps on the same task satisfy an exact
linear identity.  Neither Algorithm~\ref{alg:efficient-gaps} nor the multiplier
calibration requires \(\Sigma^{-1}\), so this redundancy causes no conceptual problem.

This part contributes a ranking-specific inferential interface.  The target
family consists of sparse score-gap directions, the nuisance lies on a centered low-rank
manifold, and the relevant inverse information operator is restricted to its tangent
space.  The resulting influence matrix simultaneously serves three roles: it debiases
each score gap, estimates the full joint covariance, and becomes the input to the
high-dimensional multiplier algorithm.  No external efficient estimator is assumed as
a black box; Appendix~\ref{app:fdim-proof} proves the information lower bound, the
one-step expansion, and covariance consistency in the present matrix model.

For ranking, however, the relevant family of score gaps grows with the number of models
and tasks. Thus fixed-dimensional inference must be strengthened to a uniform expansion
over growing collections of score-gap contrasts.

\textbf{Uniform expansion for ranking tests.}
For growing collections of score-gap contrasts, our analysis gives the uniform expansion
\[
   \resizebox{0.8\textwidth}{!}{$  \frac{\sqrt n(\widehat\psi_j-\psi_j)}{\sigma_j}
    =
    \frac1{\sqrt n}\sum_{i=1}^n \frac{\phi_j(W_i)}{\sigma_j}
    +
    r_j,
    \qquad
    \max_j |r_j|
    \lesssim
    C_A
    \sqrt{\frac{\bar d\,\log^c(n\bar d)}{n}},$}
\]
with high probability. This expansion reduces simultaneous testing of correlated score
gaps to a high-dimensional Gaussian approximation problem. Together with Theorem~\ref{thm:fdim-clt}, it provides the technical foundation
for the rank and top-\(K\) certification procedure developed next.

\section{Simultaneous Rank and \texorpdfstring{Top-\(K\)}{Top-K} Certification}
\label{sec:ranking}

We now convert score-gap inference into leaderboard-level certification. Fix a task
\(t\) and a model \(m\). The rank of \(m\) on task \(t\) is determined by the signs of
the \(\dm-1\) score gaps
\(
    \Delta_{t,\ell}^{(m)}
    :=
    \Thetastar_{t,\ell}-\Thetastar_{t,m}
    =
    \psi_{\Gamma_{t,\ell}^{(m)}}(\Thetastar),
    \, \ell\ne m,
\)
with $\Gamma_{t,\ell}^{(m)}:=e_t(e_\ell-e_m)^\top$, since
$
    \rk_t(m)
    =
    1+\sum_{\ell\ne m}\mathbf 1\{\Delta_{t,\ell}^{(m)}>0\}
     +\sum_{\ell\ne m}
       \mathbf 1\{\Delta_{t,\ell}^{(m)}=0,\ \ell\prec_0m\}.
$
Thus rank inference is a multiple-testing problem over the competitors of \(m\): each
positive gap corresponds to one model ranked above \(m\), while an exact zero is
resolved by the fixed order \(\prec_0\). The goal is not only to test
one gap at a time, but to produce a confidence set for the rank and a certified
top-\(K\) membership decision.

\subsection{Studentized gap statistics and multiplier calibration}

For each \(\ell\ne m\), denote the
one-step estimator from Section~\ref{sec:fdim} as \( \widehat\Delta_{t,\ell}^{(m)}
    :=
    \widehat\psi_{\Gamma_{t,\ell}^{(m)}}\), with foldwise estimated
influence-function summand and standard error
\[
   \resizebox{1\textwidth}{!}{$ \widehat\phi_{t,\ell,i}^{(m)}
    :=
    s\!\left(Y_i,\widehat\eta_i^{(-k(i))}\right)
    \ip{\widehat H_{\Gamma_{t,\ell}^{(m)},k(i)}}{X_i},
    \quad
    \overline\phi_{t,\ell}^{(m)}
    :=
    \frac1n\sum_{i=1}^n\widehat\phi_{t,\ell,i}^{(m)},
    \quad
    \widehat\sigma_{t,\ell}^{(m)}
    :=
    \left\{
        \frac1n\sum_{i=1}^n
        \bigl(\widehat\phi_{t,\ell,i}^{(m)}
              -\overline\phi_{t,\ell}^{(m)}\bigr)^2
    \right\}^{1/2}.$}
\]
A pointwise interval for each gap is not enough: the rank depends on all
\(\dm-1\) gap signs jointly, and these gap estimators are correlated through shared
tasks, models, comparisons, and low-rank factors. 
We need a joint calibration that accounts for this covariance structure. Equivalently, simultaneous confidence bands require controlling the maximum studentized error over all competitors,
$  \max_{\ell\ne m}
    \left|
        \frac{\sqrt n(\widehat\Delta_{t,\ell}^{(m)}-\Delta_{t,\ell}^{(m)})}
             {\widehat\sigma_{t,\ell}^{(m)}}
    \right|.$
Draw i.i.d. multipliers \(\xi_i\sim N(0,1)\) and define the multiplier-bootstrap
statistic
\[
   \resizebox{0.4\textwidth}{!}{$  T_{t,m}^{\ast}
    :=
    \max_{\ell\ne m}
    \left|
        \frac1{\sqrt n}
        \sum_{i=1}^n
        \xi_i
        \frac{\widehat\phi_{t,\ell,i}^{(m)}
              -\overline\phi_{t,\ell}^{(m)}}
             {\widehat\sigma_{t,\ell}^{(m)}}
    \right|.$}
\]
Conditional on the data, this bootstrap process preserves the empirical dependence
among the gap statistics. Let \(c_{t,m}(1-\alpha)\) be the conditional
\((1-\alpha)\)-quantile of \(T_{t,m}^{\ast}\). We form simultaneous confidence bands
\[
   \resizebox{0.5\textwidth}{!}{$  \widehat I_{t,\ell}^{(m)}
    :=
    \left[
        \widehat\Delta_{t,\ell}^{(m)}
        \pm
        c_{t,m}(1-\alpha)
        \frac{\widehat\sigma_{t,\ell}^{(m)}}{\sqrt n}
    \right]
    =
    [\widehat L_{t,\ell}^{(m)},\widehat U_{t,\ell}^{(m)}].$}
\]

For a general family \(\mathcal J\) of ranking claims, index the contrasts by
\(j\in\mathcal J\), let
\(\widehat\sigma_j^2=\widehat\Sigma_{jj}\), and define the observed studentized
statistics
\[
T_j=\frac{\sqrt n\,\widehat\psi_j}{\widehat\sigma_j}.
\]
For the one-sided null \(H_{0j}:\psi_j\le0\), a large positive \(T_j\) is evidence that
the first model in \(\Gamma_j\) is better.  Two-sided simultaneous bands use the maximum
absolute centered error; directional ranking decisions then follow by checking whether
the resulting interval lies entirely above or below zero.

\begin{algorithm}[H]
\caption{Multiplier Calibration and Simultaneous Ranking Certification}
\label{alg:multiplier-ranking}
\begin{algorithmic}[1]
%\revisioncolor
\Require Outputs \((\widehat\psi,\widehat\Phi,\widehat\Sigma)\) from
Algorithm~\ref{alg:efficient-gaps}, contrast index set \(\mathcal J\), level \(\alpha\),
and number of draws \(B_{\rm mb}\).
\Ensure Simultaneous score-gap bands, rejected signs, rank intervals, and top-\(K\)
certificates.
\State Set \(\widehat\sigma_j=\sqrt{\widehat\Sigma_{jj}}\) and
\(T_j=\sqrt n\widehat\psi_j/\widehat\sigma_j\) for every \(j\in\mathcal J\).
\For{\(b=1,\ldots,B_{\rm mb}\)}
  \State Draw one common vector \(\xi^{(b)}=(\xi_1^{(b)},\ldots,\xi_n^{(b)})\),
  with independent \(N(0,1)\) coordinates.
  \State For every \(j\), compute
  \[
  Z_j^{\ast(b)}
  =\frac1{\sqrt n\,\widehat\sigma_j}
   \sum_{i=1}^n\xi_i^{(b)}\widehat\Phi_{ij},
  \qquad
  M^{\ast(b)}=\max_{j\in\mathcal J}|Z_j^{\ast(b)}|.
  \]
\EndFor
\State Let \(c_{1-\alpha}^\ast\) be the empirical \((1-\alpha)\)-quantile of
\(\{M^{\ast(b)}\}_{b=1}^{B_{\rm mb}}\).
\State For every \(j\), form
\(I_j=[\widehat\psi_j\pm c_{1-\alpha}^\ast\widehat\sigma_j/\sqrt n]\).
\State Certify \(\psi_j>0\) when \(\inf I_j>0\), certify \(\psi_j<0\) when
\(\sup I_j<0\), and otherwise mark the sign unresolved.
\For{each requested task--model pair \((t,m)\)}
  \State Set \(A_t(m)=\#\{\ell\ne m:\widehat L_{t,\ell}^{(m)}>0\}\) and
  \(B_t(m)=\#\{\ell\ne m:\widehat U_{t,\ell}^{(m)}<0\}\).
  \State Return the rank interval
  \(\widehat{\mathcal R}_t(m)=[1+A_t(m),\dm-B_t(m)]\).
  \State If \(\dm-B_t(m)\le K\), certify top-\(K\); if \(1+A_t(m)>K\),
  certify non-top-\(K\); otherwise return unresolved.
\EndFor
\end{algorithmic}
\end{algorithm}

The same multiplier \(\xi_i^{(b)}\) must be used across all targets in draw \(b\).
This is what preserves their empirical correlation.  Drawing independent multipliers
for different columns would replace \(\widehat\Sigma\) by a diagonal covariance and
would generally invalidate the joint calibration.  If \(q=|\mathcal J|\) is moderate,
one may equivalently simulate a Gaussian vector with correlation
\(\widehat D^{-1}\widehat\Sigma\widehat D^{-1}\), where
\(\widehat D=\operatorname{diag}(\widehat\sigma_1,\ldots,\widehat\sigma_q)\).
Algorithm~\ref{alg:multiplier-ranking} avoids forming or factorizing that matrix and
remains well defined when some contrasts are redundant.

The critical value is the quantile of the maximum of many correlated studentized
errors, not the usual pointwise Gaussian quantile. Under the uniform expansion in Section~\ref{sec:fdim}, high-dimensional Gaussian
approximation and multiplier-bootstrap theory justify calibrating the maximum
studentized score-gap error over the growing contrast family
\citep{chernozhukov2013gaussian,chernozhukov2017central}.

The inferential target is family-wise coverage rather than a collection of marginal
\(p\)-values.  Ranking is a logical function of many signs: one false sign declaration
can make a rank interval too narrow or certify the wrong top-\(K\) status.  Calibrating
the maximum produces a single event on which every sign statement used by the ranking
map is correct.  Compared with Bonferroni calibration, the multiplier critical value
adapts to the effective number of distinct directions.  Strongly correlated or redundant
gaps need not pay the same price as independent tests, while the guarantee remains valid
without guessing which correlations are negligible.

The choice of \(\mathcal J\) specifies the scientific claim.  For one model on one task,
it contains \(\dm-1\) competitors.  For one model across all tasks, it contains
\(\dt(\dm-1)\) gaps.  Simultaneous certification of every model and every task can use
the full family, with a wider critical value reflecting the stronger claim.  This
separation between estimating all gaps and choosing the claim family makes multiplicity
transparent: changing the family changes the calibration, not the underlying efficient
estimator.

\subsection{Inverting simultaneous gap bands into ranks and
\texorpdfstring{top-\(K\)}{top-K} sets}

On the simultaneous coverage event, the score-gap bands can be inverted into a rank
confidence band. Define
\[
    A_t(m):=\#\{\ell\ne m:\widehat L_{t,\ell}^{(m)}>0\},
    \qquad
    B_t(m):=\#\{\ell\ne m:\widehat U_{t,\ell}^{(m)}<0\}.
\]
Here \(A_t(m)\) is the number of competitors certified to be above \(m\), and
\(B_t(m)\) is the number certified to be below \(m\). The remaining competitors have
confidence bands crossing zero and are therefore unresolved. Thus a valid confidence
band for the rank is
\[
    \widehat{\mathcal R}_t(m)
    :=
    [\,1+A_t(m),\ \dm-B_t(m)\,].
\]

\begin{theorem}[Rank confidence band for one task]\label{thm:rank-one-task}
Under the assumptions of Theorem~\ref{thm:fdim-clt}, the uniform remainder condition
from Section~\ref{sec:fdim}, and the high-dimensional Gaussian approximation conditions
in Appendix~\ref{app:ranking-proof},
\[
    \Pr\{\rk_t(m)\in\widehat{\mathcal R}_t(m)\}
    \ge
    1-\alpha-o(1).
\]
\end{theorem}

\textbf{Top-\(K\) membership certification.}
The rank band directly yields a three-way top-\(K\) decision. If
$
    \dm-B_t(m)\le K,
$
then even the worst rank compatible with the confidence bands is at most \(K\), so we
certify \(m\in\TopK(t)\). If
$
    1+A_t(m)>K,
$
then even the best compatible rank is larger than \(K\), so we certify
\(m\notin\TopK(t)\). Otherwise, the membership decision is statistically unresolved. In LLM evaluation, this three-way output: $\{\textit{certified top-}K,\,
\textit{certified non-top-}K,\,
\textit{unresolved}\}$
separates reliable
leaderboard claims from comparisons that remain too noisy to certify.

\textbf{Simultaneous inference across tasks.}
Similar construction can be applied simultaneously across tasks. For a fixed model
\(m\), replace the single-task family \(\{\ell:\ell\ne m\}\) by
\[
    \mathcal J(m):=\{(t,\ell):t\in[\dt],\ \ell\ne m\},
    \qquad
    |\mathcal J(m)|=\dt(\dm-1).
\]
We compute the multiplier-bootstrap critical value for the maximum over
\(\mathcal J(m)\), and use the resulting bands to construct
\(\widehat{\mathcal R}_t(m)\) for every task \(t\). Since \(|\mathcal J(m)|\) is
polynomial in \(\bar d\), the same high-dimensional Gaussian approximation applies up
to logarithmic factors.

\begin{corollary}[Simultaneous task-wise rank inference]\label{cor:rank-all-tasks}
Under the conditions of Theorem~\ref{thm:rank-one-task}, with the bootstrap maximum
taken over \(\mathcal J(m)\),
\[
    \Pr\left\{
        \rk_t(m)\in\widehat{\mathcal R}_t(m)
        \ \text{for all }t\in[\dt]
    \right\}
    \ge
    1-\alpha-o(1).
\]
Consequently, the certified top-\(K\), certified non-top-\(K\), and unresolved
decisions for model \(m\) are simultaneously valid across all tasks.
\end{corollary}

Appendix~\ref{app:topk-set} extends the same score-gap band inversion to the entire
task-specific top-\(K\) set, producing inner and outer confidence sets satisfying
$
    \widehat{\TopK}_{\rm in}(t)
    \subseteq
    \TopK(t)
    \subseteq
    \widehat{\TopK}_{\rm out}(t),
$
simultaneously over tasks.

%\end{revisionblock}

\section{Experiments}\label{sec:experiments}

We next complement our theory with simulation studies and an LM Arena case study.
Throughout, our estimator uses the convex initializer of
Section~\ref{sec:estimation} followed by pairwise-logistic refinement,
and inference uses the cross-fitted one-step debiased estimator combined with
a Gaussian multiplier bootstrap calibrated against the high-dimensional CLT
of Section~\ref{sec:ranking}. We compare against per-task
Bradley-Terry (BTL), which fits each task row independently and uses the
analogous Wald-type plug-in influence functions for its multiplier bootstrap.
Unless otherwise noted, each experimental cell reports a Monte Carlo summary over $N=200$ trials,
formed by bootstrap resampling from the runs we executed; uncertainty
intervals in tables are 95\% bootstrap intervals for each metric. For both
simulation and LM Arena we report initialization accuracy, top-$K$ recovery,
the joint asymptotic distribution of contrast estimators, and rank inference
both for a single task and simultaneously across tasks.

\subsection{Simulation}\label{sec:exp-sim}

We use a square setting with $\dt=\dm=50$, true rank
$r^\star=5$, and amplitude $\alpha=5$, generating
$\Theta^\star=L^\star(A^\star)^\top$
from i.i.d.\ Gaussian factors and rescaling to
$\|\Theta^\star\|_\infty\le\alpha$.
Pairwise comparisons are sampled uniformly over (task, model-pair) and
generated from the BTL model with temperature $\tau=1$. We sweep the total
number of comparisons $n\in\{4{,}000,\,8{,}000,\,16{,}000,\,32{,}000\}$.
Cross-fitting uses $K_{\rm cf}=6$ folds.

\textbf{Estimation error decay.}
Theorem~\ref{thm:entrywise} proves that our refined estimator
attains $\Fnorm{\Thetahat-\Thetastar}\lesssim 1/\sqrt n$ and
$\infnorm{\Thetahat-\Thetastar}\lesssim 1/\sqrt n$ up to $\mathrm{polylog}$
factors when $n\gtrsim r\bar d\log\bar d$.
Figure~\ref{fig:app-sim-recovery} plots both errors versus $n$ on a log-log
scale and overlays a $1/\sqrt n$ reference. Our estimator's empirical
errors track the predicted rate, while per-task BTL, which fits each
task row independently using only $n/\dt$ comparisons, has a far larger
constant. The gap is largest at $n=4{,}000$ (per-task BTL Frobenius error
exceeds $1000$, almost two orders of magnitude above our estimator)
and shrinks as $n$ grows.

\begin{figure}[!htbp]
\centering
\includegraphics[width=0.85\linewidth]{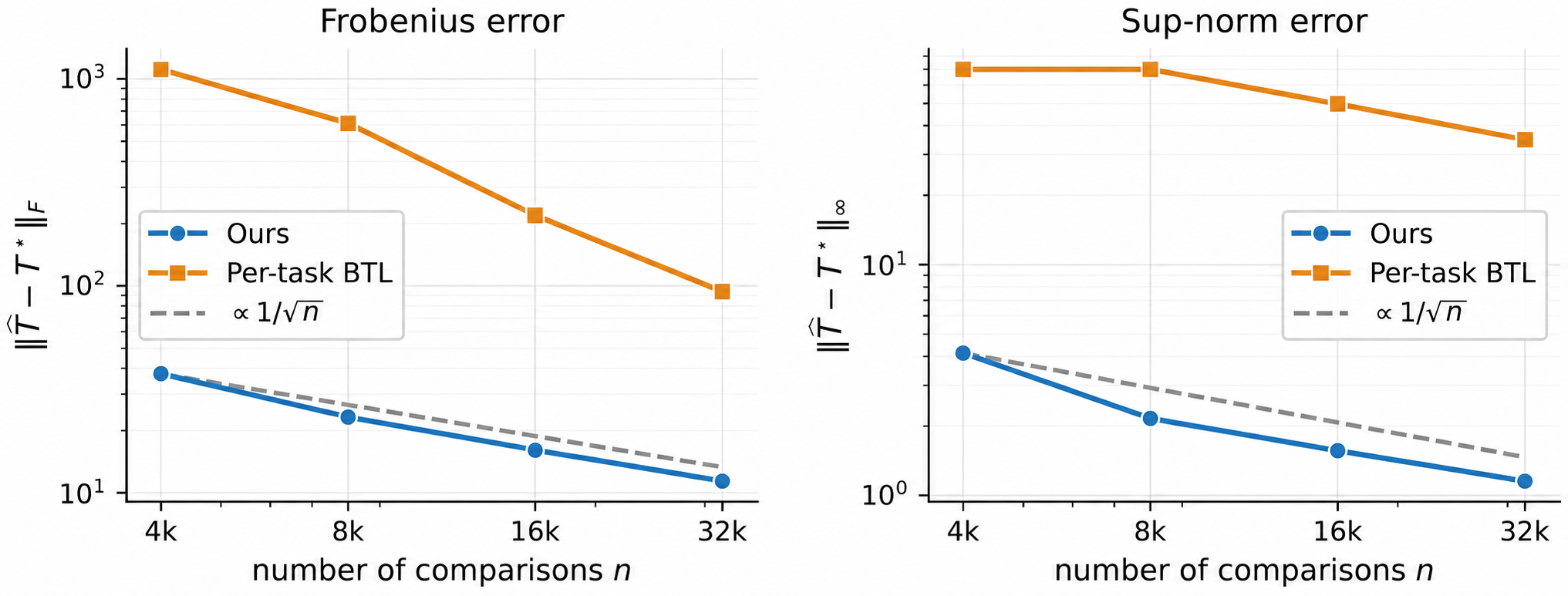}
\caption{Estimation error vs.\ $n$ at $\dt=\dm=50$, $r^\star=5$, $\alpha=5$.
Our estimator's Frobenius and sup-norm errors decay at the predicted
$1/\sqrt n$ rate; per-task BTL, which does not pool across tasks, incurs a
much larger constant.}
\label{fig:app-sim-recovery}
\end{figure}

 \textbf{Top-$K$ Recovery.} Table~\ref{tab:sim-topk} reports the
per-task top-$K$ set Hamming distance, averaged over the $\dt$ tasks, for
$K\in\{5,10\}$. Our estimator yields uniformly lower Hamming distance
than per-task BTL across all $n$, with the gap shrinking as $n$ grows.

\begin{table}[t]
\centering
\small
\setlength{\tabcolsep}{4pt}
\caption{Simulation, mean per-task top-$K$ Hamming distance across $\dt=50$
tasks, $N=200$ trials. Smaller is better; entries are mean (95\% CI).}
\label{tab:sim-topk}
\begin{tabular}{c cc cc}
\toprule
& \multicolumn{2}{c}{$K=5$} & \multicolumn{2}{c}{$K=10$} \\
\cmidrule(lr){2-3}\cmidrule(lr){4-5}
$n$ & Ours & Per-task BTL & Ours & Per-task BTL \\
\midrule
4{,}000  & 0.482 (0.478, 0.486) & 0.730 (0.728, 0.733) & 0.388 (0.385, 0.390) & 0.596 (0.594, 0.598) \\
8{,}000  & 0.339 (0.335, 0.342) & 0.617 (0.613, 0.620) & 0.257 (0.254, 0.259) & 0.489 (0.487, 0.491) \\
16{,}000 & 0.237 (0.234, 0.239) & 0.479 (0.476, 0.482) & 0.181 (0.179, 0.182) & 0.366 (0.363, 0.368) \\
32{,}000 & 0.167 (0.164, 0.169) & 0.360 (0.357, 0.362) & 0.129 (0.127, 0.130) & 0.269 (0.267, 0.271) \\
\bottomrule
\end{tabular}
\end{table}

\textbf{Joint asymptotic Gaussianity of two contrasts.}
Theorem~\ref{thm:fdim-clt} shows that for any fixed collection of
score-gap contrasts, the cross-fitted one-step estimator
\((\widehat\psi_1,\dots,\widehat\psi_q)\) is jointly asymptotically
Gaussian, with covariance equal to the inverse of the semiparametric
information. We verify this claim directly with two contrasts. Pick
models $a,b,c\in[\dm]$ on a single task $t\in[\dt]$ and form
\(\psi_1=\Thetastar_{t,a}-\Thetastar_{t,b}\) and
\(\psi_2=\Thetastar_{t,a}-\Thetastar_{t,c}\); the shared model $a$ induces
nontrivial correlation between the two contrast estimators.
Over $N=500$ Monte Carlo trials at $n=16{,}000$ we compute
\((\widehat\psi_1,\widehat\psi_2)\) and the per-trial plug-in covariance
\(\widehat\Sigma\). The empirical covariance \(\Sigma_{\mathrm{emp}}\) of
\((\widehat\psi_1,\widehat\psi_2)\) across trials and the
mean-of-plug-in $\Sigma_{\mathrm{thy}}$ are reported in
Table~\ref{tab:app-joint-sigma}. The two are close in entrywise scale and
correlation, and the $95\%$ confidence ellipse derived from
$\Sigma_{\mathrm{thy}}$ achieves nominal coverage on the empirical samples
(0.950 vs.\ target 0.95).
Figure~\ref{fig:app-joint-gaussian} shows the centered Monte Carlo samples
together with the theoretical 95\% Gaussian ellipse and the standardized
marginals overlaid on $\mathcal N(0,1)$, both consistent with bivariate
Gaussianity.

\begin{table}[!htbp]
\centering
\small
\setlength{\tabcolsep}{6pt}
\caption{Empirical vs.\ theoretical (mean plug-in) covariance of the
two-contrast cross-fitted one-step estimators across $N=500$ trials at
$\dt=\dm=50$, $r^\star=5$, $\alpha=5$, $n=16{,}000$.}
\label{tab:app-joint-sigma}
\begin{tabular}{l cc cc}
\toprule
& \multicolumn{2}{c}{$\Sigma_{\mathrm{emp}}$} & \multicolumn{2}{c}{$\Sigma_{\mathrm{thy}}$} \\
\cmidrule(lr){2-3}\cmidrule(lr){4-5}
& $\widehat\psi_1$ & $\widehat\psi_2$ & $\widehat\psi_1$ & $\widehat\psi_2$ \\
\midrule
$\widehat\psi_1$ & 0.276 & 0.114 & 0.343 & 0.161 \\
$\widehat\psi_2$ & 0.114 & 0.356 & 0.161 & 0.385 \\
\bottomrule
\end{tabular}
\\[0.4em]
\begin{tabular}{l c c c c c}
correlation & $\rho_{\mathrm{emp}}=0.365$ & $\rho_{\mathrm{thy}}=0.443$ & 95\% ellipse cov.\ & $0.950$ & (target $0.95$) \\
\end{tabular}
\end{table}

\begin{figure}[!htbp]
\centering
\includegraphics[width=0.92\linewidth]{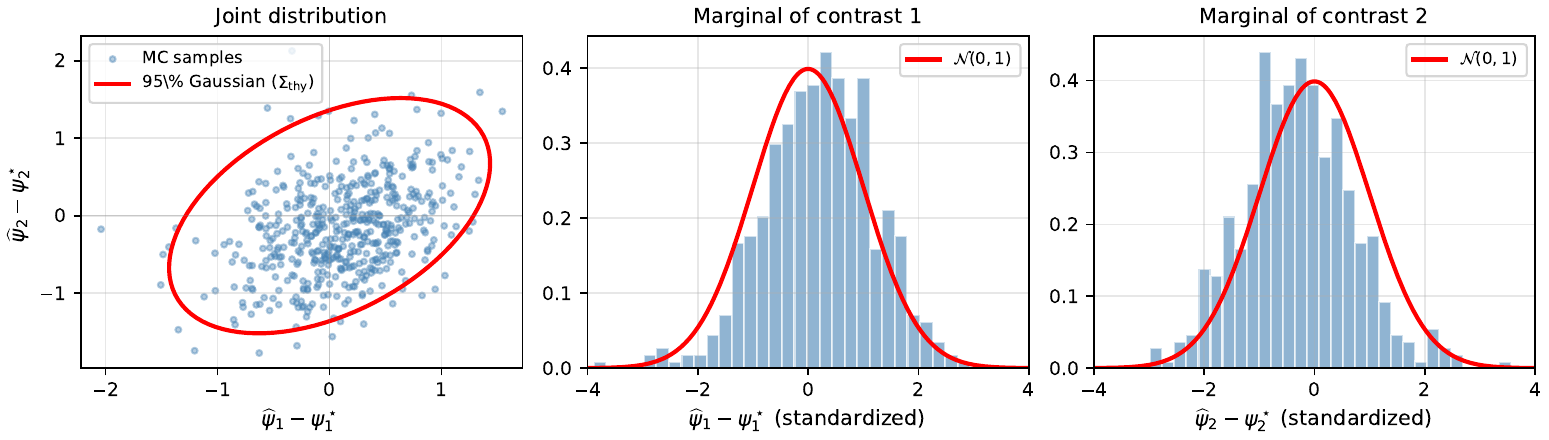}
\caption{Joint asymptotic Gaussianity at $n=16{,}000$. \textbf{Left}:
centered Monte Carlo samples $(\widehat\psi_1-\psi_1^\star,\widehat\psi_2-\psi_2^\star)$
and the theoretical 95\% Gaussian ellipse derived from
$\Sigma_{\mathrm{thy}}$. \textbf{Middle/Right}: standardized marginals
overlaid on $\mathcal N(0,1)$.}
\label{fig:app-joint-gaussian}
\end{figure}

\textbf{Single-task rank inference.}
Theorem~\ref{thm:rank-one-task} provides a multiplier-bootstrap rank
confidence set for a single model on a single task. To exercise the
procedure for an arbitrary target, we draw one model $m_t$ uniformly at
random for each task $t\in[\dt]$ (fixed across trials) and form the
$1-\alpha$ rank confidence set for $m_t$ on task $t$ at $\alpha=0.05$,
$K_{\text{top}}=10$, and $n=16{,}000$. Table~\ref{tab:app-sim-single}
reports the per-task coverage of the true rank, the fraction of trials at
which the rule correctly certifies in/out of top-$K$, and the mean rank-CI
width. Our method preserves coverage at the nominal level and
correctly certifies $\approx 29\%$ of (task, trial) instances; per-task BTL
correctly certifies only $\approx 8\%$ since at this sample size each task
has roughly $n/\dt=320$ comparisons, which is too sparse for per-task inference
to produce informative rank CIs.

\begin{table}[!htbp]
\centering
\small
\setlength{\tabcolsep}{4pt}
\caption{Simulation, single-task rank confidence set for a randomly chosen
model on each task at $n=16{,}000$, $K_{\text{top}}=10$, $N=200$ trials.}
\label{tab:app-sim-single}
\begin{tabular}{l cc}
\toprule
& Ours & Per-task BTL \\
\midrule
Coverage of true rank          & 1.000 (1.000, 1.000) & 0.967 (0.964, 0.970) \\
Correct top-$K$ certification & 0.289 (0.284, 0.293) & 0.081 (0.077, 0.086) \\
Mean rank-CI width             & 36.7 (36.6, 36.8)    & 44.6 (44.5, 44.8) \\
\bottomrule
\end{tabular}
\end{table}

 \textbf{Simultaneous inference.} We also show simultaneous inference across tasks for a fixed model. Fixing one model
$m^\star$ chosen at random, we form a single multiplier-bootstrap critical
value over the joint family $\{(t,\ell):\ell\ne m^\star,\,t\in[\dt]\}$ at
$n=32{,}000$ and $K_{\text{top}}=10$, yielding simultaneous rank-CIs for
$m^\star$ across all $\dt=50$ tasks. Table~\ref{tab:sim-simul} reports the
fraction of (trial, task) instances on which the procedure produces a
non-trivial decision (\emph{resolved}, meaning either in top $K$ or not in top $K$), and the average rank-CI width. Our method resolves over $5\times$ more cases than per-task BTL while
maintaining nominal coverage. 

\begin{table}[t]
\centering
\small
\setlength{\tabcolsep}{4pt}
\caption{Simulation, simultaneous rank confidence set for one model across
all $\dt=50$ tasks, $n=32{,}000$, $K_{\text{top}}=10$, $N=200$ trials.}
\label{tab:sim-simul}
\begin{tabular}{l cc}
\toprule
& Ours & Per-task BTL \\
\midrule
Per-task coverage            & 1.000 (1.000, 1.000) & 0.990 (0.988, 0.992) \\
Resolved (non-trivial)       & 0.315 (0.310, 0.321) & 0.056 (0.053, 0.059) \\
Mean rank-CI width           & 31.3 (31.1, 31.4)    & 46.4 (46.3, 46.5) \\
\bottomrule
\end{tabular}
\end{table}

\subsection{LM Arena}\label{sec:exp-arena}

We use the Arena 140K dataset \citep{lmarena_human_preference_140k}, restrict to
the top-30 most frequently compared models, and assign each comparison to one
of $\dt=10$ task categories using the platform's category metadata. After
preprocessing this leaves $n=81{,}150$ pairwise comparisons across $\dm=30$
models and $\dt=10$ tasks. We define the full-data reference score matrix $\Theta^\star$ by fitting a separate BTL model to
each task using all available comparisons. For the low-rank method, we set \(r=3\)
based on the empirical singular value spectrum of this reference matrix, whose
leading three singular values account for more than \(94\%\) of the total energy. We use the subsampling fractions reported for each experiment below,
where $f=1$ corresponds to using all $n$ comparisons. Per-task BTL fits each
task independently on the available subsample.

\textbf{Top-$K$ recovery sweep.}
We report the per-task top-$K$ set Hamming
distance of our estimator and per-task BTL across all subsampling
fractions $f\in\{0.02,0.05,0.10,0.20,0.50\}$
(Table~\ref{tab:app-real-topk}). Our estimator dominates per-task
BTL by a wide margin at $f\le0.10$ and the two methods become comparable
around $f=0.5$. The pattern mirrors the
simulation: pooling across tasks via low-rank structure is most beneficial
when the per-task data is sparse.

\begin{table}[!htbp]
\centering
\small
\setlength{\tabcolsep}{4pt}
\caption{LM Arena, mean per-task top-$K$ set Hamming distance across the
$\dt=10$ task categories, $N=200$ trials.}
\label{tab:app-real-topk}
\begin{tabular}{c r cc cc}
\toprule
& & \multicolumn{2}{c}{$K=5$} & \multicolumn{2}{c}{$K=10$} \\
\cmidrule(lr){3-4}\cmidrule(lr){5-6}
$f$ & $n_{\text{sub}}$ & Ours & Per-task BTL & Ours & Per-task BTL \\
\midrule
0.02 & 1{,}623  & 0.577 & 0.654 & 0.411 & 0.477 \\
0.05 & 4{,}057  & 0.463 & 0.528 & 0.311 & 0.376 \\
0.10 & 8{,}115  & 0.364 & 0.427 & 0.251 & 0.296 \\
0.20 & 16{,}230 & 0.294 & 0.333 & 0.205 & 0.233 \\
0.50 & 40{,}575 & 0.221 & 0.198 & 0.144 & 0.134 \\
\bottomrule
\end{tabular}
\end{table}

\textbf{Single-task ranking: gemini-2.5-pro on \texttt{math}.}
We report the single-task rank-CI for
\texttt{gemini-2.5-pro} on the \texttt{math} task (true rank $1$).
Table~\ref{tab:app-real-single-math} sweeps the subsampling fraction; in
the sparse regime our joint method certifies in top-$10$ on a much larger
fraction of trials than per-task BTL ($22\%$ vs.\ $0\%$ at $f=0.20$),
and at $f=1$ both methods correctly certify with comparable widths. Our method's advantage is precisely in the sparse regime where
single-task BTL has too few comparisons to pin down the rank.

\begin{table}[!htbp]
\centering
\small
\setlength{\tabcolsep}{4pt}
\caption{LM Arena, single-task rank CI for \texttt{gemini-2.5-pro} on
\texttt{math} (true rank $1$), $K_{\text{top}}=10$, $N=200$ trials.}
\label{tab:app-real-single-math}
\begin{tabular}{c ccc ccc}
\toprule
& \multicolumn{3}{c}{Ours} & \multicolumn{3}{c}{Per-task BTL} \\
\cmidrule(lr){2-4}\cmidrule(lr){5-7}
$f$ & Coverage & Cert.\ rate & Width & Coverage & Cert.\ rate & Width \\
\midrule
0.20 & 1.000 & 0.220 & 15.9 & 1.000 & 0.000 & 21.6 \\
0.30 & 1.000 & 0.500 & 11.6 & 1.000 & 0.160 & 16.3 \\
0.50 & 1.000 & 0.520 & 10.5 & 1.000 & 0.320 & 12.1 \\
1.00 & 1.000 & 1.000 &  7.4 & 1.000 & 1.000 &  6.1 \\
\bottomrule
\end{tabular}
\end{table}

\textbf{Simultaneous ranking: Gemini-2.5-pro across all tasks.} For Gemini 2.5 pro, which performs in top $3$ in all tasks, we form simultaneous rank CIs across all $\dt=10$ tasks via a
single multiplier-bootstrap critical value over the family
$\{(t,\ell):t\in[\dt],\,\ell\ne m^\star\}$. Table~\ref{tab:real-simul}
reports per-task coverage, the per-task fraction of resolved decisions, and
mean rank-CI width. Our method resolves $\approx 87\%$ of (task,trial)
cases at $f=1$ versus $61\%$ for per-task BTL, with smaller CI widths.

A representative phenomenon is the \texttt{analytical} task, on which the
per-task BTL CI never certifies in top-$10$ at any subsampling fraction we
considered (the rank-CI width is consistently above $10$). With low-rank
pooling, our estimator certifies in top-$10$ on every trial at $f=1$ for
the same task. Tasks where the chosen model is consistently best (rank $1$
on most categories) similarly benefit from the cross-task information: in
the sparse regime, a single task simply does not have enough comparisons to
certify, but pooling across tasks resolves the question.
\vspace{-2em}
\begin{table}[h!]
\centering
\small
\setlength{\tabcolsep}{4pt}
\caption{LM Arena, simultaneous rank confidence set for gemini-2.5-pro across
all $\dt=10$ tasks, $K_{\text{top}}=10$, $N=200$ trials. \textit{Resolved}
counts (task, trial) instances on which the rule certifies in or out of
top-$K$.}
\label{tab:real-simul}
\begin{tabular}{c ccc ccc}
\toprule
& \multicolumn{3}{c}{Ours} & \multicolumn{3}{c}{Per-task BTL} \\
\cmidrule(lr){2-4}\cmidrule(lr){5-7}
$f$ & Coverage & Resolved & Width & Coverage & Resolved & Width \\
\midrule
0.20 & 0.998 & 0.218 & 17.4 & 1.000 & 0.044 & 23.2 \\
0.30 & 0.998 & 0.434 & 12.7 & 1.000 & 0.124 & 19.6 \\
0.50 & 0.984 & 0.742 &  7.9 & 1.000 & 0.270 & 15.1 \\
1.00 & 0.958 & 0.866 &  6.7 & 1.000 & 0.610 &  9.1 \\
\bottomrule
\end{tabular}
\end{table}

Both the simulation and LM Arena studies show that
exploiting low-rank structure across tasks reduces estimation error and, more importantly, produces substantially
tighter rank confidence sets than per-task BTL. % %while maintaining the bootstrap coverage guarantees of Theorem~\ref{thm:rank-one-task} and Corollary~\ref{cor:rank-all-tasks}. 
The advantage is largest in the sparse
regime, which is common in real LLM benchmark
deployment, where pooling information across tasks is an effective way to
get informative rank inference.

% \section{Conclusion}\label{sec:conclusion}

% We develop a statistical framework for uncertainty-aware task-specific LLM ranking from sparse pairwise comparisons using rank-$r$ structure to share information across related tasks while preserving task-level heterogeneity, yielding entrywise score control, task-wise top-K recovery guarantees, efficient score-gap inference, and simultaneous rank/top-K certification. Experiments on synthetic data and LM Arena comparisons show that low-rank sharing improves sample efficiency over independent per-task Bradley--Terry estimation and produces tighter, better-calibrated ranking certificates in sparse regimes. 

\phantomsection
\label{sec:score-whitening-ref}

%\section*{Broader Impact}
%Task-specific LLM rankings increasingly influence model selection, procurement, deployment, and public perceptions of model capability. By reporting uncertainty-aware rank and top-K certificates, our framework helps reduce overconfident claims based on sparse or imbalanced preference data and makes clear when apparent leaderboard differences are statistically unresolved. 

{\baselineskip=23pt
\bibliographystyle{asa}
\bibliography{citation}
}

\newpage
\appendix
\baselineskip=24pt
\setcounter{equation}{0}
\setcounter{section}{0}
\renewcommand{\thesection}{S.\arabic{section}}
\renewcommand{\theequation}{S\arabic{equation}}
\renewcommand{\theHsection}{S.\arabic{section}}
\renewcommand{\theHsubsection}{S.\arabic{section}.\arabic{subsection}}

\begin{center}    
{\large\bf Supplementary Material for \\``Low Rank for Rank: Uncertainty-Aware Task-Specific LLM Ranking under Sparse Pairwise Comparisons"}
\end{center}
\bigskip

This supplementary material provides all technical proofs for
the main results. Appendix~\ref{app:notation} collects the notation, assumptions,
rate conditions, and master high-probability event used throughout the analysis.
Appendix~\ref{app:entrywise-proof} proves the uniform entrywise estimation result
in Theorem~\ref{thm:entrywise}, including the convex low-rank initializer and the
model- and task-side refinement analysis. Appendix~\ref{app:topk-hamming} proves
the task-wise top-\(K\) Hamming and exact-recovery guarantees in
Proposition~\ref{prop:topk-hamming}. Appendix~\ref{app:fdim-proof} develops the
semiparametric efficiency theory and proves the finite-dimensional score-gap
inference result in Theorem~\ref{thm:fdim-clt}. Appendix~\ref{app:ranking-proof}
establishes the multiplier-bootstrap theory for simultaneous rank and top-\(K\)
certification.

For readers' convenience, the following table summarizes the correspondence between the main-text results and their appendix proofs.
\begin{center}
\small
\begin{tabular}{@{}lll@{}}
\toprule
\textbf{Main-text result} & \textbf{Statement} & \textbf{Appendix proof}\\
\midrule
Theorem~\ref{thm:entrywise} & uniform entrywise estimation & Appendix~\ref{app:entrywise-proof-final} (Theorem~\ref{thm:pairwise-max-app})\\
\quad convex initialization & Frobenius rate \(\sqrt{r\dt\dm\bar d\log(n\bar d)/n}\) & Appendix~\ref{app:rsc-main} (Theorem~\ref{thm:convex-main-app})\\
\quad model-row refinement & \(\ell_{2,\infty}\)-bound on model factor & Appendix~\ref{app:entrywise-proof-leftdet} (Proposition~\ref{prop:left-final-app})\\
\quad task-column refinement & \(\ell_{2,\infty}\)-bound on task factor & Appendix~\ref{app:entrywise-proof-right} (Proposition~\ref{prop:right-final-app})\\
Proposition~\ref{prop:topk-hamming} & taskwise top-\(K\) Hamming & Appendix~\ref{app:topk-hamming-prop} (Proposition~\ref{prop:topk-hamming-app})\\
Theorem~\ref{thm:fdim-clt} & joint efficient CLT, fixed \(q\) & Appendix~\ref{app:fdim-proof-clt} (Theorem~\ref{thm:fdim-clt-app})\\
\quad single-contrast remainder & remainder \(\le C_A\bar d\log^c\bar d/n\) & Appendix~\ref{app:fdim-proof-decomp} (Theorem~\ref{thm:single-contrast-remainder-app})\\
\quad uniform remainder & remainder \(\le C_A\sqrt{\bar d\log^c/n}\) over \(\mathcal F\) & Appendix~\ref{app:fdim-proof-uniform} (Theorem~\ref{thm:uniform-remainder-app})\\
\quad joint efficiency & Loewner lower bound \(\bar\Sigma\succeq\Sigma\) & Appendix~\ref{app:fdim-proof-loewner} (Proposition~\ref{prop:fdim-loewner-app})\\
\quad variance consistency & relative-error \(O(r_n)\) & Appendix~\ref{app:fdim-proof-varcons} (Proposition~\ref{prop:fdim-var-cons-app})\\
\quad covariance consistency & correlation-error \(O(r_n)\) & Appendix~\ref{app:fdim-proof-cov} (Proposition~\ref{prop:fdim-cov-cons-app})\\
Theorem~\ref{thm:rank-one-task} & rank confidence band for one task & Appendix~\ref{app:ranking-proof-app1} (Theorem~\ref{thm:rank-one-task-app})\\
Corollary~\ref{cor:rank-all-tasks} & simultaneous taskwise rank inference & Appendix~\ref{app:ranking-proof-app2} (Corollary~\ref{cor:rank-all-tasks-app})\\
Top-\(K\) set extension & inner/outer top-\(K\) confidence sets & Appendix~\ref{app:topk-set} (Theorem~\ref{thm:topk-set-app})\\
\bottomrule
\end{tabular}
\end{center}

\section{Notation, assumptions, and master good event}\label{app:notation}

This appendix collects the notation, assumptions, and probability
calibrations used throughout
Appendices~\ref{app:entrywise-proof}--\ref{app:topk-set}.  All assumptions
are stated directly for the present \((\dt\times\dm)\) pairwise-ranking model.
Each condition used by a main-text theorem is listed here and subsequently verified
or invoked at the precise step where it enters the proof.

\subsection{Notation}\label{app:notation-crosswalk}

We retain the notation of Section~\ref{sec:setup}.  The latent ability
matrix is \(\Thetastar\in\R^{\dt\times\dm}\), row-centered
(\(\Thetastar\mathbf 1_{\dm}=0\)), of rank \(r\) with reduced singular value
decomposition
\[
    \Thetastar=U^\star\Sigma^\star(V^\star)^\top,
    \qquad
    U^\star\in\R^{\dt\times r},\
    V^\star\in\R^{\dm\times r},\
    \Sigma^\star=\mathrm{diag}(\sigma_1^\star,\ldots,\sigma_r^\star).
\]
The singular vectors are \(\mu\)-incoherent, the condition number is
\(\kappa:=\sigma_1^\star/\sigma_r^\star\), and the entrywise bound is
\(\infnorm{\Thetastar}\le B\).  We write
\(\bar d:=\max(\dt,\dm)\) for the maximum mode dimension and
\(\dstar:=\dt\dm\) for the ambient cardinality of the matrix.  The
\emph{effective comparison dimension}, which appears throughout the
analysis as the natural normalization for pairwise contrasts, is
\[
    \dstar_{\rm eff}
    :=
    \frac{\dt\,(\dm-1)}{2}
    \asymp
    \frac{\dt\,\dm}{2}
    =
    \frac{\dstar}{2},
\]
and we will use \(\dstar\) and \(\dstar_{\rm eff}\) interchangeably up to
the absolute constant \(2\).

\textbf{Design and observation model.}
For each round \(i\in[n]\), the design tensor is
\(X_i:=e_{t_i}(e_{m_i}-e_{m_i'})^\top\in\R^{\dt\times\dm}\), where the
task index \(t_i\in[\dt]\) is sampled from a distribution \(\nu\) and
the unordered model pair \(\{m_i,m_i'\}\subset[\dm]\) is sampled from a
task-dependent distribution \(\pi_{t_i}\).  The predictor is
\(\eta_i^\star:=\ip{X_i}{\Thetastar}=\Thetastar_{t_i,m_i}-\Thetastar_{t_i,m_i'}\),
and the observation is \(Y_i\sim\mathrm{Bernoulli}(\sigma(\eta_i^\star))\)
with \(\sigma(x):=(1+e^{-x})^{-1}\) the logistic link.  The pairwise score
and Fisher weight are
\[
    s_\eta(y,\eta):=-\partial_\eta\ell(y,\eta)
    =
    y-\sigma(\eta),
    \qquad
    I(\eta):=\E[s_\eta(Y,\eta)^2\mid\eta]
    =
    \sigma(\eta)\{1-\sigma(\eta)\}\in(0,1/4].
\]
Thus \(s_\eta\) is the log-likelihood score; the gradient of the convex
negative log-likelihood used for estimation is \(-s_\eta=\sigma(\eta)-y\).
Under \(\infnorm{\Thetastar}\le B\), the predictor satisfies
\(|\eta^\star|\le 2B\), and consequently the Fisher information is
bounded above and below: there exist constants
\(0<c_B\le C_B<\infty\) depending only on \(B\) such that
\(c_B\le I(\eta^\star)\le C_B\) almost surely.

\textbf{Tangent space and operators.}
The signal tangent space at \(\Thetastar\) intersected with the
row-centering identification constraint is
\[
    \mathbb T
    :=
    \{U^\star A^\top + C(V^\star)^\top
    :\, A\in\R^{\dm\times r},\ C\in\R^{\dt\times r},\
       A^\top\mathbf 1_{\dm}=0\}.
\]
Indeed, \(\Thetastar\mathbf 1_{\dm}=0\) and
\(\sigma_r^\star>0\) imply
\((V^\star)^\top\mathbf 1_{\dm}=0\).  Put
\(R_{\rm m}:=I_{\dm}-\dm^{-1}\mathbf 1_{\dm}\mathbf 1_{\dm}^\top\).
The orthogonal projector onto \(\mathbb T\) is
\[
 P_{\mathbb T}H
 =P_{U^\star}H R_{\rm m}+(I-P_{U^\star})H P_{V^\star};
\]
Appendix~\ref{app:fdim-proof-projector} verifies this formula.  The Fisher
(information) operator
\(G:\R^{\dt\times\dm}\to\R^{\dt\times\dm}\) is defined by
\begin{equation}\label{eq:G-operator-app}
    \ip{G U}{V}
    :=\E^\star\!\left[I(\eta^\star)\ip{U}{X}\ip{V}{X}\right],
    \qquad U,V\in\R^{\dt\times\dm},
\end{equation}
and the restricted information operator on the tangent space is
\(A_{\mathbb T}:=(P_{\mathbb T}GP_{\mathbb T})|_{\mathbb T}:\mathbb T\to\mathbb T\).
Lemma~\ref{lem:weighted-moment-app} shows that this operator is positive
definite.  We use the ambient Moore--Penrose extension
\[
 \mathcal A^\dagger:=P_{\mathbb T}A_{\mathbb T}^{-1}P_{\mathbb T}.
\]
For any contrast \(\Gamma\in\R^{\dt\times\dm}\), the efficient direction
is \(H_\Gamma^\star:=\mathcal A^\dagger\Gamma\) and the efficient
influence function is
\[
    \phi_\Gamma(W_i)
    :=
    s_\eta(Y_i,\eta_i^\star)\ip{H_\Gamma^\star}{X_i},
    \qquad
    V_{\rm eff}(\Gamma)
    :=
    \E^\star[\phi_\Gamma^2]
    =
    \ip{P_{\mathbb T}\Gamma}{A_{\mathbb T}^{-1}P_{\mathbb T}\Gamma}.
\]

\textbf{Plug-in operators.}
%\begin{revisionblock}
Given an initial estimator \(\widehat\Theta\) of \(\Thetastar\) and the
estimated singular subspaces \(\widehat U,\widehat V\) (obtained, e.g.,
from the rank-\(r\) SVD of \(\widehat\Theta\)), we define the
\emph{estimated tangent projector} \(\widehat P_{\mathbb T}\) using
\(\widehat U,\widehat V\) in place of \(U^\star,V^\star\), and the
\emph{foldwise design-averaged plug-in information operator}
\begin{equation}\label{eq:Ghat-def-app}
    \ip{\widehat G U}{V}
    :=
    \E_X\!\left[
    I(\ip{\widehat\Theta}{X})\,\ip{U}{X}\,\ip{V}{X}\right],
    \qquad X\sim(\nu,\pi),
\end{equation}
where the randomized design law \((\nu,\pi)\) is known and the
expectation is a finite sum over its support.  Thus \(\widehat G\)
depends on the training data only through \(\widehat\Theta\), and is
independent of the evaluation fold.  The estimated restricted
operator is
\(\widehat A_{\widehat{\mathbb T}}
:=(\widehat P_{\mathbb T}\widehat G\widehat P_{\mathbb T})|_{\widehat{\mathbb T}}\),
and its ambient extension and the estimated direction are
\[
 \widehat{\mathcal A}^{\dagger}
 :=\widehat P_{\mathbb T}
   \widehat A_{\widehat{\mathbb T}}^{-1}\widehat P_{\mathbb T},
 \qquad
 \widehat H_\Gamma:=\widehat{\mathcal A}^{\dagger}\Gamma.
\]
%\end{revisionblock}
\textbf{Norms.}
We use the following norms.  For a matrix \(M\in\R^{\dt\times\dm}\),
\(\Fnorm M\) is the Frobenius norm, \(\infnorm M:=\max_{t,m}|M_{t,m}|\)
the entrywise norm, \(\norm M_*\) the nuclear norm,
\(\norm M_{\rm op}\) the spectral norm,
\(\twoninfnorm M:=\max_t\norm{e_t^\top M}_2\) the row \(\ell_2/\ell_\infty\)
norm, and \(\norm M_{\infty\to\infty}\) the induced
\(\ell_\infty\to\ell_\infty\) operator norm.  For a vector \(v\),
\(\oneNorm v:=\sum_i|v_i|\), \(\norm v_2\) the Euclidean norm,
\(\infnorm v:=\max_i|v_i|\).  For tensors / linear operators acting on
\(\R^{\dt\times\dm}\), these norms apply to the tensor when flattened
as a vector.

\subsection{Score regularity}\label{app:notation-score}

The following assumption records the score identities and uniform bounds used in the
BTL likelihood expansions below.

\begin{assumption}[Score regularity for the BTL link]\label{ass:score-app}
The BTL log-likelihood satisfies, with \(s_\eta\) and \(I\) as above:
\begin{enumerate}[label=(\roman*),leftmargin=2.4em,topsep=2pt,itemsep=2pt]
\item \emph{Centering.}
\(\E^\star[s_\eta(Y,\eta^\star)\mid X]=0\).
\item \emph{Bounded support.}
\(\infnorm{\Thetastar}\le B\) for a constant \(B>0\), so
\(|\eta^\star|\le 2B\) almost surely, and hence the Fisher weight
satisfies
\[
    c_B\le I(\eta^\star)\le C_B,
\]
where \(c_B:=\sigma(2B)\sigma(-2B)>0\) and
\(C_B:=1/4\) depend only on \(B\).
\item \emph{Bounded score derivatives.} For every \((y,\eta)\) in the
relevant range,
\(|\partial_\eta s_\eta(y,\eta)|\le 1/4\) and
\(|\partial_\eta^2 s_\eta(y,\eta)|\le L_3\) for an absolute constant
\(L_3>0\) (one can take \(L_3=1/(6\sqrt 3)\)).
\item \emph{Sub-Gaussian / sub-exponential tail.} Since
\(s_\eta(Y,\eta)=Y-\sigma(\eta)\in[-1,1]\) almost surely, every moment is
bounded uniformly: \(\norm{s_\eta(Y,\eta)}_{\psi_2}\le 1\) and
\(\norm{s_\eta(Y,\eta)}_{\psi_1}\le 1\).
\end{enumerate}
\end{assumption}

\begin{remark}[BTL verification]\label{rem:btl-verification-app}
For the BTL model, \(s_\eta(y,\eta)=y-\sigma(\eta)\),
\(\partial_\eta s_\eta=-\sigma'(\eta)=-I(\eta)\in[-1/4,0]\), and
\(\partial_\eta^2 s_\eta=-\sigma''(\eta)\) is uniformly bounded.  The
boundedness of \(I(\eta^\star)\) under \(\infnorm{\Thetastar}\le B\) is
the consequence of \(|\eta^\star|\le 2B\) plus the strict positivity of
\(\sigma'\) on any compact interval of \(\R\).
\end{remark}

\subsection{Sampling design}\label{app:notation-design}

We work under the near-uniform sampling design of
Section~\ref{sec:setup}.

\begin{assumption}[Near-uniform sampling]\label{ass:design-app}
There exist constants \(0<c_\nu\le C_\nu<\infty\) and
\(0<c_\pi\le C_\pi<\infty\), independent of \(\dt,\dm,n\), such that
for every task \(t\in[\dt]\) and every unordered pair
\(\{m,m'\}\subset[\dm]\),
\[
    \frac{c_\nu}{\dt}\le\nu_t\le\frac{C_\nu}{\dt},
    \qquad
    \frac{c_\pi}{\binom{\dm}{2}}\le\pi_t(\{m,m'\})\le\frac{C_\pi}{\binom{\dm}{2}}.
\]
\end{assumption}

In the analysis we routinely use the following two consequences of
Assumption~\ref{ass:design-app}.

\begin{lemma}[Pairwise Frobenius reduction in the matrix case]\label{lem:frob-reduction-app}
Let \(H\in\R^{\dt\times\dm}\) satisfy the row-centering condition
\(H\mathbf 1_{\dm}=0\).  Then under
Assumption~\ref{ass:design-app},
\begin{equation}\label{eq:frob-reduction-app}
    \E^\star\!\bigl[\ip{H}{X}^2\bigr]
    \asymp
    \frac{\Fnorm H^2}{\dstar},
    \qquad
    \dstar=\dt\dm.
\end{equation}
\end{lemma}

\begin{proof}
Conditional on a task \(t\) and an unordered pair \(\{m,m'\}\),
\(\ip{H}{X}=H_{t,m}-H_{t,m'}\).  Fixing the task \(t\) and letting
\(z(t)\in\R^{\dm}\) be the \(t\)-th row of \(H\), with
\(\sum_m z_m(t)=0\) by hypothesis, we use the elementary identity
\[
    \sum_{m<m'}(z_m(t)-z_{m'}(t))^2
    =
    \dm\sum_m z_m(t)^2.
\]
Under near-uniform pair sampling,
\(\E_{\{m,m'\}}[(z_m(t)-z_{m'}(t))^2]\asymp 2\norm{z(t)}_2^2/(\dm-1)\).
Averaging over \(t\sim\nu\) using \(\nu_t\asymp 1/\dt\) yields
\(\E^\star[\ip{H}{X}^2]\asymp\Fnorm H^2/(\dt(\dm-1))\asymp\Fnorm H^2/\dstar\).
\end{proof}

\begin{lemma}[Weighted second-moment for the Fisher operator]\label{lem:weighted-moment-app}
Under Assumptions~\ref{ass:score-app}--\ref{ass:design-app}, for any
row-centered \(H\),
\(c_B\Fnorm H^2/\dstar\lesssim\E^\star[I(\eta^\star)\ip{H}{X}^2]\lesssim C_B\Fnorm H^2/\dstar\).
In particular, under \(\infnorm{\Thetastar}\le B\), the operator
\(G\) restricted to the row-centered subspace \(\mathbb S
:=\{H:H\mathbf 1_{\dm}=0\}\) satisfies
\[
 \frac{c}{\dstar}\Fnorm H^2
 \le \ip{H}{GH}
 \le \frac{C}{\dstar}\Fnorm H^2,
 \qquad H\in\mathbb S.
\]
Consequently \(A_{\mathbb T}\) is positive definite on \(\mathbb T\),
with all eigenvalues between \(c/\dstar\) and \(C/\dstar\).
\end{lemma}

\subsection{Initial estimator and signal strength}\label{app:notation-init}

The following assumption records the exact output required from the initial estimator.

\begin{assumption}[Initial estimator]\label{ass:initial-app}
The initial estimator \(\widehat\Theta\) computed on an auxiliary
sample \(\mathcal D_{\rm aux}\) of size
\(n_{\rm aux}\asymp n\) satisfies the row-centering gauge
\(\widehat\Theta\mathbf 1_{\dm}=0\) and the entrywise rate
\begin{equation}\label{eq:init-entrywise-app}
    \infnorm{\widehat\Theta-\Thetastar}
    \le
    C_1\sqrt{\frac{\bar d\,\log^c\bar d}{n_{\rm aux}}}
\end{equation}
for absolute constants \(C_1,c>0\).
\end{assumption}

The estimator constructed in Appendix~\ref{app:entrywise-proof} satisfies
Assumption~\ref{ass:initial-app}.  In particular, the
nuclear-norm penalized convex initializer of
Appendix~\ref{app:entrywise-proof-rsc} produces a Frobenius-accurate
intermediate estimator, and the three-split row-wise refinement of
Appendices~\ref{app:entrywise-proof-algorithm}--\ref{app:entrywise-proof-final}
upgrades it to the entrywise rate~\eqref{eq:init-entrywise-app}.

\begin{assumption}[Signal strength]\label{ass:signal-app}
The Frobenius norm of the latent score matrix satisfies
\begin{equation}\label{eq:signal-app}
    \Fnorm{\Thetastar}\ge c_{\rm sig}\sqrt{\dstar}
\end{equation}
for an absolute constant \(c_{\rm sig}>0\).
\end{assumption}

By the rank constraint and the upper bound
\(\Fnorm{\Thetastar}^2=\sum_{i=1}^r(\sigma_i^\star)^2\le r(\sigma_1^\star)^2\)
together with \(\kappa=\sigma_1^\star/\sigma_r^\star\),
Assumption~\ref{ass:signal-app} gives
\(\sigma_r^\star\asymp\sqrt{\dstar}\) up to factors of \(r\) and
\(\kappa\), which are bounded constants.  Hence the spectral signal-to-noise
ratio scales as \(\sqrt{\dstar}\) since the noise level under the BTL
model is of constant order.

\subsection{Functional regularity for inference}\label{app:notation-functional}

For inference on a contrast \(\psi_\Gamma(\Theta)=\ip{\Gamma}{\Theta}\), we
require structural assumptions on \(\Gamma\) so that the score-gap
contrasts considered in Section~\ref{sec:ranking} are admissible.

\begin{assumption}[Bounded \(\ell_1\) gradient and finite support]\label{ass:bounded-Gamma-app}
The gradient \(\Gamma=\nabla\psi(\Thetastar)\in\R^{\dt\times\dm}\) has
\(|\mathrm{supp}(\Gamma)|\le M\) and \(\oneNorm\Gamma\le C_\psi\) for
absolute constants \(M,C_\psi>0\).
\end{assumption}

For the score-gap contrasts
\(\Gamma=e_t(e_m-e_{m'})^\top\) used in
Sections~\ref{sec:fdim}--\ref{sec:ranking},
Assumption~\ref{ass:bounded-Gamma-app} holds with
\(M=2\) and \(C_\psi=2\).

\begin{assumption}[Uniform score-gap alignment]\label{ass:alignment-app}
There exists a constant \(\alpha_0>0\), independent of
\((t,m,m')\), such that, for every \(t\in[\dt]\) and all distinct
\(m,m'\in[\dm]\),
\begin{equation}\label{eq:alignment-app}
 \Fnorm{P_{\mathbb T}\{e_t(e_m-e_{m'})^\top\}}
 \ge
 \alpha_0\sqrt{\frac1{\dt}+\frac1{\dm}}.
\end{equation}
\end{assumption}

Assumption~\ref{ass:alignment-app} ensures that every score-gap contrast
used for ranking has a non-negligible component in the correctly centered
tangent space.  It is imposed as a structural condition: incoherence gives
upper bounds on leverage scores but, by itself, does not imply this lower
bound.  The scale in~\eqref{eq:alignment-app} is natural because
Lemma~\ref{lem:tangent-envelope-app} below gives the matching upper bound
\(\Fnorm{P_{\mathbb T}\Gamma}\lesssim
\sqrt{\dt^{-1}+\dm^{-1}}\).

\subsection{Sample size and simplified rates}\label{app:notation-samplesize}

We work under the near-optimal sample-size scaling.

\begin{assumption}[Sample size]\label{ass:sample-app}
Fix an absolute exponent \(c\ge c_0\ge7\), where \(c_0\) is chosen once
and for all large enough for the concentration, studentization,
Gaussian-approximation, and polynomial-family union bounds below.
The sample size satisfies
\begin{equation}\label{eq:sample-size-app}
    n\ge C_0\,\bar d\,\log^c(n\bar d)
\end{equation}
for a sufficiently large absolute constant \(C_0>0\) depending only on
the structural parameters \((\mu,r,\kappa,B,c_\nu,C_\nu,c_\pi,C_\pi,c_{\rm sig},\alpha_0)\).
\end{assumption}

Under Assumption~\ref{ass:sample-app}, the spectral signal-to-noise
ratio \(\sigma_r(\Thetastar)/\sigma\asymp\sqrt{\dstar}\) (since
\(\sigma=O(1)\) under the BTL model) automatically dominates the noise
scale required by the subspace-perturbation theory, and all explicit
\(\sigma_r(\Thetastar)\)-dependent factors in the perturbation and
inverse-information bounds below collapse
into the structural constants.  We therefore state all rates in the
simplified form
\begin{equation}
    r_n
    :=
    \sqrt{\frac{\bar d\,\log^c(n\bar d)}{n}}
    =
    \sqrt{\frac{\bar d\,\mathrm{polylog}(n\bar d)}{n}}.
    \label{eqn:r_n}
\end{equation}

\subsection{The constant \texorpdfstring{$C_A$}{C\_A}}\label{app:notation-CA}

%\begin{revisionblock}
The entrywise regularity of the population inverse-information
operator is summarized by the following condition.

\begin{assumption}[Population inverse-information stability]
\label{ass:population-CA-app}
Let
\begin{equation}\label{eq:population-CA-app}
 C_A
 :=
 1\vee
 \frac{\norm{\mathcal A^\dagger}_{\infty\to\infty}}{\dstar}.
\end{equation}
We allow \(C_A\) to depend on the population Fisher geometry and assume
that \(C_A r_n=o(1)\).
\end{assumption}

The normalization by \(\dstar\) agrees with the Fisher-information
scale in
Lemma~\ref{lem:weighted-moment-app}: the eigenvalues of the restricted
Fisher operator are of order \(1/\dstar\).  The next proposition gives
the general range of \(C_A\) and identifies a dimension-free regime.

\begin{proposition}[Range and calibration of \(C_A\)]
\label{prop:CA-range-app}
Under Assumptions~\ref{ass:score-app}--\ref{ass:design-app}, the
bounded-signal condition, and fixed rank and incoherence constants,
\begin{equation}\label{eq:CA-range-app}
       c\ \le\ C_A\ \le\ C\sqrt{\bar d}.
\end{equation}
If task and pair sampling are exactly uniform and
\(I(\eta^\star)\equiv w_0\) is constant, then
\begin{equation}\label{eq:CA-constant-app}
 G_0H=\lambda_0H,
 \quad H\in\mathbb S,
 \qquad
 \lambda_0=\frac{2w_0}{\dt(\dm-1)},
 \qquad
 \mathcal A_0^\dagger=\lambda_0^{-1}P_{\mathbb T},
\end{equation}
so \(C_A\le C\norm{P_{\mathbb T}}_{\infty\to\infty}=O(1)\).
More generally, under the same uniform design, if
\begin{equation}\label{eq:CA-near-constant-app}
 \sup_{t,m\ne m'}
 \left|I(\Thetastar_{t,m}-\Thetastar_{t,m'})-w_0\right|
 \le \epsilon w_0
\end{equation}
for a sufficiently small dimension-independent \(\epsilon\), then
\(C_A=O(1)\).  For the logistic link, a sufficiently small
dimension-independent bound on \(\infnorm{\Thetastar}\) implies
\eqref{eq:CA-near-constant-app} with \(w_0=1/4\).
\end{proposition}

\begin{proof}
The spectral bounds in Lemma~\ref{lem:weighted-moment-app} imply
\(\norm{A_{\mathbb T}^{-1}}_{\rm op}\asymp\dstar\).  Since an induced
matrix norm dominates the spectral radius,
\(\norm{\mathcal A^\dagger}_{\infty\to\infty}\ge c\dstar\), proving
the lower bound in~\eqref{eq:CA-range-app}.  For a canonical ambient
matrix atom \(E_\omega\), the closed form
\eqref{eq:tangent-projector-matrix} and incoherence give the tangent
leverage bound
\[
 \Fnorm{P_{\mathbb T}E_\omega}^2
 \le C\frac{\bar d}{\dstar}.
\]
Self-adjointness and Cauchy--Schwarz therefore yield
\[
 \begin{split}
 \norm{\mathcal A^\dagger}_{\infty\to\infty}
 &=\max_\omega\oneNorm{\mathcal A^\dagger E_\omega}\\
 &\le \sqrt{\dstar}\max_\omega
       \Fnorm{A_{\mathbb T}^{-1}P_{\mathbb T}E_\omega}
 \le C\sqrt{\dstar}\,\dstar
       \sqrt{\frac{\bar d}{\dstar}}
 =C\dstar\sqrt{\bar d},
 \end{split}
\]
which proves the upper bound.

Under constant weights, the pairwise-difference identity gives
\eqref{eq:CA-constant-app}.  Moreover,
\(\norm{P_{\mathbb T}}_{\infty\to\infty}\le C\): indeed, by
self-adjointness its induced norm is
\(\max_\omega\oneNorm{P_{\mathbb T}E_\omega}\), and substituting the
incoherent factor rows into~\eqref{eq:tangent-projector-matrix} bounds
this quantity by a fixed-rank constant.  This proves the constant-weight
claim.

For the perturbative claim write \(G=G_0+E\) and
\(A_{\mathbb T}=\lambda_0 I_{\mathbb T}+P_{\mathbb T}EP_{\mathbb T}\).
Condition~\eqref{eq:CA-near-constant-app}, the pairwise graph form of
the information operator, and the preceding projector bound give
\[
 \norm{\lambda_0^{-1}P_{\mathbb T}EP_{\mathbb T}}
       _{\infty\to\infty}
 \le C\epsilon.
\]
Choosing \(\epsilon\) so that the right side is at most \(1/2\), the
Neumann expansion gives
\[
 \mathcal A^\dagger
 =\left(I+\lambda_0^{-1}P_{\mathbb T}EP_{\mathbb T}\right)^{-1}
   \lambda_0^{-1}P_{\mathbb T},
 \qquad
 \norm{\mathcal A^\dagger}_{\infty\to\infty}\le C\dstar.
\]
Finally, \(I(x)=1/4+O(x^2)\) near zero for the logistic link, so a
sufficiently small uniform signal bound implies
\eqref{eq:CA-near-constant-app}.
\end{proof}

The pairwise graph representation and near-uniform
sampling give \(\norm G_{\infty\to\infty}\le C/\dstar\), while the
projector calculation above gives
\(\norm{P_{\mathbb T}}_{\infty\to\infty}\le C\).  Hence
\begin{equation}\label{eq:population-composite-CA-app}
 \norm{\mathcal A^\dagger P_{\mathbb T}G}_{\infty\to\infty}
 \le C C_A.
\end{equation}
Foldwise analogues are established in
Theorem~\ref{thm:derived-foldwise-CA-app}.
%\end{revisionblock}

\subsection{Probability calibration to \texorpdfstring{$1-n^{-a}$}{1-n^-a}}\label{app:notation-probcalib}

Throughout, all "high-probability" statements are uniform over a free
constant \(a>0\) that may be taken arbitrarily large at the cost of an
absolute constant prefactor absorbed into \(\mathrm{polylog}(n\bar d)\).
Specifically, every high-probability bound below has the form
\(\Pr(\text{good event})\ge 1-n^{-a}\) for any fixed \(a>0\).

\textbf{Origin of the calibration.}
Every concentration statement below is proved with a free tail parameter
\(x\).  We set \(x=(a+c_0)\log(n\bar d)\), where \(c_0\) covers the
displayed finite union bounds.  Assumption~\ref{ass:sample-app} is stated
with the same logarithm, so no probability statement from another source
is silently strengthened.

\textbf{The master good event \(\mathcal E_n\).}
Let \(a>0\) be any fixed (large) constant.  Define
\[
    \mathcal E_n
    :=
    \mathcal E_{\rm init}\cap
    \mathcal E_{\rm refine}\cap
    \mathcal E_{\rm rem}\cap
    \mathcal E_{\rm var}\cap
    \mathcal E_{\rm cov}\cap
    \mathcal E_{\rm re}\cap
    \mathcal E_A,
\]
where the constituent events are as follows.
\begin{enumerate}[label=(\arabic*),leftmargin=2.4em,topsep=2pt,itemsep=2pt]
\item \(\mathcal E_{\rm init}\) --- Frobenius accuracy of the convex
initializer (Theorem~\ref{thm:convex-main-app}).
\item \(\mathcal E_{\rm refine}\) --- entrywise accuracy of the refined
estimator (Theorem~\ref{thm:pairwise-max-app}); this implies
\(\mathcal E_\infty=\{\infnorm{\Thetahat-\Thetastar}\le\varepsilon_n\}\)
with \(\varepsilon_n=Cr_n\).
\item \(\mathcal E_{\rm rem}\) --- uniform single-contrast one-step
remainder (Theorem~\ref{thm:uniform-remainder-app}).
\item \(\mathcal E_{\rm var}\) --- relative variance consistency for
plug-in standard errors
(Proposition~\ref{prop:fdim-var-cons-app}).
\item \(\mathcal E_{\rm cov}\) --- entrywise covariance consistency
(Proposition~\ref{prop:fdim-cov-cons-app}).
\item \(\mathcal E_{\rm re}\) --- the foldwise tangent geometry
and information-concentration event
(Lemma~\ref{lem:foldwise-geometry-app}).
\item {\(\mathcal E_A\) --- the foldwise inverse-
information and relative plug-in perturbation bounds in
Theorem~\ref{thm:derived-foldwise-CA-app}.}
\end{enumerate}
Each event holds with probability at least \(1-n^{-a}\); by the union
bound,
\[
    \Pr(\mathcal E_n^c)
    \le
    7\,n^{-a}
    \le
    n^{-a/2},
\]
and we relabel \(a\) so that \(\Pr(\mathcal E_n)\ge 1-n^{-a}\).
All deterministic statements in the proofs are made on \(\mathcal E_n\);
the deficit \(\Pr(\mathcal E_n^c)=o(1)\) is absorbed into the
\(o(1)\) slack of every coverage statement in
Appendix~\ref{app:ranking-proof}.

\subsection{Rectangular dimensions}\label{app:notation-balanced}

No balance condition on \((\dt,\dm)\) is imposed.  In the transposed
refinement below, a model-row update has curvature
\(n/(\dt\dm)\) and uniform row error \(\sqrt{\dt\bar d/n}\), whereas a task-row of
the orthonormal factor has norm \(O(\dt^{-1/2})\).  The task-column
update has curvature \(n\) and error \(n^{-1/2}\).  Consequently its two
entrywise reconstruction terms are, respectively,
\(O(\sqrt{\bar d/n})\) and \(O(\sqrt{\dt/n})\), giving the stated
\(O(\sqrt{\bar d/n})\) rate.  Hidden constants depend
only on the structural parameters \((\mu,r,\kappa,B,c_B,C_B)\) of the
score and signal, the design constants
\((c_\nu,C_\nu,c_\pi,C_\pi)\), the alignment \(\alpha_0\), the
signal-strength constant \(c_{\rm sig}\), and the
probability-calibration constant \(a\).

\section{Proof of Theorem~\ref{thm:entrywise}: convex initializer and entrywise refinement}\label{app:entrywise-proof}

This appendix proves Theorem~\ref{thm:entrywise} of
Section~\ref{sec:estimation} via the two-stage estimator
\(\Thetahat=\mathsf{Refine}_r(\widehat\Theta_0)\) constructed from a
nuclear-norm penalized convex MLE \(\widehat\Theta_0\) (initializer)
followed by a three-split row-wise refinement.  The structure of the
appendix is as follows.

\begin{itemize}[leftmargin=1.6em,topsep=2pt,itemsep=3pt]
\item Appendix~\ref{app:entrywise-proof-rsc} establishes the convex
stage.  We adapt the restricted strong convexity (RSC) framework of
\cite{negahban2012restricted} to the pairwise logistic loss,
obtaining a Frobenius-accurate initializer
\(\widehat\Theta_0\) with rate
\(\Fnorm{\widehat\Theta_0-\Thetastar}\lesssim\sqrt{r\,\dt\,\dm\,\bar d\,\log\bar d/n}\)
under the row-centering identifiability constraint.
\item Appendix~\ref{app:entrywise-proof-algorithm} sets up the
three-split refinement algorithm and the proof roadmap (six blocks).
\item Appendix~\ref{app:entrywise-proof-brouwer} establishes the
Brouwer inward-pointing zero lemma which underlies the
deterministic existence steps.
\item Appendix~\ref{app:entrywise-proof-leftdet} (Block~II) gives the
deterministic local-existence statement for the row update.
\item Appendix~\ref{app:entrywise-proof-leftdet} (Block~III) verifies the
existence condition probabilistically: six concentration lemmas
(coverage, noise, two bias terms, higher-order moments, curvature)
and a uniform proposition delivering
\(\twoninfnorm{\widetilde L-L^\star}\lesssim
\sqrt{\dt\bar d/n_2}\,\mathrm{polylog}(n\bar d)\).
\item Appendix~\ref{app:entrywise-proof-recenter} (Block~IV) handles
the gauge re-centering and proves the pairwise Gram identity used to
obtain the task-column curvature \(\lambda_t\asymp n_3\).
\item Appendix~\ref{app:entrywise-proof-right} (Blocks~V and VI)
deals with the right factor, yielding
the bound in~\eqref{eq:task-factor-rate-app} for
\(\twoninfnorm{\widetilde A-A^\star}\).
\item Appendix~\ref{app:entrywise-proof-final} (final assembly)
combines the two factor bounds into the entrywise rate
\(\infnorm{\Thetahat-\Thetastar}\lesssim\sqrt{\bar d\,\mathrm{polylog}(n\bar d)/n}\),
proving Theorem~\ref{thm:entrywise}.
\end{itemize}

\subsection{Stage 1: convex initializer via pairwise-logistic RSC}\label{app:entrywise-proof-rsc}

We split the comparison sample into three independent parts
\(\mathcal D_1,\mathcal D_2,\mathcal D_3\) of sizes
\(n_1,n_2,n_3\asymp n\).  The convex stage uses \(\mathcal D_1\) only.
For convenience write \(M^\star:=\Thetastar\), \(d_1:=\dt\), and \(d_2:=\dm\)
in this subsection to keep the two matrix modes explicit.

\subsubsection{Estimator}\label{app:rsc-estimator}

Let \(\ell(y,\eta)=\log(1+e^\eta)-y\eta\) be the logistic negative
log-likelihood and define the empirical risk
\[
    \cL_{n_1}(M)
    :=
    \frac{1}{n_1}\sum_{i\in\mathcal D_1}
    \ell\bigl(Y_i,\ip{X_i}{M}\bigr).
\]
The convex estimator is the nuclear-norm penalized MLE under the
row-centering and entrywise-bound constraints:
\begin{equation}\label{eq:penalized-program-app}
    \widehat M\in
    \arg\min_{M\in\mathcal C_B}
    \bigl\{\cL_{n_1}(M)+\lambda\norm M_*\bigr\},
    \qquad
    \mathcal C_B
    :=
    \{M\in\R^{d_1\times d_2}:\,M\mathbf 1_{d_2}=0,\,\infnorm M\le B\}.
\end{equation}
Since \(M^\star\in\mathcal C_B\) by Assumption~\ref{ass:score-app}(ii),
the truth is feasible.  The choice of \(\lambda\) is fixed in
Theorem~\ref{thm:convex-main-app} below.  We retain the bounded,
row-centered convex solution \(\widehat M\) as the offset estimator in
the refinement.  Its rank-\(r\) SVD is used only to estimate singular
subspaces.  No entrywise clipping is applied after rank truncation,
because such clipping need not preserve rank or row centering.

The role of the row-centering constraint is essential: pairwise
comparisons depend only on within-task differences, so any matrix of
the form \(c\mathbf 1_{d_2}^\top\), \(c\in\R^{d_1}\), lies in the null
space of the design.
Without the constraint, no RSC statement can hold.

\subsubsection{Step 1: the population pairwise quadratic identity}\label{app:rsc-popid}

We begin with the one genuinely pairwise-specific algebraic identity.

\begin{lemma}[Population pairwise quadratic identity]\label{lem:population-isometry-app}
Let \(X=e_t(e_m-e_{m'})^\top\) where \(t\in[d_1]\) is sampled from
\(\nu\) and \(\{m,m'\}\subset[d_2]\) from a task-dependent distribution
\(\pi_t\) satisfying Assumption~\ref{ass:design-app}.  For every
\(\Delta\in\R^{d_1\times d_2}\) with \(\Delta\mathbf 1_{d_2}=0\),
\begin{equation}\label{eq:population-isometry-app}
    \E\bigl[\ip{X}{\Delta}^2\bigr]
    \asymp
    \frac{2}{d_1\,(d_2-1)}\Fnorm\Delta^2
    \asymp
    \frac{\Fnorm\Delta^2}{d_1\,d_2}.
\end{equation}
\end{lemma}

\begin{proof}
Conditional on \(t\), \(\ip{X}{\Delta}=\Delta_{t,m}-\Delta_{t,m'}\).
Letting \(z(t)\in\R^{d_2}\) denote the \(t\)-th row of \(\Delta\), with
\(\sum_m z_m(t)=0\) by hypothesis, the standard sum-of-squared-pairwise-differences identity gives
\[
    \sum_{m<m'}(z_m(t)-z_{m'}(t))^2
    =
    d_2\sum_{m=1}^{d_2}z_m(t)^2.
\]
Under near-uniform pair sampling
(Assumption~\ref{ass:design-app}), the conditional expectation
satisfies
\(\E_{\{m,m'\}}[(z_m(t)-z_{m'}(t))^2\mid t]\asymp 2\norm{z(t)}_2^2/(d_2-1)\).
Averaging over \(t\sim\nu\) using \(\nu_t\asymp 1/d_1\) and summing
\(\norm{z(t)}_2^2\) over \(t\) yields
\(\E[\ip{X}{\Delta}^2]\asymp 2\Fnorm\Delta^2/(d_1(d_2-1))\).
\end{proof}

The identity~\eqref{eq:population-isometry-app} is the pairwise
analogue of the Negahban--Wainwright population norm-equivalence
identity for matrix completion.  The centered gauge
\(\Delta\mathbf 1_{d_2}=0\) is essential: \(\ip{X}{\Delta}=0\) identically
for any \(\Delta=c\mathbf 1_{d_2}^\top\), so without centering the left-hand
side of~\eqref{eq:population-isometry-app} would vanish on a non-trivial
subspace.

\subsubsection{Step 2: pairwise quadratic restricted strong convexity}\label{app:rsc-rsc}

Define the spikiness ratio
\(\alpha_{\rm sp}(\Delta):=\sqrt{d_1 d_2}\,\infnorm\Delta/\Fnorm\Delta\)
and the rank surrogate
\(\beta_{\rm ra}(\Delta):=\norm\Delta_*/\Fnorm\Delta\), and let
\[
    \mathcal C_{\rm pw}(n_1;c_0)
    :=
    \Bigl\{\Delta\in\R^{d_1\times d_2}:\,\Delta\mathbf 1_{d_2}=0,\;
    \alpha_{\rm sp}(\Delta)\,\beta_{\rm ra}(\Delta)
    \le \frac{1}{c_0}\sqrt{\frac{n_1}{\bar d\log(n\bar d)}}\Bigr\}.
\]

\begin{theorem}[Pairwise quadratic RSC]\label{thm:pairwise-quadratic-rsc-app}
Fix any \(a>0\).  There exist absolute constants \(c_0,c,C>0\) such that
whenever
\(n_1\ge C\bar d\log(n\bar d)\),
with probability at least \(1-n^{-a}\),
\begin{equation}\label{eq:pairwise-rsc-app}
    \frac{1}{n_1}\sum_{i\in\mathcal D_1}\ip{X_i}{\Delta}^2
    \;\ge\;
    \frac{c}{d_1\,d_2}\Fnorm\Delta^2
    \qquad\text{for all }\Delta\in\mathcal C_{\rm pw}(n_1;c_0).
\end{equation}
\end{theorem}

\begin{proof}
Set \(D=d_1d_2\), \(d_{\min}=d_1\wedge d_2\), and
\(Q_n(\Delta)=n_1^{-1}\sum_i\ip{X_i}{\Delta}^2\).
Both sides are homogeneous, so fix \(\Fnorm\Delta=1\).  For
\[
 \mathcal S(\alpha,R)=
 \{\Delta:\Delta\mathbf1_{d_2}=0,\ \Fnorm\Delta=1,\
 \sqrt D\,\infnorm\Delta\le\alpha,\ \norm\Delta_*\le R\},
\]
Lemma~\ref{lem:population-isometry-app} gives
\(c/D\le Q(\Delta):=\E Q_n(\Delta)\le C/D\).

Let \(Z_\varepsilon=n_1^{-1}\sum_i\varepsilon_iX_i\).  Symmetrization,
contraction on \(|\ip{X_i}{\Delta}|\le2\alpha/\sqrt D\), and
operator/nuclear duality give
\[
 \E\sup_{\mathcal S(\alpha,R)}|Q_n-Q|
 \le \frac{8\alpha R}{\sqrt D}\E\norm{Z_\varepsilon}_{\rm op}.
\]
Here \(\norm{X_i}_{\rm op}\le\sqrt2\), and near-uniform sampling gives
\[
 \max\{\norm{\E X_iX_i^\top}_{\rm op},
        \norm{\E X_i^\top X_i}_{\rm op}\}\le C/d_{\min}.
\]
Rectangular matrix Bernstein, integrated over its tail, therefore yields
\[
 \E\norm{Z_\varepsilon}_{\rm op}
 \le C\left\{\sqrt{\frac{\log(d_1+d_2)}{n_1d_{\min}}}
             +\frac{\log(d_1+d_2)}{n_1}\right\}.
\]

Bousquet's inequality supplies the required tail bound directly for
this empirical process.  Indeed, for
\(f_\Delta(X)=\ip{X}{\Delta}^2\),
\(0\le f_\Delta\le4\alpha^2/D\) and
\(\sup_\Delta\Var(f_\Delta)\le C\alpha^2/D^2\).  Thus, with probability
at least \(1-2e^{-x}\),
\begin{align}
 \sup_{\mathcal S(\alpha,R)}|Q_n-Q|
 \le C\bigg[
 &\frac{\alpha R}{\sqrt D}
   \sqrt{\frac{\log(d_1+d_2)}{n_1d_{\min}}}
 +\frac{\alpha}{D}\sqrt{\frac{x}{n_1}}
 +\frac{\alpha^2x}{Dn_1}\bigg].
 \label{eq:rsc-local-process-app}
\end{align}
This argument is on the independent matrices \(X_i\); it does not use
a product-space Lipschitz assertion for dependent transformed
coordinates.

Apply~\eqref{eq:rsc-local-process-app} on dyadic \((\alpha,R)\) shells
and take \(x=(a+C)\log(n\bar d)\).  Feasibility bounds the number of
nontrivial shells by \(O(\log(nD))\).  If
\[
 \alpha_{\rm sp}(\Delta)\beta_{\rm ra}(\Delta)
 \le c_0^{-1}\sqrt{\frac{n_1}{\bar d\log(n\bar d)}},
\]
then \(D/d_{\min}\le\bar d\) and
\(\beta_{\rm ra}(\Delta)\ge1\) make the right side of
\eqref{eq:rsc-local-process-app} at most \(c/(2D)\), for sufficiently
large \(c_0\) and sample-size constant.  Hence
\(Q_n(\Delta)\ge c\Fnorm\Delta^2/(2D)\).
Rescaling proves~\eqref{eq:pairwise-rsc-app}.  When the restricted
inequality fails, its negation and the nuclear-norm cone below give the
small-radius alternative used in Theorem~\ref{thm:convex-main-app}.
\end{proof}
\subsubsection{Step 3: from quadratic RSC to logistic RSC}\label{app:rsc-logistic}

The pairwise logistic loss is not quadratic, but its Bregman divergence
inherits the quadratic curvature on the bounded-signal feasible set.

\begin{lemma}[Logistic curvature reduction]\label{lem:likelihood-curvature-app}
For any \(\Delta\) with \(M^\star+\Delta\in\mathcal C_B\),
\begin{equation}\label{eq:bregman-lower-bound-app}
    \delta\cL_{n_1}(M^\star;\Delta)
    :=
    \cL_{n_1}(M^\star+\Delta)-\cL_{n_1}(M^\star)-\ip{\nabla\cL_{n_1}(M^\star)}{\Delta}
    \;\ge\;
    \frac{c_B}{2 n_1}\sum_{i\in\mathcal D_1}\ip{X_i}{\Delta}^2,
\end{equation}
where \(c_B:=\inf_{|x|\le 2B}\sigma'(x)>0\).  Consequently, on the event
of Theorem~\ref{thm:pairwise-quadratic-rsc-app},
\[
    \delta\cL_{n_1}(M^\star;\Delta)
    \;\ge\;
    \kappa_{\rm pw}\Fnorm\Delta^2,
    \qquad\kappa_{\rm pw}\asymp\frac{1}{d_1 d_2},
\]
holds for all \(\Delta\in\mathcal C_{\rm pw}(n_1;c_0)\) with
\(M^\star+\Delta\in\mathcal C_B\).
\end{lemma}

\begin{proof}
Taylor's theorem applied to \(\eta\mapsto\ell(y,\eta)\) at
\(\ip{X_i}{M^\star}\) gives, for some intermediate
\(\xi_i\) with \(|\xi_i|\le 2B\),
\[
    \delta\cL_{n_1}(M^\star;\Delta)
    =
    \frac{1}{2 n_1}\sum_{i\in\mathcal D_1}\sigma'(\xi_i)\ip{X_i}{\Delta}^2.
\]
By definition of \(c_B\), \(\sigma'(\xi_i)\ge c_B>0\), giving the first
inequality.  Combining with~\eqref{eq:pairwise-rsc-app} yields the
second.
\end{proof}

\subsubsection{Step 4: gradient operator-norm bound}\label{app:rsc-gradient}

\begin{lemma}[Gradient operator-norm bound]\label{lem:gradient-bound-app}
Fix any \(a>0\).  With probability at least \(1-n^{-a}\),
\begin{equation}\label{eq:gradient-bound-app}
    \norm{\nabla\cL_{n_1}(M^\star)}_{\rm op}
    \le
    C\sqrt{\frac{\log\{n(d_1+d_2)\}}{n_1\,(d_1\wedge d_2)}}.
\end{equation}
\end{lemma}

\begin{proof}
The gradient at the truth is
\(\nabla\cL_{n_1}(M^\star)=n_1^{-1}\sum_i(\sigma(\ip{X_i}{M^\star})-Y_i)X_i\).
Each summand is mean zero (by the model) and has operator norm at most
\(\sqrt 2\) (since \(\norm{X_i}_{\rm op}=\sqrt 2\) and the scalar prefactor
\(\sigma(\ip{X_i}{M^\star})-Y_i\in[-1,1]\)).  The matrix variance proxy
on the right is
\[
    \E[(\sigma(\ip{X_i}{M^\star})-Y_i)^2 X_i X_i^\top]
    \preceq
    \E[X_i X_i^\top]
    =2\,\mathrm{diag}(\nu)
    \preceq \frac{C}{d_1}I_{d_1};
\]
the other variance is
\(\E[X_i^\top X_i]\preceq C d_2^{-1}R_{\rm m}\).
Hence \(\sigma_X^2:=\max(\norm{\sum_i\E[X_i X_i^\top]}_{\rm op},\norm{\sum_i\E[X_i^\top X_i]}_{\rm op})\le C n_1/(d_1\wedge d_2)\).
By the rectangular matrix Bernstein inequality (Tropp 2015, Theorem 6.1.1)
applied to the rescaled sum,
\[
    \Pr\Bigl[\norm{n_1^{-1}\textstyle\sum_i(\sigma-Y)X_i}_{\rm op}\ge t\Bigr]
    \le
    (d_1+d_2)\exp\Bigl(-\frac{n_1\,t^2/2}{C/(d_1\wedge d_2)+\sqrt 2 t/3}\Bigr).
\]
Setting \(t=C\sqrt{a\log\{n(d_1+d_2)\}\,n_1^{-1}/(d_1\wedge d_2)}\) for \(C\)
large enough yields the claimed bound with probability \(1-n^{-a}\).
\end{proof}

\subsubsection{Step 5: main convex initialization theorem}\label{app:rsc-main}

\begin{theorem}[Frobenius bound for the convex initializer]\label{thm:convex-main-app}
Fix any \(a>0\).  Under the model assumptions of
Section~\ref{sec:setup} and Assumption~\ref{ass:design-app}, set
\(\lambda:=2C\sqrt{\log\{n(d_1+d_2)\}/(n_1\,(d_1\wedge d_2))}\) where \(C\) is the
constant from Lemma~\ref{lem:gradient-bound-app}, and assume
\(n_1\ge C\,\mathrm{poly}(\mu,r,\kappa,B)\,\bar d\log^c(n\bar d)\).
Then with probability at least \(1-n^{-a}\),
\begin{equation}\label{eq:convex-main-frob-app}
    \Fnorm{\widehat M-M^\star}
    \;\le\;
    \frac{C\,\lambda\sqrt r}{\kappa_{\rm pw}}
    \;\le\;
    C'\sqrt{\frac{r\,d_1\,d_2\,\bar d\,\log(n\bar d)}{n_1}}.
\end{equation}
Consequently, writing \(\widehat M_r=\mathrm{SVD}_r(\widehat M)\),
\(\Fnorm{\widehat M_r-M^\star}\le2\Fnorm{\widehat M-M^\star}\).
The bounded matrix \(\widehat M\) is retained for offsets and
\(\widehat M_r\) is used only for its singular subspaces.
\end{theorem}

\begin{proof}
Set \(\widehat\Delta:=\widehat M-M^\star\) and let \(\mathbb T_r\) denote
the rank-\(r\) tangent space at \(M^\star\) (without the row-centering
constraint), with associated decomposition
\(\widehat\Delta=\widehat\Delta_{\mathbb T_r}+\widehat\Delta_{\mathbb T_r^\perp}\).

\textbf{Basic inequality.}
Since \(M^\star\in\mathcal C_B\) is feasible, optimality of
\(\widehat M\) gives
\[
    \cL_{n_1}(M^\star+\widehat\Delta)+\lambda\norm{M^\star+\widehat\Delta}_*
    \le
    \cL_{n_1}(M^\star)+\lambda\norm{M^\star}_*.
\]
Rearranging,
\[
    \delta\cL_{n_1}(M^\star;\widehat\Delta)
    \le
    -\ip{\nabla\cL_{n_1}(M^\star)}{\widehat\Delta}
    +\lambda\bigl(\norm{M^\star}_*-\norm{M^\star+\widehat\Delta}_*\bigr).
\]
By the operator/nuclear-norm duality and our choice
\(\lambda\ge 2\norm{\nabla\cL_{n_1}(M^\star)}_{\rm op}\) (which holds on
the event of Lemma~\ref{lem:gradient-bound-app}),
\(|\ip{\nabla\cL_{n_1}(M^\star)}{\widehat\Delta}|\le(\lambda/2)\norm{\widehat\Delta}_*\).
By the standard nuclear-norm decomposability,
\(\norm{M^\star}_*-\norm{M^\star+\widehat\Delta}_*\le\norm{\widehat\Delta_{\mathbb T_r}}_*-\norm{\widehat\Delta_{\mathbb T_r^\perp}}_*\).
Combining,
\[
    \delta\cL_{n_1}(M^\star;\widehat\Delta)
    \le
    \frac{3\lambda}{2}\norm{\widehat\Delta_{\mathbb T_r}}_*
    -\frac{\lambda}{2}\norm{\widehat\Delta_{\mathbb T_r^\perp}}_*.
\]
Since \(\delta\cL_{n_1}\ge 0\), this forces the cone condition
\(\norm{\widehat\Delta_{\mathbb T_r^\perp}}_*\le 3\norm{\widehat\Delta_{\mathbb T_r}}_*\),
and using \(\rank(\widehat\Delta_{\mathbb T_r})\le 2r\) we obtain
\(\norm{\widehat\Delta_{\mathbb T_r}}_*\le\sqrt{2r}\Fnorm{\widehat\Delta}\)
and hence
\(\norm{\widehat\Delta}_*\le 4\sqrt{2r}\,\Fnorm{\widehat\Delta}\).

\textbf{Applying RSC: two cases.}
If \(\widehat\Delta\in\mathcal C_{\rm pw}(n_1;c_0)\), then on the event of
Theorem~\ref{thm:pairwise-quadratic-rsc-app},
\[
    \delta\cL_{n_1}(M^\star;\widehat\Delta)
    \ge
    \kappa_{\rm pw}\Fnorm{\widehat\Delta}^2,
    \qquad
    \kappa_{\rm pw}\asymp \frac{c_B}{d_1 d_2}.
\]

\textbf{Combining.}
On the same event,
\[
    \kappa_{\rm pw}\Fnorm{\widehat\Delta}^2
    \le
    \delta\cL_{n_1}(M^\star;\widehat\Delta)
    \le
    \frac{3\lambda}{2}\norm{\widehat\Delta_{\mathbb T_r}}_*
    \le
    \frac{3\lambda\sqrt{2r}}{2}\Fnorm{\widehat\Delta}.
\]
Dividing by \(\Fnorm{\widehat\Delta}\),
\[
    \Fnorm{\widehat\Delta}
    \le
    \frac{3\lambda\sqrt{2r}}{2\kappa_{\rm pw}}
    \le
    C\,\lambda\sqrt r\,d_1\,d_2.
\]
Substituting
\(\lambda\asymp\sqrt{\log(n\bar d)/(n_1(d_1\wedge d_2))}\) and using
\((d_1d_2)^2/(d_1\wedge d_2)=d_1d_2\bar d\) gives
\[
 \Fnorm{\widehat\Delta}
 \le C\sqrt{\frac{r\,d_1d_2\bar d\log(n\bar d)}{n_1}}.
\]

If on the other hand \(\widehat\Delta\notin\mathcal C_{\rm pw}(n_1;c_0)\),
then the negation of its defining inequality, the cone bound
\(\beta_{\rm ra}(\widehat\Delta)\le4\sqrt{2r}\), and feasibility
\(\infnorm{\widehat\Delta}\le2B\) give
\[
 {2B\sqrt{d_1d_2}\over\Fnorm{\widehat\Delta}}
 4\sqrt{2r}
 \ge \alpha_{\rm sp}(\widehat\Delta)
       \beta_{\rm ra}(\widehat\Delta)
 >{1\over c_0}\sqrt{\frac{n_1}{\bar d\log(n\bar d)}}.
\]
Rearranging yields
\[
 \Fnorm{\widehat\Delta}
 \le C(B,r)\sqrt{\frac{d_1d_2\bar d\log(n\bar d)}{n_1}}.
\]
Thus the same displayed Frobenius rate holds in both cases.

The post-processing claim follows because rank-\(r\) SVD truncation
increases Frobenius distance to the rank-\(r\) truth by at most a factor
of two.  No entrywise post-processing is used.
\end{proof}

\subsection{Stage 2: three-split refinement algorithm and roadmap}\label{app:entrywise-proof-algorithm}

\subsubsection{Algorithm}\label{app:refine-algorithm}

We describe the refinement in the orientation in which comparisons act
on rows.  Transpose the signal and write
\[
 M^\star:=(\Thetastar)^\top=L^\star(A^\star)^\top,
 \qquad L^\star\in\R^{\dm\times r},\quad
 A^\star\in\R^{\dt\times r},\quad
 (A^\star)^\top A^\star=I_r,\quad
 \mathbf 1_{\dm}^\top L^\star=0.
\]
Thus model indices are rows and task indices are columns; the singular
values are absorbed into \(L^\star\).  Let
\(R_{\rm m}=I_{\dm}-\dm^{-1}\mathbf1_{\dm}\mathbf1_{\dm}^\top\).

\textbf{Stage~A: initialization and task-factor construction.}
On \(\mathcal D_1\), compute the bounded, row-centered convex estimator
\(\widehat\Theta_0\) of Theorem~\ref{thm:convex-main-app} and set
\(\widehat M_0=\widehat\Theta_0^\top\).  Let \(Q_0\) contain its top-\(r\)
right singular vectors.  Project \(Q_0\) onto the rotation-invariant set
\(\{A:\twoninfnorm A\le C_A'\dt^{-1/2},\ A^\top A=I_r\}\) to obtain
\(\widehat A\).  The untruncated bounded matrix \(\widehat M_0\) is used
for offsets.

\textbf{Stage~B: model-row refinement.}
On \(\mathcal D_2\), fix \(\widehat A\).  For each model \(u\in[\dm]\),
reorient every comparison involving \(u\) so that \(u\) is first and
solve
\[
    S_u(\ell)
    :=
    \sum_{q=1}^{M_u}\widehat A[t_q]
    \{Z_q-\sigma(\widehat A[t_q]^\top\ell-\widehat o_q)\}=0,
    \qquad \widehat o_q:=\widehat M_{0,w_q,t_q},
\]
where \(t_q\) is the task and \(w_q\) the opponent model.  Stack the
solutions as \(\widetilde L\in\R^{\dm\times r}\).

\textbf{Stage~C: model centering and task-column refinement.}
Set \(\overline L=R_{\rm m}\widetilde L\).  This restores
\(\mathbf1_{\dm}^\top\overline L=0\) and preserves every model-row
difference.  For each task \(t\), use the observations in \(\mathcal D_3\)
having task \(t\), put
\(x_i=\overline L[m_i]-\overline L[m_i']\), and solve
\[
    S_t(a)
    :=
    \sum_{i\in\mathcal I_t}x_i\{Y_i-\sigma(x_i^\top a)\}=0,
    \qquad \mathcal I_t=\{i\in\mathcal D_3:t_i=t\}.
\]
Stack the solutions as \(\widetilde A\in\R^{\dt\times r}\).

\textbf{Final estimator.}
Set \(\widetilde M=\overline L\widetilde A^\top\) and
\(\Thetahat=\widetilde M^\top\).  It satisfies
\(\Thetahat\mathbf1_{\dm}=0\) exactly; the final assembly below shows
that it has exactly rank \(r\) on the same high-probability event.

\subsubsection{Roadmap}\label{app:refine-roadmap}

The proof proceeds in six self-contained blocks, each corresponding to one operation
of Algorithm~\ref{alg:ranking-recovery}.

\begin{enumerate}[leftmargin=2em,itemsep=2pt,topsep=2pt,label=\textbf{Block~\Roman*:}]
\item \emph{Frobenius-error transfer to the task factor.}
After the orthogonal Procrustes alignment,
Wedin's theorem gives
\(\Fnorm{\widehat A-A^\star}\lesssim
\Fnorm{\widehat M_0-M^\star}/\sigma_r(M^\star)\), while the projection
step gives \(\twoninfnorm{\widehat A}\lesssim\dt^{-1/2}\) and
\(\widehat A^\top\widehat A=I_r\).
\item \emph{Deterministic local existence for the row update.}
We prove an inward-pointing zero-existence lemma for \(S_u\)
(Lemma~\ref{lem:left-deterministic-app}) which states that under a
sufficient condition on a "noise plus bias plus higher-order" combination
\(R_u\le\lambda_u^2/(4L_3\gamma_u)\) for a suitable curvature
\(\lambda_u\), the score equation has a solution close to the truth.
\item \emph{Probabilistic verification of the row condition.}
On \(\mathcal D_2\), we verify the sufficient condition uniformly over
\(u\in[\dm]\) using six concentration lemmas: row-wise coverage,
noise envelope, two bias terms, higher-order moments \(\beta_u,\gamma_u\),
and curvature \(\lambda_u\).  This gives
\(\twoninfnorm{\widetilde L-L^\star}\lesssim
\sqrt{\dt\bar d/n_2}\,\mathrm{polylog}(n\bar d)\).
\item \emph{Re-centering and pairwise Gram identity.}
The model-centering projection \(R_{\rm m}\) preserves pairwise differences and
preserves entrywise / Frobenius distance (Lemma~\ref{lem:centering-app}).
A pairwise Gram identity (Lemma~\ref{lem:pairwise-gram-app}) then
upgrades the column-update curvature to
\(\lambda_t\gtrsim n_3\) by capturing the gain from sampling pairs
on the row factor.
\item \emph{Deterministic local existence for the column update.}
The same Brouwer argument (Lemma~\ref{lem:right-deterministic-app})
gives a sufficient condition for the column score equation
\(S_t(a)=0\) to have a solution near the truth.
\item \emph{Probabilistic verification of the column condition.}
We verify uniformly in \(t\in[\dt]\) using parallel concentration
lemmas, obtaining
\(\twoninfnorm{\widetilde A-A^\star}\lesssim n^{-1/2}\,\mathrm{polylog}(n\bar d)\).
\end{enumerate}

\subsection{Block I: Brouwer inward-pointing zero lemma}\label{app:entrywise-proof-brouwer}

We restate the deterministic existence lemma underlying Blocks~II and V.

\begin{lemma}[Brouwer inward-pointing zero]\label{lem:inward-pointing-app}
Let \(F:\R^r\to\R^r\) be continuous, fix \(\vartheta^\star\in\R^r\) and
\(\xi>0\).  If
\((\vartheta-\vartheta^\star)^\top F(\vartheta)\le 0\)
for every \(\vartheta\) on the sphere
\(\{\vartheta:\norm{\vartheta-\vartheta^\star}=\xi\}\), then there
exists \(\widetilde\vartheta\) with \(F(\widetilde\vartheta)=0\) and
\(\norm{\widetilde\vartheta-\vartheta^\star}\le\xi\).
\end{lemma}

\begin{proof}
Suppose for contradiction that \(F\) has no zero in the closed ball
\(B_\xi(\vartheta^\star)\).  Define the continuous map
\(G:B_\xi(\vartheta^\star)\to B_\xi(\vartheta^\star)\) by
\(G(\vartheta):=\vartheta^\star+\xi F(\vartheta)/\norm{F(\vartheta)}\),
which lands on the sphere.  By Brouwer's fixed-point theorem, \(G\) has
a fixed point \(\vartheta^\dagger\), which satisfies
\(\vartheta^\dagger-\vartheta^\star=\xi F(\vartheta^\dagger)/\norm{F(\vartheta^\dagger)}\),
and so
\((\vartheta^\dagger-\vartheta^\star)^\top F(\vartheta^\dagger)=\xi\norm{F(\vartheta^\dagger)}>0\),
contradicting the inward-pointing hypothesis on the sphere.
\end{proof}

\subsection{Blocks II--III: model-row update}\label{app:entrywise-proof-leftdet}

We condition on \(\mathcal D_1\).  Choose the orthogonal Procrustes
matrix \(O\) minimizing \(\Fnorm{\widehat A-A^\star O}\), and relabel
\((L^\star O,A^\star O)\) as \((L^\star,A^\star)\).  Wedin's theorem and
comparison with the feasible point \(A^\star O\) in the Stage~A
projection give
\begin{equation}\label{eq:task-factor-initial-app}
 \delta_A:=\Fnorm{\widehat A-A^\star}
 \le {C\Fnorm{\widehat M_0-M^\star}\over\sigma_r(M^\star)},
 \quad
 \twoninfnorm{\widehat A}\le C\dt^{-1/2},
 \quad \widehat A^\top\widehat A=I_r .
\end{equation}
For a fixed model \(u\), reorient the observations involving \(u\) and
write them as \((t_q,w_q,Z_q)_{q=1}^{M_u}\).  Put
\[
 \eta_q^\star=(a_{t_q}^\star)^\top\ell_u^\star-M^\star_{w_q,t_q},
 \quad \widehat o_q=\widehat M_{0,w_q,t_q},\quad
 \varepsilon_q=Z_q-\sigma(\eta_q^\star).
\]
The score is the one displayed in Appendix~\ref{app:refine-algorithm}.
For its Taylor expansion define
\[
 N_u=\sum_q\widehat a_{t_q}\varepsilon_q,\qquad
 H_u=\sum_q\sigma'(\eta_q^\star)\widehat a_{t_q}\widehat a_{t_q}^\top,
 \qquad \lambda_u=\lambda_{\min}(H_u),
\]
\[
 d_q=(\widehat a_{t_q}-a_{t_q}^\star)^\top\ell_u^\star
       -(\widehat o_q-M^\star_{w_q,t_q}),
\quad
 B_u=\sum_q\sigma'(\eta_q^\star)\widehat a_{t_q}d_q,
\]
\[
 \beta_u=\sup_{\norm v=1}\sum_q|\widehat a_{t_q}^\top v|d_q^2,
 \qquad
 \gamma_u=\sup_{\norm v=1}\sum_q|\widehat a_{t_q}^\top v|^3 .
\]

\begin{lemma}[Model-row deterministic existence]\label{lem:left-deterministic-app}
If
\begin{equation}\label{eq:left-sufficient-app}
 R_u:=\norm{N_u}+\norm{B_u}+L_3\beta_u
 \le {\lambda_u^2\over4L_3\gamma_u},
\end{equation}
then \(S_u(\ell)=0\) has a solution satisfying
\(\norm{\widetilde\ell_u-\ell_u^\star}\le2R_u/\lambda_u\).
\end{lemma}

\begin{proof}
For \(\delta=\ell-\ell_u^\star\), Taylor's theorem gives
\[
 S_u(\ell_u^\star+\delta)
 =N_u-H_u\delta-B_u-\mathcal R_u(\delta),
\quad
 |\delta^\top\mathcal R_u(\delta)|
 \le L_3\{\gamma_u\norm\delta^3+\beta_u\norm\delta\}.
\]
On the sphere \(\norm\delta=2R_u/\lambda_u\), condition
\eqref{eq:left-sufficient-app} makes
\(\delta^\top S_u(\ell_u^\star+\delta)\le0\).
Lemma~\ref{lem:inward-pointing-app} supplies a zero in the ball.
\end{proof}

\begin{proposition}[Uniform model-factor bound]\label{prop:left-final-app}
Let \(\Delta_F=\Fnorm{\widehat M_0-M^\star}\) and
\(x=(a+C)\log(n\bar d)\).  On the initialization event, simultaneously
for all \(u\in[\dm]\),
\begin{align}
 M_u&\asymp n_2/\dm,\qquad
 \lambda_u\ge c n_2/(\dt\dm),\nonumber\\
 \norm{N_u}&\le
 C\{\sqrt{n_2x/(\dt\dm)}+x/\sqrt\dt\},\nonumber\\
 \norm{B_u}+L_3\beta_u
 &\le C\left\{
 {n_2\delta_A\over\dm\sqrt\dt}
 +{n_2\Delta_F\over\dt\dm^{3/2}}
 +\sqrt{n_2x/(\dt\dm)}\right\}\operatorname{polylog}(n\bar d),
 \label{eq:model-row-master-app}\\
 \gamma_u&\le Cn_2/(\dm\dt^{3/2}).\nonumber
\end{align}
Consequently, with probability at least \(1-n^{-a}\),
\begin{equation}\label{eq:model-factor-rate-app}
 \twoninfnorm{\widetilde L-L^\star}
 \le C\sqrt{\frac{\dt\bar d}{n_2}}\,
       \operatorname{polylog}(n\bar d).
\end{equation}
\end{proposition}

\begin{proof}
A comparison contains a fixed model with probability between constant
multiples of \(1/\dm\); scalar Bernstein gives the first display for
\(M_u\).  Conditional on \(\mathcal D_1\), the summands in \(N_u\) are
centered, have norm at most \(C/\sqrt\dt\), and total conditional
variance at most \(Cn_2/(\dt\dm)\), proving its displayed bound by
vector Bernstein.

For curvature, near-uniform task and pair sampling and
\(\widehat A^\top\widehat A=I_r\) give
\[
 \E(H_u\mid\mathcal D_1)
 \succeq {c n_2\over\dt\dm}
          \sum_{t=1}^{\dt}\widehat a_t\widehat a_t^\top
 ={c n_2\over\dt\dm}I_r.
\]
Each summand has norm \(C/\dt\).  Matrix Bernstein with tail parameter
\(x\) gives, conditionally on \(\mathcal D_1\),
\[
 \Pr_2\!\left[
 \norm{H_u-\E_2H_u}_{\rm op}>
 C\left\{\sqrt{\frac{n_2x}{\dt^2\dm}}+\frac{x}{\dt}\right\}
 \right]\le 2r e^{-x},
\]
where \(\E_2,\Pr_2\) denote conditional expectation and probability.
Since \(n_2\ge C\dm x\), this proves the curvature bound uniformly in
\(u\).

We next give the full bias calculation.  Put
\[
 q_{t,u}:=(\widehat a_t-a_t^\star)^\top\ell_u^\star,
 \qquad e_{w,t}:=\widehat M_{0,w,t}-M^\star_{w,t},
\]
and split \(B_u=B_u^{(A)}-B_u^{(O)}\), where
\[
 B_u^{(A)}=\sum_q\sigma'(\eta_q^\star)\widehat a_{t_q}q_{t_q,u},
 \qquad
 B_u^{(O)}=\sum_q\sigma'(\eta_q^\star)\widehat a_{t_q}e_{w_q,t_q}.
\]
Incoherence and~\eqref{eq:task-factor-initial-app} imply
\[
 \sum_{t=1}^{\dt}q_{t,u}^2\le C\dt\delta_A^2,
 \qquad |q_{t,u}|\le C,
 \qquad |e_{w,t}|\le 2B.
\]
Near-uniform sampling and Cauchy--Schwarz over the task and opponent
indices therefore give the two conditional means
\begin{equation}\label{eq:model-bias-means-app}
 \norm{\E_2B_u^{(A)}}
 \le {Cn_2\delta_A\over\dm\sqrt\dt},
 \qquad
 \norm{\E_2B_u^{(O)}}
 \le {Cn_2\Delta_F\over\dt\dm^{3/2}}.
\end{equation}
For the second inequality, if
\(s_{u,t}=\sum_{w\ne u}c_{u,w,t}e_{w,t}\) with
\(|c_{u,w,t}|\le C\), then
\(\sum_t s_{u,t}^2\le C\dm\Delta_F^2\), which supplies precisely the
factor \(\dm^{-3/2}\).

Let \(\zeta_{i,u}^{(A)}\) and \(\zeta_{i,u}^{(O)}\) be the centered
vector summands of the two biases.  Their conditional variance proxies
and envelopes are
\begin{align}
 \sum_{i\in\mathcal D_2}\E_2\norm{\zeta_{i,u}^{(A)}}^2
 &\le {Cn_2\delta_A^2\over\dt\dm},
 &\max_i\norm{\zeta_{i,u}^{(A)}}&\le {C\over\sqrt\dt},\nonumber\\
 \sum_{i\in\mathcal D_2}\E_2\norm{\zeta_{i,u}^{(O)}}^2
 &\le {Cn_2\Delta_F^2\over\dt^2\dm^2},
 &\max_i\norm{\zeta_{i,u}^{(O)}}&\le {C\over\sqrt\dt}.
 \label{eq:model-bias-moments-app}
\end{align}
Consequently, vector Bernstein yields, for every fixed \(u\),
\begin{align}
 \Pr_2\!\left[
  \norm{B_u^{(A)}-\E_2B_u^{(A)}}>
  C\left\{\delta_A\sqrt{\frac{n_2x}{\dt\dm}}
           +\frac{x}{\sqrt\dt}\right\}\right]&\le2e^{-x},\nonumber\\
 \Pr_2\!\left[
  \norm{B_u^{(O)}-\E_2B_u^{(O)}}>
  C\left\{\frac{\Delta_F\sqrt{n_2x}}{\dt\dm}
           +\frac{x}{\sqrt\dt}\right\}\right]&\le2e^{-x}.
 \label{eq:model-bias-tails-app}
\end{align}

For \(\beta_u\), use \((q-e)^2\le2q^2+2e^2\) and write
\(\beta_u\le2\beta_u^{(A)}+2\beta_u^{(O)}\).  Let
\(\mathcal N\) be a \(1/4\)-net of \(\mathbb S^{r-1}\), with
\(|\mathcal N|\le9^r\).  For fixed \(v\in\mathcal N\), denote the
corresponding nonnegative summands by
\[
 Y_{i,u,v}^{(A)}=\mathbf1_{\{i\ni u\}}
       |\widehat a_{t_i}^\top v|q_{t_i,u}^2,
 \quad
 Y_{i,u,v}^{(O)}=\mathbf1_{\{i\ni u\}}
       |\widehat a_{t_i}^\top v|e_{w_i,t_i}^2.
\]
Their conditional means, variance proxies, and envelopes satisfy
\begin{align}
 \sum_i\E_2Y_{i,u,v}^{(A)}&\le {Cn_2\delta_A^2\over\dm\sqrt\dt},
 &\sum_i\Var_2(Y_{i,u,v}^{(A)})&\le {Cn_2\delta_A^2\over\dt\dm},
 &\max_iY_{i,u,v}^{(A)}&\le {C\over\sqrt\dt},\nonumber\\
 \sum_i\E_2Y_{i,u,v}^{(O)}&\le {Cn_2\Delta_F^2\over\dt^{3/2}\dm^2},
 &\sum_i\Var_2(Y_{i,u,v}^{(O)})&\le {Cn_2\Delta_F^2\over\dt^2\dm^2},
 &\max_iY_{i,u,v}^{(O)}&\le {C\over\sqrt\dt}.
 \label{eq:model-beta-moments-app}
\end{align}
Scalar Bernstein and the fixed-dimensional net argument hence give
\begin{align}
 \Pr_2\Bigg[\beta_u>C\Bigg\{
 &{n_2\delta_A^2\over\dm\sqrt\dt}
 +{n_2\Delta_F^2\over\dt^{3/2}\dm^2}
 +\delta_A\sqrt{\frac{n_2x}{\dt\dm}}
 +{\Delta_F\sqrt{n_2x}\over\dt\dm}
 +{x\over\sqrt\dt}\Bigg\}\Bigg]
 \le C_r e^{-x}.
 \label{eq:model-beta-tail-app}
\end{align}
Finally, on the coverage event,
\[
 \gamma_u\le M_u\twoninfnorm{\widehat A}^3
 \le {Cn_2\over\dm\dt^{3/2}}.
\]
Taking a union bound in
\eqref{eq:model-bias-tails-app}--\eqref{eq:model-beta-tail-app} and the
curvature and coverage events gives total conditional failure
probability at most
\(C_r\dm e^{-x}\le n^{-a}\).  On the initialization event,
\(\delta_A=o(1)\) and \(\Delta_F=o(\sqrt{\dt\dm})\); thus the displayed
terms reduce to~\eqref{eq:model-row-master-app}, after absorbing powers
of \(x\) into \(\operatorname{polylog}(n\bar d)\).
These calculations also verify~\eqref{eq:left-sufficient-app} under
\(n_2\ge C\bar d\log^c(n\bar d)\).

By Assumption~\ref{ass:signal-app},
\(\sigma_r(M^\star)\asymp\sqrt{\dt\dm}\); Theorem
\ref{thm:convex-main-app} and~\eqref{eq:task-factor-initial-app} give
\(\delta_A\lesssim\sqrt{\bar d/n_1}\).
Substitution in \(2R_u/\lambda_u\), with \(n_1\asymp n_2\), yields
\eqref{eq:model-factor-rate-app}.  All inequalities used the free
parameter \(x\), so the finite union bound has failure probability
at most \(n^{-a}\).
\end{proof}

\subsection{Block IV: model centering and pairwise Gram identity}\label{app:entrywise-proof-recenter}

\begin{lemma}[Recentering preserves model differences]\label{lem:centering-app}
If \(\mathbf1_{\dm}^\top L^\star=0\) and
\(\overline L=R_{\rm m}\widetilde L\), then
\[
 \overline\ell_u-\overline\ell_v
 =\widetilde\ell_u-\widetilde\ell_v,\qquad
 \twoninfnorm{\overline L-L^\star}
 \le2\twoninfnorm{\widetilde L-L^\star},\qquad
 \Fnorm{\overline L-L^\star}\le\Fnorm{\widetilde L-L^\star}.
\]
\end{lemma}

\begin{proof}
The first identity follows because \(R_{\rm m}\) subtracts the same row
mean from every row.  The \(2,\infty\) bound follows by the triangle
inequality, and the Frobenius bound because \(R_{\rm m}\) is an
orthogonal projector.
\end{proof}

\begin{lemma}[Pairwise Gram identity]\label{lem:pairwise-gram-app}
For any \(L=(\ell_1,\ldots,\ell_{\dm})^\top\) satisfying
\(\mathbf1_{\dm}^\top L=0\),
\[
 \sum_{u<v}(\ell_u-\ell_v)(\ell_u-\ell_v)^\top
 =\dm L^\top L,\qquad
 \E_{u<v}[(\ell_u-\ell_v)(\ell_u-\ell_v)^\top]
 ={2\over\dm-1}L^\top L.
\]
\end{lemma}

\begin{proof}
Expand the ordered double sum.  The two cross terms vanish because the
rows sum to zero; dividing by two gives the unordered identity.
\end{proof}

\begin{corollary}[Gram lower bound for \(\overline L\)]\label{cor:pairwise-gram-bar-app}
On the event of Proposition~\ref{prop:left-final-app},
\(\lambda_{\min}(\overline L^\top\overline L)\ge c\dt\dm\).
\end{corollary}

\begin{proof}
Since \(L^{\star\top}L^\star=\Sigma^{\star2}\) and
\(\sigma_r^2(M^\star)\asymp\dt\dm\), Weyl's inequality and
\[
 \norm{\overline L^\top\overline L-L^{\star\top}L^\star}_{\rm op}
 \le2\norm{L^\star}_{\rm op}\Fnorm{\overline L-L^\star}
      +\Fnorm{\overline L-L^\star}^2
\]
prove the claim under the sample-size assumption.
\end{proof}

\subsection{Blocks V--VI: task-column update}\label{app:entrywise-proof-right}

For \(i\in\mathcal I_t=\{i\in\mathcal D_3:t_i=t\}\), set
\(x_i=\overline\ell_{m_i}-\overline\ell_{m_i'}\),
\(x_i^\star=\ell_{m_i}^\star-\ell_{m_i'}^\star\), and
\(h_i=x_i-x_i^\star\), and define
\(\varepsilon_i:=Y_i-\sigma((x_i^\star)^\top a_t^\star)\).

\begin{lemma}[Task-column deterministic existence]\label{lem:right-deterministic-app}
With
\[
 N_t=\sum_{i\in\mathcal I_t}x_i\varepsilon_i,\quad
 H_t=\sum_{i\in\mathcal I_t}\sigma'((x_i^\star)^\top a_t^\star)x_ix_i^\top,
 \quad\lambda_t=\lambda_{\min}(H_t),
\]
\[
 B_t=\sum_{i\in\mathcal I_t}\sigma'((x_i^\star)^\top a_t^\star)
           x_i(h_i^\top a_t^\star),
\quad
 \beta_t=\sup_{\norm v=1}\sum_i|x_i^\top v|(h_i^\top a_t^\star)^2,
\quad
 \gamma_t=\sup_{\norm v=1}\sum_i|x_i^\top v|^3,
\]
the condition
\(\norm{N_t}+\norm{B_t}+L_3\beta_t
 \le\lambda_t^2/(4L_3\gamma_t)\)
implies a solution with
\[
 \norm{\widetilde a_t-a_t^\star}
 \le {2\{\norm{N_t}+\norm{B_t}+L_3\beta_t\}\over\lambda_t}.
\]
\end{lemma}

\begin{proof}
Repeat the Taylor expansion in the proof of
Lemma~\ref{lem:left-deterministic-app}, now with regressor \(x_i\), and
apply Lemma~\ref{lem:inward-pointing-app}.
\end{proof}

\begin{proposition}[Uniform task-factor bound]\label{prop:right-final-app}
Let
\(\epsilon_L=\twoninfnorm{\overline L-L^\star}\).  With probability at
least \(1-n^{-a}\), uniformly in \(t\in[\dt]\),
\[
 |\mathcal I_t|\asymp n_3/\dt,\quad
 \lambda_t\ge cn_3,\quad
 \norm{N_t}\le C\sqrt{n_3}\operatorname{polylog}(n\bar d),
\]
\[
 \norm{B_t}\le C{n_3\epsilon_L\over\dt}\operatorname{polylog}(n\bar d),
 \quad
 \beta_t\le C{n_3\epsilon_L^2\over\dt^{3/2}}\operatorname{polylog}(n\bar d),
 \quad
 \gamma_t\le Cn_3\sqrt\dt\operatorname{polylog}(n\bar d),
\]
and hence
\begin{equation}\label{eq:task-factor-rate-app}
 \twoninfnorm{\widetilde A-A^\star}
 \le C\left\{{1\over\sqrt{n_3}}+
              {\epsilon_L\over\dt}+
              {\epsilon_L^2\over\dt^{3/2}}\right\}
       \operatorname{polylog}(n\bar d).
\end{equation}
\end{proposition}

\begin{proof}
Scalar Bernstein gives \(|\mathcal I_t|\asymp n_3/\dt\).
Conditional on the first two splits, the score-noise summands have
envelope \(C\sqrt\dt\) and variance sum \(Cn_3\), giving the displayed
noise bound.  Lemma~\ref{lem:pairwise-gram-app} and
Corollary~\ref{cor:pairwise-gram-bar-app} give
\[
 \E(H_t\mid\mathcal D_1,\mathcal D_2)
 \succeq {c n_3\over\dt}\,
 {2\over\dm-1}\overline L^\top\overline L
 \succeq cn_3I_r.
\]
Matrix Bernstein proves the same lower bound empirically.

Here are the explicit conditional calculations.  Write
\(\E_3,\Pr_3\) for expectation and probability conditional on
\((\mathcal D_1,\mathcal D_2)\).  On the event of
Proposition~\ref{prop:left-final-app}, the pairwise Gram identity and
near-uniform pair sampling imply, uniformly over unit \(v\),
\begin{equation}\label{eq:task-pair-moments-app}
 \E_{\rm pair}\norm{x}^2\le C\dt,\qquad
 \E_{\rm pair}(x^\top v)^2\le C\dt,\qquad
 \E_{\rm pair}(h^\top a_t^\star)^2
 \le {C\epsilon_L^2\over\dt},
\end{equation}
while
\[
 \norm{x}\le C\sqrt\dt,\qquad
 \norm{h}\le2\epsilon_L,\qquad
 |h^\top a_t^\star|\le {C\epsilon_L\over\sqrt\dt}.
\]
For example, the last expectation in
\eqref{eq:task-pair-moments-app} follows from
\(\Fnorm{\overline L-L^\star}^2\le\dm\epsilon_L^2\) and the pairwise
Gram identity applied to \(\overline L-L^\star\).

For curvature, a summand has operator norm at most \(C\dt\), and the
conditional matrix-variance proxy is at most \(Cn_3\dt\).  Therefore
\begin{equation}\label{eq:task-curvature-tail-app}
 \Pr_3\!\left[
  \norm{H_t-\E_3H_t}_{\rm op}>
  C\{\sqrt{n_3\dt x}+\dt x\}\right]
 \le2r e^{-x}.
\end{equation}
Together with the preceding lower bound for \(\E_3H_t\), this yields
\(\lambda_t\ge cn_3\) whenever \(n_3\ge C\dt x\).

For \(B_t\), Cauchy--Schwarz and
\eqref{eq:task-pair-moments-app} first give
\begin{equation}\label{eq:task-bias-mean-app}
 \norm{\E_3B_t}
 \le {Cn_3\over\dt}
 \left\{\E_{\rm pair}\norm{x}^2\right\}^{1/2}
 \left\{\E_{\rm pair}(h^\top a_t^\star)^2\right\}^{1/2}
 \le {Cn_3\epsilon_L\over\dt}.
\end{equation}
If \(\zeta_{i,t}^{(B)}\) denotes its centered vector summand, then
\begin{equation}\label{eq:task-bias-moments-app}
 \sum_{i\in\mathcal D_3}\E_3\norm{\zeta_{i,t}^{(B)}}^2
 \le {Cn_3\epsilon_L^2\over\dt},
 \qquad
 \max_i\norm{\zeta_{i,t}^{(B)}}\le C\epsilon_L.
\end{equation}
Thus vector Bernstein gives
\begin{equation}\label{eq:task-bias-tail-app}
 \Pr_3\!\left[
  \norm{B_t-\E_3B_t}>
  C\left\{\epsilon_L\sqrt{\frac{n_3x}{\dt}}
            +\epsilon_Lx\right\}\right]
 \le2e^{-x}.
\end{equation}

For \(\beta_t\), take a \(1/4\)-net
\(\mathcal N\subset\mathbb S^{r-1}\), \(|\mathcal N|\le9^r\), and
for fixed \(v\in\mathcal N\) put
\(Y_{i,t,v}=\mathbf1_{\{t_i=t\}}|x_i^\top v|
(h_i^\top a_t^\star)^2\).  The mean, variance proxy, and envelope are
\begin{equation}\label{eq:task-beta-moments-app}
 \sum_i\E_3Y_{i,t,v}\le {Cn_3\epsilon_L^2\over\dt^{3/2}},
 \qquad
 \sum_i\Var_3(Y_{i,t,v})\le {Cn_3\epsilon_L^4\over\dt^2},
 \qquad
 \max_iY_{i,t,v}\le {C\epsilon_L^2\over\sqrt\dt}.
\end{equation}
Scalar Bernstein and the net argument imply
\begin{equation}\label{eq:task-beta-tail-app}
 \Pr_3\!\left[
 \beta_t>C\left\{
 {n_3\epsilon_L^2\over\dt^{3/2}}
 +{\epsilon_L^2\sqrt{n_3x}\over\dt}
 +{\epsilon_L^2x\over\sqrt\dt}\right\}\right]
 \le C_r e^{-x}.
\end{equation}

The same net controls \(\gamma_t\).  For
\(G_{i,t,v}=\mathbf1_{\{t_i=t\}}|x_i^\top v|^3\),
\begin{equation}\label{eq:task-gamma-moments-app}
 \sum_i\E_3G_{i,t,v}\le Cn_3\sqrt\dt,\qquad
 \sum_i\Var_3(G_{i,t,v})\le Cn_3\dt^2,\qquad
 \max_iG_{i,t,v}\le C\dt^{3/2},
\end{equation}
and hence
\begin{equation}\label{eq:task-gamma-tail-app}
 \Pr_3\!\left[
 \gamma_t>C\{n_3\sqrt\dt+\dt\sqrt{n_3x}+\dt^{3/2}x\}
 \right]\le C_r e^{-x}.
\end{equation}
For completeness, the noise calculation used above is
\[
 \Pr_3\!\left[
 \norm{N_t}>C\{\sqrt{n_3x}+\sqrt\dt\,x\}\right]\le2e^{-x},
\]
because its variance proxy is \(Cn_3\) and its envelope is
\(C\sqrt\dt\).

With \(x=(a+C)\log(n\bar d)\), a union bound over
\(t\in[\dt]\) and the fixed nets has failure probability at most
\(C_r\dt e^{-x}\le n^{-a}\).  Since \(n_3\ge C\bar d\log^c(n\bar d)\),
the fluctuation terms in
\eqref{eq:task-bias-tail-app}--\eqref{eq:task-gamma-tail-app} are
absorbed by their respective leading terms, giving exactly the three
bounds in the proposition.  The deterministic lemma now proves
\eqref{eq:task-factor-rate-app}; its sufficient condition follows from
the same sample-size inequality.
\end{proof}
\subsection{Final assembly: proof of Theorem~\ref{thm:entrywise}}\label{app:entrywise-proof-final}

\begin{theorem}[Entrywise refinement; restatement of Theorem~\ref{thm:entrywise}]\label{thm:pairwise-max-app}
Under the assumptions of Section~\ref{sec:setup} and
Assumptions~\ref{ass:score-app}, \ref{ass:design-app},
\ref{ass:signal-app}, and~\ref{ass:sample-app}, fix any \(a>0\)
and let
\(n\ge C\,\mathrm{poly}(\mu,r,\kappa,B)\,\bar d\log^c(n\bar d)\).
Then with probability at least \(1-n^{-a}\),
\[
    \infnorm{\Thetahat-\Thetastar}
    \;\le\;
    C\sqrt{\frac{\bar d\,\mathrm{polylog}(n\bar d)}{n}}.
\]
\end{theorem}

\begin{proof}
On
\(\mathcal E_{\rm init}\cap\mathcal E_{\rm refine}\) (the intersection of
the convex-initialization and refinement events), apply the
factorization
\[
    \widetilde M-M^\star
    =
    (\overline L-L^\star)(A^\star)^\top
    +
    \overline L\,(\widetilde A-A^\star)^\top.
\]
The entrywise norm is bounded by
\[
    \infnorm{\Thetahat-\Thetastar}
    \le
    \twoninfnorm{\overline L-L^\star}\cdot\twoninfnorm{A^\star}
    +\twoninfnorm{\widetilde A-A^\star}\cdot\twoninfnorm{\overline L}.
\]
Substituting
\(\twoninfnorm{\overline L-L^\star}\lesssim\sqrt{\dt\bar d/n_2}\,\mathrm{polylog}\)
(Proposition~\ref{prop:left-final-app} via Lemma~\ref{lem:centering-app}),
the bound~\eqref{eq:task-factor-rate-app},
\(\twoninfnorm{A^\star}\le C\dt^{-1/2}\) (incoherence), and
\(\twoninfnorm{\overline L}\le 2\twoninfnorm{L^\star}\le C\dt^{1/2}\)
(incoherence + recentering bound) yields
\[
    \infnorm{\Thetahat-\Thetastar}
    \;\lesssim\;
    \sqrt{\frac{\dt\bar d}{n}}\cdot\frac{1}{\sqrt\dt}
      +\sqrt\dt\left\{\frac1{\sqrt n}
       +\frac{\sqrt{\dt\bar d/n}}{\dt}
       +\frac{\dt\bar d/n}{\dt^{3/2}}\right\}
    \;\lesssim\;
    \sqrt{\frac{\bar d}{n}}+\frac{\bar d}{n}
    \;\lesssim\;
    \sqrt{\frac{\bar d}{n}}\,\mathrm{polylog}(n\bar d).
\]
The factor bounds also preserve exact rank.  In the normalization used
above, \(\sigma_r(L^\star)\asymp\sqrt{\dt\dm}\) and
\(\sigma_r(A^\star)=1\).  The Frobenius bounds implied by
Propositions~\ref{prop:left-final-app} and~\ref{prop:right-final-app}
are \(o(\sigma_r(L^\star))\) and \(o(1)\), respectively, under the
stated sample-size and signal conditions.  Weyl's inequality therefore
gives
\[
 \sigma_r(\overline L)>0,\qquad \sigma_r(\widetilde A)>0
\]
on the same event.  Both factors have column rank \(r\), so
\(\widetilde M=\overline L\widetilde A^\top\), and hence
\(\Thetahat\), has exactly rank \(r\).
The probability calibration to \(1-n^{-a}\) follows from
Appendix~\ref{app:notation-probcalib}.
\end{proof}

This completes the proof of Theorem~\ref{thm:entrywise}.\qed

\section{Proof of Proposition~\ref{prop:topk-hamming}: top-\texorpdfstring{$K$}{K} Hamming and exact recovery}\label{app:topk-hamming}

This appendix proves the deterministic reduction from entrywise score
estimation to the taskwise top-\(K\) Hamming bound, establishes exact
recovery under the margin condition, and discusses minimax optimality.

\subsection{Setup}\label{app:topk-setup}

For task \(t\in[\dt]\), abbreviate \(\theta_m:=\Thetastar_{t,m}\),
\(\widehat\theta_m:=\Thetahat_{t,m}\), and let
\(S_t:=\TopK(t)\), \(\widehat S_t:=\widehat{\TopK}(t)\).
Let \(\theta_{(1)}\ge\theta_{(2)}\ge\cdots\ge\theta_{(\dm)}\) denote the
sorted true scores for task \(t\), and define the \(K\)-boundary midpoint
\(\tau_K(t):=(\theta_{(K)}+\theta_{(K+1)})/2\).
The normalized top-\(K\) Hamming loss is
\(\mathsf{Ham}_{K,t}:=(2K)^{-1}|\widehat S_t\triangle S_t|\),
and the boundary mass at radius \(r\) is
\(\mathcal R_{K,t}(r;\Thetastar)
:=(2K)^{-1}|\{m:|\theta_m-\tau_K(t)|\le r\}|\).
On the event
\(\mathcal E_\infty:=\{\infnorm{\Thetahat-\Thetastar}\le\varepsilon_n\}\)
of Theorem~\ref{thm:entrywise}, all statements below are deterministic.

\subsection{Hamming bound from entrywise error}\label{app:topk-hamming-prop}

\begin{proposition}[Restatement of Proposition~\ref{prop:topk-hamming}]\label{prop:topk-hamming-app}
On \(\mathcal E_\infty\), for every \(t\in[\dt]\),
\(\mathsf{Ham}_{K,t}\le\mathcal R_{K,t}(2\varepsilon_n;\Thetastar)\).
Hence if \(\Pr(\mathcal E_\infty)\ge 1-n^{-a}\), the Hamming bound holds
simultaneously over all tasks with probability at least \(1-n^{-a}\).
\end{proposition}

\begin{proof}
We prove the deterministic inclusion
\(\widehat S_t\triangle S_t\subseteq\{m:|\theta_m-\tau_K(t)|\le 2\varepsilon_n\}\)
on \(\mathcal E_\infty\); the cardinality bound then follows by dividing by
\(2K\).  Take a false positive \(u\in\widehat S_t\setminus S_t\); since
\(|\widehat S_t|=|S_t|=K\), there exists a false negative
\(v\in S_t\setminus\widehat S_t\).  Because \(u\in\widehat S_t\) and
\(v\notin\widehat S_t\), \(\widehat\theta_u\ge\widehat\theta_v\) under the
deterministic tie-breaking rule.  Hence
\(\theta_u\ge\widehat\theta_u-\varepsilon_n\ge\widehat\theta_v-\varepsilon_n\ge\theta_v-2\varepsilon_n\).
Since \(v\in S_t\), \(\theta_v\ge\theta_{(K)}\), so
\(\theta_u\ge\theta_{(K)}-2\varepsilon_n\); and since \(u\notin S_t\),
\(\theta_u\le\theta_{(K+1)}\).  Combining,
\(\theta_{(K)}-2\varepsilon_n\le\theta_u\le\theta_{(K+1)}\), and using
\(\tau_K(t)\in[\theta_{(K+1)},\theta_{(K)}]\),
\(|\theta_u-\tau_K(t)|\le 2\varepsilon_n\).  The false-negative case is
symmetric.
\end{proof}

\subsection{Exact recovery under a margin condition}\label{app:topk-exact}

\begin{corollary}[Exact top-\(K\) recovery]\label{cor:topk-exact-app}
Define
$
\Delta_K(t):=\theta_{(K)}-\theta_{(K+1)}.
$
If \(\Delta_K(t)>2\varepsilon_n\), then on \(\mathcal E_\infty\),
$
\widehat S_t=S_t.
$
In particular, if
$
\min_{t\in[\dt]}\Delta_K(t)>2\varepsilon_n,
$
then exact top-\(K\) recovery holds simultaneously over all tasks with probability
at least \(1-n^{-a}\).
\end{corollary}

\begin{proof}
Fix a task \(t\). On \(\mathcal E_\infty\), for every
\(m\in S_t\) and \(m'\notin S_t\),
\[
\widehat\theta_m-\widehat\theta_{m'}
\ge
\theta_m-\theta_{m'}-2\varepsilon_n
\ge
\theta_{(K)}-\theta_{(K+1)}-2\varepsilon_n
=
\Delta_K(t)-2\varepsilon_n.
\]
Hence, if \(\Delta_K(t)>2\varepsilon_n\), then
\(\widehat\theta_m>\widehat\theta_{m'}\) for every
\(m\in S_t\) and \(m'\notin S_t\), which implies
\(\widehat S_t=S_t\).

Because \(\mathcal E_\infty\) controls the entrywise error uniformly over all
tasks and models, the same argument applies simultaneously to every task.
Therefore, if
\(\min_{t\in[\dt]}\Delta_K(t)>2\varepsilon_n\), exact top-\(K\) recovery holds
for all tasks on \(\mathcal E_\infty\), which has probability at least
\(1-n^{-a}\).
\end{proof}

% \begin{corollary}[Exact top-\(K\) recovery]\label{cor:topk-exact-app}
% Define \(\Delta_K(t):=\theta_{(K)}-\theta_{(K+1)}\).  If
% \(\Delta_K(t)>4\varepsilon_n\), then on \(\mathcal E_\infty\) the boundary
% mass \(\mathcal R_{K,t}(2\varepsilon_n;\Thetastar)=0\), hence
% \(\widehat S_t=S_t\).  In particular, if
% \(\min_{t\in[\dt]}\Delta_K(t)>4\varepsilon_n\), then exact top-\(K\)
% recovery holds simultaneously over all tasks with probability at least
% \(1-n^{-a}\).
% \end{corollary}

% \begin{proof}
% Under \(\Delta_K(t)>4\varepsilon_n\), every \(m\in S_t\) has
% \(\theta_m-\tau_K(t)\ge\Delta_K(t)/2>2\varepsilon_n\), and symmetrically
% for \(m\notin S_t\); so
% \(\mathcal R_{K,t}(2\varepsilon_n;\Thetastar)=0\), and
% Proposition~\ref{prop:topk-hamming-app} applies.  The simultaneous
% statement follows by union bound over \(t\in[\dt]\).
% \end{proof}

\subsection{Minimax-optimality remark}\label{app:topk-minimax}

The proposition is a deterministic reduction from entrywise score
estimation to top-\(K\) Hamming accuracy.  For minimax optimality we
appeal to the single-task BTL ranking literature:
\cite{chen2022partial} characterize the minimax rate for normalized
Hamming partial recovery in BTL top-\(K\) ranking and show that MLE
attains both partial and exact recovery thresholds, while
\cite{chen2019spectral} establish the minimax sample complexity for
exact top-\(K\) identification through entrywise score control combined
with a \(K\)-versus-\((K+1)\) margin condition.  Our entrywise rate
\(\varepsilon_n\asymp\sqrt{\bar d\,\mathrm{polylog}(n\bar d)/n}\) and the
exact recovery margin \(2\varepsilon_n\) match these single-task minimax
characterizations up to logarithmic factors.  The gain from low-rank
structure is the factor \(\bar d\) instead of \(\bar d^{\,2}\) in the
per-task sample complexity; the dependence on \(\dt\) is only through
the union bound and is logarithmic.

\section{Proof of Theorem~\ref{thm:fdim-clt}: finite-dimensional inference and efficiency}\label{app:fdim-proof}

This appendix proves Theorem~\ref{thm:fdim-clt} of
Section~\ref{sec:fdim} together with the efficiency claim referenced in
the discussion of \(V_{\rm eff}(\Gamma)\).  Throughout we work under the
assumptions of Appendix~\ref{app:notation}.  The structure of the
appendix is as follows.

\begin{itemize}[leftmargin=1.6em,topsep=2pt,itemsep=3pt]
\item Appendix~\ref{app:fdim-proof-onestep} restates the one-step
estimator algorithm in matrix form (essentially a copy of
Section~\ref{sec:fdim} for self-containment).
\item Appendix~\ref{app:fdim-proof-1d} proves the scalar
(1-dimensional) semiparametric efficiency lower bound directly by an information
inequality.
\item Appendix~\ref{app:fdim-proof-loewner} extends the lower bound to
the joint Loewner inequality
\(\bar\Sigma\succeq\Sigma\) for fixed \(q=O(1)\) contrasts via reduction
to the scalar bound applied at every linear combination
\(\Gamma_u=\sum_j u_j\Gamma_j\), \(u\in\R^q\).
\item Appendix~\ref{app:fdim-proof-projector} writes out the
closed-form row-centered matrix tangent projector used throughout.
\item Appendix~\ref{app:fdim-proof-decomp} gives the exact foldwise
one-step identity, proves the second-order exact-rank normal leakage,
and establishes the single-contrast remainder bound by conditional
Taylor expansion and transported-tangent information perturbation.
\item Appendix~\ref{app:fdim-proof-uniform} extends the single-contrast
bound to a uniform statement over a polynomial-size contrast family by
union bound, with the per-term envelope/variance accounting carried out
explicitly.  This is the input to Appendix~\ref{app:ranking-proof}.
\item Appendix~\ref{app:fdim-proof-be} proves the Berry--Esseen rate
\(\rho_n\lesssim\sqrt{\bar d/n}\) for the standardized leading term,
with full computation of the second and third moments.
\item Appendix~\ref{app:fdim-proof-clt} combines the uniform one-step expansion of
Appendix~\ref{app:fdim-proof-uniform} with the oracle Gaussian approximation of
Appendix~\ref{app:fdim-proof-be} to obtain the multivariate rectangle CLT proving
Theorem~\ref{thm:fdim-clt}.
\item Appendices~\ref{app:fdim-proof-varcons} and~\ref{app:fdim-proof-cov}
prove the relative variance consistency and the covariance consistency
in correlation form, respectively, both of which are inputs to
Appendix~\ref{app:ranking-proof}.
\end{itemize}

\subsection{The one-step estimator and its plug-in operators}\label{app:fdim-proof-onestep}

%\begin{revisionblock}
Recall the cross-fitting one-step procedure in Section~\ref{sec:fdim}.
For a fixed fold \(\mathcal I_k\), let
\((\widehat\Theta^{(-k)},\widehat P_{\mathbb T}^{(-k)})\) be trained on
\(\mathcal I_k^c\).  Construct \(\widehat G_k\) by averaging the
plug-in Fisher weights over the known design law, as in
\eqref{eq:Ghat-def-app}; let
\(\widehat H_{\Gamma,k}=\widehat{\mathcal A}_k^\dagger\Gamma\).  The
fold estimator for
\(\psi_\Gamma(\Thetastar)=\ip{\Gamma}{\Thetastar}\) is
\begin{equation}\label{eq:onestep-app}
    \widehat\psi_{\Gamma,k}
    :=
    \psi_\Gamma(\widehat\Theta^{(-k)})
    +
    \frac1{|\mathcal I_k|}\sum_{i\in\mathcal I_k}
    s_\eta\bigl(Y_i,\widehat\eta_i^{(-k)}\bigr)\,
    \ip{\widehat H_{\Gamma,k}}{X_i},
    \qquad
    \widehat\eta_i^{(-k)}:=\ip{\widehat\Theta^{(-k)}}{X_i}.
\end{equation}
The final estimate is
\(\widehat\psi_\Gamma=\sum_k|\mathcal I_k|\widehat\psi_{\Gamma,k}/n\).
Conditional on the training sample, the complete plug-in
direction is fixed and the observations in each evaluation fold remain
i.i.d.  For the algebra below we suppress \(k\); every identity is
applied foldwise and then averaged.
%\end{revisionblock}

For a finite contrast family \(\Gamma_1,\ldots,\Gamma_q\), let
\[
    \psi_j:=\psi_{\Gamma_j}(\Thetastar),
    \quad
    \widehat\psi_j:=\widehat\psi_{\Gamma_j},
    \quad
    H_j^\star:=\mathcal A^\dagger\Gamma_j,
    \quad
    \widehat H_j:=\widehat{\mathcal A}^{\dagger}\Gamma_j,
\]
\[
    \phi_j(W_i):=s_\eta(Y_i,\eta_i^\star)\ip{H_j^\star}{X_i},
    \quad
    \widehat\phi_j(W_i):=s_\eta(Y_i,\widehat\eta_i)\ip{\widehat H_j}{X_i},
\]
\[
    \sigma_j^2:=V_{\rm eff}(\Gamma_j)=\E^\star[\phi_j^2]
    =\ip{P_{\mathbb T}\Gamma_j}{A_{\mathbb T}^{-1}P_{\mathbb T}\Gamma_j},
    \quad
    \Sigma_{jk}:=\E^\star[\phi_j\phi_k]
    =\ip{P_{\mathbb T}\Gamma_j}{A_{\mathbb T}^{-1}P_{\mathbb T}\Gamma_k},
\]
and the standardized oracle coordinate
\(Z_{ij}:=\phi_j(W_i)/\sigma_j\), so \(\E^\star Z_{ij}=0\) and
\(\E^\star Z_{ij}^2=1\).

\subsection{Single-contrast (1D) semiparametric efficiency lower bound}\label{app:fdim-proof-1d}

For any fixed contrast \(\Gamma\in\R^{\dt\times\dm}\), the
\emph{semiparametric efficiency bound} for any regular estimator
\(\widehat\psi\) of \(\psi_\Gamma(\Thetastar)\) is
\begin{equation}\label{eq:1d-eff-bound-app}
    \operatorname{avar}^\star\!\left\{
      \sqrt n\bigl(\widehat\psi-\psi_\Gamma(\Thetastar)\bigr)
    \right\}
    \;\ge\;
    V_{\rm eff}(\Gamma),
    \qquad
    V_{\rm eff}(\Gamma)
    =
    \ip{P_{\mathbb T}\Gamma}{A_{\mathbb T}^{-1}P_{\mathbb T}\Gamma}.
\end{equation}
We give the complete information-inequality argument for the present centered
low-rank matrix model.

\begin{theorem}[1D semiparametric efficiency bound]\label{thm:fdim-1d-lb-app}
For every regular estimator of \(\psi_\Gamma\) in the exact-rank,
row-centered model, its asymptotic variance is at least
\(V_{\rm eff}(\Gamma)\).  The influence function
\[
 \phi_\Gamma=s_\eta(Y,\eta^\star)\ip{H_\Gamma^\star}{X},
 \qquad H_\Gamma^\star=\mathcal A^\dagger\Gamma,
\]
has variance \(V_{\rm eff}(\Gamma)\) and attains this bound.
\end{theorem}

\begin{proof}
First verify that every tangent direction is generated by a path that
stays inside the exact-rank model.  The bounded-logit constant can be
enlarged by a fixed slack around \(\Thetastar\), so the local paths
below remain in the same bounded-logit neighborhood.  Write
\(H=U^\star A^\top+C(V^\star)^\top\), with
\(A^\top\mathbf1_{\dm}=0\), and define
\[
 \Theta_\varepsilon
 =(U^\star\Sigma^\star+\varepsilon C)
  (V^\star+\varepsilon A(\Sigma^\star)^{-1})^\top .
\]
For all sufficiently small \(\varepsilon\), both factors have column
rank \(r\), so \(\Theta_\varepsilon\) has exactly rank \(r\).
Moreover,
\[
 \Theta_\varepsilon\mathbf1_{\dm}=0,\qquad
 \Theta_0=\Thetastar,\qquad
 {d\over d\varepsilon}\Theta_\varepsilon\big|_{\varepsilon=0}=H,
\]
because \((V^\star)^\top\mathbf1_{\dm}=0\) and
\((A(\Sigma^\star)^{-1})^\top\mathbf1_{\dm}
=(\Sigma^\star)^{-1}A^\top\mathbf1_{\dm}=0\).

Along this path the one-observation score at zero is
\[
 \dot\ell_H(X,Y)=s_\eta(Y,\eta^\star)\ip{H}{X},
\]
and its Fisher information is
\(\E\dot\ell_H^2=\ip{H}{GH}=\ip{H}{A_{\mathbb T}H}\).
Bounded logits and the bounded score derivatives imply differentiability
in quadratic mean, hence the \(n^{-1/2}\)-local submodels are LAN.
If the sampling law of \(X\) is also treated as an unknown nuisance,
its tangent scores have the form \(a(X)\) with \(\E a(X)=0\).  They are
orthogonal to every conditional BTL score because
\[
 \E\!\left[a(X)s_\eta(Y,\eta^\star)\ip{H}{X}\right]
 =\E\!\left[a(X)\ip{H}{X}\E\{s_\eta(Y,\eta^\star)\mid X\}\right]=0.
\]
Thus enlarging the nuisance tangent space to include unknown
\((\nu,\pi)\) does not change the canonical gradient or the efficiency
bound below.
The pathwise derivative of the target is
\[
 \dot\psi_\Gamma(H)=\ip{\Gamma}{H}
                   =\ip{P_{\mathbb T}\Gamma}{H}.
\]
The H\'ajek--Le Cam convolution bound for every regular estimator therefore
gives the variance lower bound
\[
 \sup_{H\in\mathbb T\setminus\{0\}}
 {\,\ip{P_{\mathbb T}\Gamma}{H}^2\over
    \ip{H}{A_{\mathbb T}H}}
 =
 \ip{P_{\mathbb T}\Gamma}
     {A_{\mathbb T}^{-1}P_{\mathbb T}\Gamma}
 =V_{\rm eff}(\Gamma).
\]
The maximizing direction is
\(H_\Gamma^\star=A_{\mathbb T}^{-1}P_{\mathbb T}\Gamma\).

Finally,
\[
 \E^\star\phi_\Gamma^2
 =\ip{H_\Gamma^\star}{GH_\Gamma^\star}
 =\ip{H_\Gamma^\star}{A_{\mathbb T}H_\Gamma^\star}
 =V_{\rm eff}(\Gamma),
\]
and, for every \(H\in\mathbb T\),
\[
 \E^\star[\phi_\Gamma\dot\ell_H]
 =\ip{H_\Gamma^\star}{GH}
 =\ip{P_{\mathbb T}\Gamma}{H}
 =\dot\psi_\Gamma(H).
\]
Thus \(\phi_\Gamma\) is the canonical gradient and the one-step
estimator attaining it is semiparametrically efficient.
\end{proof}
\subsection{Multivariate (Loewner) lower bound by 1D + arbitrary \texorpdfstring{$u$}{u}}\label{app:fdim-proof-loewner}

For a fixed finite contrast family
\(\Gamma_1,\ldots,\Gamma_q\) (\(q=O(1)\)), the candidate efficient
covariance is
\(\Sigma_{jk}=\ip{P_{\mathbb T}\Gamma_j}{A_{\mathbb T}^{-1}P_{\mathbb T}\Gamma_k}\).

\begin{proposition}[Joint semiparametric efficiency lower bound]\label{prop:fdim-loewner-app}
Let \(\bar\psi=(\bar\psi_1,\ldots,\bar\psi_q)^\top\) be any regular
estimator of \(\psi=(\psi_1,\ldots,\psi_q)^\top\), with
\(\sqrt n(\bar\psi-\psi)\rightsquigarrow\mathcal N(0,\bar\Sigma)\) for
some covariance \(\bar\Sigma\succeq 0\).  Then
\(\bar\Sigma-\Sigma\succeq 0\) in Loewner order; equivalently,
\(u^\top\bar\Sigma u\ge u^\top\Sigma u\) for every \(u\in\R^q\).
\end{proposition}

\begin{proof}
We reduce to the scalar bound applied at every linear combination
\(u^\top\psi=\psi_{\Gamma_u}\) for the combined contrast
\(\Gamma_u:=\sum_{j=1}^q u_j\Gamma_j\).

\textbf{Step 1: \(u^\top\bar\psi\) is a regular estimator of \(\psi_{\Gamma_u}\).}
Since \(\bar\psi\) is regular for the vector \(\psi\), the linear
functional \(u^\top\bar\psi\) is regular for the corresponding scalar
target \(u^\top\psi=\sum_j u_j\psi_j=\sum_j u_j\ip{\Gamma_j}{\Thetastar}=\ip{\Gamma_u}{\Thetastar}=\psi_{\Gamma_u}(\Thetastar)\).
Its asymptotic variance is \(u^\top\bar\Sigma u\) by the continuous
mapping theorem.

\textbf{Step 2: scalar lower bound.}
Apply Theorem~\ref{thm:fdim-1d-lb-app} to the scalar functional
\(\psi_{\Gamma_u}\): for every regular estimator,
\[
    \operatorname{avar}^\star\!\left\{
      \sqrt n\bigl(u^\top\bar\psi-u^\top\psi\bigr)
    \right\}
    \ge
    V_{\rm eff}(\Gamma_u).
\]
Since the left side equals \(u^\top\bar\Sigma u\), we obtain
\(u^\top\bar\Sigma u\ge V_{\rm eff}(\Gamma_u)=\ip{P_{\mathbb T}\Gamma_u}{A_{\mathbb T}^{-1}P_{\mathbb T}\Gamma_u}\).

\textbf{Step 3: bilinearity.}
By linearity of \(P_{\mathbb T}\),
\(P_{\mathbb T}\Gamma_u=\sum_j u_j P_{\mathbb T}\Gamma_j\), and by
bilinearity of \(\ip{\cdot}{A_{\mathbb T}^{-1}\cdot}\),
\[
    \ip{P_{\mathbb T}\Gamma_u}{A_{\mathbb T}^{-1}P_{\mathbb T}\Gamma_u}
    =
    \sum_{j,k}u_j u_k\ip{P_{\mathbb T}\Gamma_j}{A_{\mathbb T}^{-1}P_{\mathbb T}\Gamma_k}
    =
    u^\top\Sigma u.
\]
Combining steps 2 and 3 gives \(u^\top(\bar\Sigma-\Sigma)u\ge 0\).
Since this holds for every \(u\in\R^q\), \(\bar\Sigma\succeq\Sigma\) in
Loewner order.
\end{proof}

\begin{remark}[Singular covariance and redundant contrasts]\label{rem:fdim-singular-app}
The covariance \(\Sigma\) may be singular without invalidating the
proof.  Suppose \(\Gamma_3=\Gamma_1+\Gamma_2\) so that
\(\phi_3=\phi_1+\phi_2\) and the third row of \(\Sigma\) is the exact
sum of the first two.  The Loewner bound still holds: the proof above
does not require invertibility, only the scalar efficiency bound for
every linear combination \(u^\top\psi\).  If \(u^\top\Sigma u=0\) for
some \(u\ne 0\), then the corresponding contrast has zero efficient
variance; this is an exact local redundancy and produces a degenerate
Gaussian limit on a lower-dimensional subspace.  In particular, the
inverse \(A_{\mathbb T}^{-1}\) is well-defined on \(\mathbb T\),
so the proof is unaffected.
\end{remark}

\subsection{Closed-form matrix tangent projector}\label{app:fdim-proof-projector}

We write out the closed form of \(P_{\mathbb T}\) before proceeding to
the decomposition.

\begin{lemma}[Closed-form matrix tangent projector]\label{lem:fdim-projector-app}
Let \(\Thetastar=U^\star\Sigma^\star(V^\star)^\top\) with
\(U^\star\in\R^{\dt\times r}\), \(V^\star\in\R^{\dm\times r}\) and
\((V^\star)^\top\mathbf 1_{\dm}=0\), which follows from
\(\Thetastar\mathbf1_{\dm}=0\).  Define
\(P_{U^\star}:=U^\star(U^\star)^\top\),
\(P_{V^\star}:=V^\star(V^\star)^\top\), and
\(R_{\rm m}:=I_{\dm}-\dm^{-1}\mathbf 1\mathbf 1^\top\).  Then
\begin{equation}\label{eq:tangent-projector-matrix}
    P_{\mathbb T}\Gamma
    =
    P_{U^\star}\Gamma R_{\rm m}
    +(I-P_{U^\star})\Gamma P_{V^\star},
    \qquad\Gamma\in\R^{\dt\times\dm}.
\end{equation}
\end{lemma}

\begin{proof}
The (unconstrained) rank-\(r\) tangent space at \(\Thetastar\) is
\(\{U^\star A^\top+B(V^\star)^\top:A\in\R^{\dm\times r},B\in\R^{\dt\times r}\}\).
Differentiating \(\Theta\mathbf1_{\dm}=0\) restricts the first
component by \(A^\top\mathbf1_{\dm}=0\); the second is already centered
because \((V^\star)^\top\mathbf1_{\dm}=0\).  Decompose an arbitrary
matrix into the four mutually orthogonal blocks determined by
\((P_{U^\star},I-P_{U^\star})\) on the left and
\((P_{V^\star},R_{\rm m}-P_{V^\star},I-R_{\rm m})\) on the right.
Keeping exactly the two tangent blocks gives
\eqref{eq:tangent-projector-matrix}.
\end{proof}

%\begin{revisionblock}
\begin{lemma}[Foldwise geometry and information bounds]
\label{lem:foldwise-geometry-app}
Let \(\widehat\Theta\) be the exact-rank-\(r\), row-centered estimator
constructed by the refinement procedure on the training sample of a
cross-fitting fold, and let \(\widehat{\mathbb T}\) be its tangent
space.  Write \(\mathbb P_{\rm ev}\) for the empirical measure on
that fold's independent evaluation subsample and define
\[
 \ip{\widehat G_{\rm ev}H_1}{H_2}
 :=\mathbb P_{\rm ev}\!\left[
 I(\widehat\eta)\ip{H_1}{X}\ip{H_2}{X}\right].
\]
Under the conditions of
Theorem~\ref{thm:entrywise}, uniformly over the fixed number of folds
and with probability at least \(1-n^{-a}\),
\begin{align}
 \frac{c}{\dstar}\Fnorm D^2
 \le \ip{D}{QD}
 \le \frac{C}{\dstar}\Fnorm D^2
 &\quad(D\in\mathbb T+\widehat{\mathbb T},\ 
         Q\in\{\widehat G,\widehat G_{\rm ev}\}),
         \label{eq:foldwise-energy-app}\\
 \max_{Q\in\{\widehat G,\widehat G_{\rm ev}\}}
 \norm{\widehat P_{\mathbb T}(Q-G)
       \widehat P_{\mathbb T}}_{\rm op}
 &\le \frac{Cr_n}{\dstar}.\label{eq:foldwise-info-conc-app}
\end{align}
There is an ambient orthogonal Procrustes map \(W\), equal to the
identity on \((\mathbb T+\widehat{\mathbb T})^\perp\), such that
\begin{equation}\label{eq:foldwise-transport-app}
 W\mathbb T=\widehat{\mathbb T},\qquad
 \norm{W-I}_{\rm op}\le Cr_n,\qquad
 \max_{\Gamma\in\mathcal F_{\rm gap}}
 \frac{\Fnorm{W^*\widehat P_{\mathbb T}\Gamma-P_{\mathbb T}\Gamma}}
      {\Fnorm{P_{\mathbb T}\Gamma}}
 \le \frac{Cr_n}{\alpha_0}.
\end{equation}
The same rowwise perturbation argument gives the induced-norm bounds
\begin{equation}\label{eq:foldwise-projector-infty-app}
 \norm{\widehat P_{\mathbb T}-P_{\mathbb T}}_{\infty\to\infty}
 \le Cr_n,
 \qquad
 \norm{\widehat P_{\mathbb T}}_{\infty\to\infty}
 \vee\norm{P_{\mathbb T}}_{\infty\to\infty}\le C.
\end{equation}
Moreover, for \(\Delta=\widehat\Theta-\Thetastar\),
\begin{equation}\label{eq:derived-normal-leakage-app}
 \infnorm{(I-P_{\mathbb T})\Delta}
 +\infnorm{(I-\widehat P_{\mathbb T})\Delta}
 \le Cr_n^2.
\end{equation}
\end{lemma}

\begin{proof}
We condition first on the training sample.  The model- and task-factor
bounds in Propositions~\ref{prop:left-final-app} and
\ref{prop:right-final-app}, together with the Gram lower bounds used in
their proofs, imply after their common Procrustes alignment
\[
 \twoninfnorm{\widehat L-L^\star}
 \le C\sqrt{\frac{\dt\bar d}{n}}\,\operatorname{polylog}(n\bar d),
 \qquad
 \twoninfnorm{\widehat A-A^\star}
 \le C\sqrt{\frac{\bar d}{\dt n}}\,
       \operatorname{polylog}(n\bar d).
\]
The second display follows by substituting the first bound into
\eqref{eq:task-factor-rate-app}; its \(n^{-1/2}\) term is no larger
because \(\bar d\ge\dt\).  The corresponding Frobenius bounds and the
signal-strength condition give, by Wedin's theorem,
\[
 \norm{P_{\widehat U}-P_{U^\star}}_{\rm op}
 +\norm{P_{\widehat V}-P_{V^\star}}_{\rm op}\le Cr_n.
\]
The closed-form projector~\eqref{eq:tangent-projector-matrix} and the
rowwise factor bounds sharpen this, for every score-gap atom, to
\[
 \Fnorm{(\widehat P_{\mathbb T}-P_{\mathbb T})\Gamma}
 \le Cr_n\sqrt{\frac{\bar d}{\dstar}}\,\Fnorm\Gamma.
\]
Applying the same formula to a canonical atom and summing the two
fixed-rank factor blocks shows
\[
 \max_\omega\oneNorm{(\widehat P_{\mathbb T}-P_{\mathbb T})E_\omega}
 \le Cr_n.
\]
Self-adjointness converts this row bound into the first assertion of
\eqref{eq:foldwise-projector-infty-app}; the second follows by the
triangle inequality and
\(\max_\omega\oneNorm{P_{\mathbb T}E_\omega}\le C\), proved directly
from~\eqref{eq:tangent-projector-matrix} and incoherence.
Let \(Q:\mathbb T\to\widehat{\mathbb T}\) be the canonical
principal-angle isometry
\[
 Q=\widehat P_{\mathbb T}P_{\mathbb T}
   (P_{\mathbb T}\widehat P_{\mathbb T}P_{\mathbb T})^{-1/2}
   \quad\text{on }\mathbb T.
\]
The polar correction differs from the identity by \(O(r_n^2)\), while
the preceding atomwise bound controls the first-order projector term.
Extending \(Q\) by the inverse principal-angle rotations and by the
identity on \((\mathbb T+\widehat{\mathbb T})^\perp\) gives \(W\).
Dividing the resulting
atomwise bound by Assumption~\ref{ass:alignment-app} proves
\eqref{eq:foldwise-transport-app}.

For the normal components, write the refinement factorization as
\(\widehat\Theta=\widehat A\widehat L^\top\) and
\(\Thetastar=A^\star(L^\star)^\top\), where the column spaces are the
corresponding left and right singular spaces.  Exact rank gives
\[
 (I-P_{U^\star})\Delta(I-P_{V^\star})
 =(I-P_{U^\star})(\widehat A-A^\star)
   (\widehat L-L^\star)^\top(I-P_{V^\star}),
\]
and the analogous identity with the true and estimated spaces
interchanged.  Applying
\(|e_t^\top BC^\top e_m|\le\twoninfnorm B\twoninfnorm C\)
and the two factor bounds proves
\eqref{eq:derived-normal-leakage-app}.

It remains to prove the information statements.  Conditional on the
nuisance-training sample, \(\widehat{\mathbb T}\) is fixed, has
dimension at most \(r(\dt+\dm)\), and the evaluation subsample is
independent of it.  The estimated-factor
incoherence gives
\[
 \sup_X\Fnorm{\widehat P_{\mathbb T}X}^2
 \le C\frac{\bar d}{\dstar}.
\]
Since \(I\) is Lipschitz and
\(\infnorm{\widehat\Theta-\Thetastar}\le Cr_n\), the finite design
average~\eqref{eq:Ghat-def-app} and
Lemma~\ref{lem:frob-reduction-app} give
\[
 \norm{\widehat P_{\mathbb T}(\widehat G-G)
       \widehat P_{\mathbb T}}_{\rm op}
 \le \frac{Cr_n}{\dstar}.
\]
Matrix Bernstein on the fixed sum space
\(\mathbb T+\widehat{\mathbb T}\), whose dimension is at most
\(2r(\dt+\dm)\), gives
\[
 \norm{\widehat P_{\mathbb T}
       (\widehat G_{\rm ev}-\widehat G)
       \widehat P_{\mathbb T}}_{\rm op}
 \le \frac{Cr_n}{\dstar}.
\]
The two displays prove~\eqref{eq:foldwise-info-conc-app}.  The
population lower and upper bounds hold on the whole row-centered space
by Lemma~\ref{lem:weighted-moment-app}; perturbing them first by the
Fisher-weight bias and then by the evaluation-sample Bernstein bound
proves~\eqref{eq:foldwise-energy-app} for all sufficiently large \(n\).
The fixed-fold union bound completes the proof.
\end{proof}

\begin{lemma}[Relative perturbation of the design-averaged Fisher map]
\label{lem:preconditioned-fisher-infty-app}
Condition on one fold's nuisance-training sample and abbreviate its
estimated tangent projector by \(\widehat P\).  Define the population-
weight operator on the estimated tangent space by
\[
 \widetilde A
 :=(\widehat P G\widehat P)|_{\widehat{\mathbb T}},
 \qquad
 \widetilde{\mathcal A}^{\dagger}
 :=\widehat P\widetilde A^{-1}\widehat P.
\]
Suppose
\[
 \frac{\norm{\widetilde{\mathcal A}^{\dagger}}_{\infty\to\infty}}
      {\dstar}\le C_0C_A,
 \qquad
 \norm{\widehat P}_{\infty\to\infty}\le C_0 .
\]
Then, on the foldwise refinement event,
\begin{equation}\label{eq:preconditioned-fisher-infty-app}
 \norm{\widetilde{\mathcal A}^{\dagger}\widehat P
       (\widehat G-G)\widehat P}_{\infty\to\infty}
 \le C C_A r_n .
\end{equation}
The conclusion holds simultaneously over the fixed number of folds.
\end{lemma}

\begin{proof}
Let \(\Delta=\widehat\Theta-\Thetastar\).  Since the design atom is
\(X=e_t(e_m-e_{m'})^\top\), Lipschitz continuity of \(I\) gives
\[
 \left|
 I(\langle X,\widehat\Theta\rangle)
 -I(\langle X,\Thetastar\rangle)
 \right|
 \le L_I|\Delta_{t,m}-\Delta_{t,m'}|
 \le 2L_I\infnorm\Delta.
\]
For each task, \(\widehat G-G\) is therefore a weighted graph
Laplacian whose edge weights differ by at most
\(C\infnorm\Delta\) times their sampling probabilities.  Summing the
absolute entries in one graph-Laplacian row and using near-uniform
task and pair sampling yields the exact induced-norm estimate
\begin{equation}\label{eq:Ghat-G-infty-app}
 \norm{\widehat G-G}_{\infty\to\infty}
 \le \frac{C}{\dstar}\infnorm\Delta
 \le \frac{Cr_n}{\dstar}.
\end{equation}
Consequently,
\[
 \begin{split}
 &\norm{\widetilde{\mathcal A}^{\dagger}\widehat P
       (\widehat G-G)\widehat P}_{\infty\to\infty}\\
 &\qquad\le
 \norm{\widetilde{\mathcal A}^{\dagger}}_{\infty\to\infty}
 \norm{\widehat P}_{\infty\to\infty}^2
 \norm{\widehat G-G}_{\infty\to\infty}
 \le C C_A r_n.
 \end{split}
\]
The \(C_A\dstar\) scale of the inverse cancels the \(1/\dstar\) scale
of the pairwise design, yielding a relative perturbation of order
\(C_A r_n\).
\end{proof}

\begin{theorem}[Foldwise plug-in inverse stability]
\label{thm:derived-foldwise-CA-app}
Under Assumptions~\ref{ass:score-app}--\ref{ass:population-CA-app}
and the conditions of Theorem~\ref{thm:entrywise}, suppose
\(C_A r_n=o(1)\).  Then, uniformly over the fixed number of
cross-fitting folds and with probability at least \(1-n^{-a}\) for
every fixed \(a>0\),
\begin{equation}\label{eq:CA-def-app}
\begin{aligned}
 \frac{\norm{\mathcal A^\dagger}_{\infty\to\infty}}{\dstar}
 \vee
 \norm{\mathcal A^\dagger P_{\mathbb T}G}_{\infty\to\infty}
 &\le C C_A,\\
 \max_k\left\{
 \frac{\norm{\widehat{\mathcal A}^{\dagger}_k}
             _{\infty\to\infty}}{\dstar}
 \vee
 \norm{\widehat{\mathcal A}^{\dagger}_k
       \widehat P_{\mathbb T}^{(-k)}G}_{\infty\to\infty}
 \right\}
 &\le C C_A,\\
 \max_k
 \norm{\widehat{\mathcal A}^{\dagger}_k
       \widehat P_{\mathbb T}^{(-k)}(\widehat G_k-G)
       \widehat P_{\mathbb T}^{(-k)}}_{\infty\to\infty}
 &\le C C_A r_n .
\end{aligned}
\end{equation}
\end{theorem}

\begin{proof}
The first inverse bound is the definition of \(C_A\), and the first
composite bound is~\eqref{eq:population-composite-CA-app}.  We prove
the two foldwise statements for one fold and suppress \(k\).

Set \(\widehat P=\widehat P_{\mathbb T}\), and define
\(\widetilde A=(\widehat P G\widehat P)|_{\widehat{\mathbb T}}\).
To compare inverses acting on different tangent spaces without
choosing coordinates, introduce the invertible full-space operators
\[
 B=P_{\mathbb T}GP_{\mathbb T}
   +\lambda_\perp(I-P_{\mathbb T}),
 \qquad
 \widetilde B=\widehat P G\widehat P
   +\lambda_\perp(I-\widehat P),
 \qquad
 \lambda_\perp=\frac{c_\perp}{\dstar},
\]
where \(c_\perp>0\) is fixed.  Their inverses satisfy
\[
 B^{-1}=\mathcal A^\dagger
         +\lambda_\perp^{-1}(I-P_{\mathbb T}),
 \qquad
 \widetilde{\mathcal A}^{\dagger}
 =\widehat P\widetilde B^{-1}\widehat P.
\]
By~\eqref{eq:population-CA-app},
\eqref{eq:foldwise-projector-infty-app}, and
\(\norm G_{\infty\to\infty}\le C/\dstar\),
\[
 \norm{B^{-1}}_{\infty\to\infty}\le C C_A\dstar,
 \qquad
 \norm{\widetilde B-B}_{\infty\to\infty}
 \le \frac{Cr_n}{\dstar}.
\]
Consequently
\(\norm{B^{-1}(\widetilde B-B)}_{\infty\to\infty}
\le C C_A r_n=o(1)\).  The resolvent identity, or equivalently the
Neumann series, now gives
\begin{equation}\label{eq:tilde-inverse-derived-app}
 \frac{\norm{\widetilde{\mathcal A}^{\dagger}}
             _{\infty\to\infty}}{\dstar}
 \le C C_A,
 \qquad
 \norm{\widetilde{\mathcal A}^{\dagger}\widehat P G}
             _{\infty\to\infty}
 \le C C_A.
\end{equation}
The second inequality follows either from the displayed inverse bound
and \(\norm{\widehat P G}_{\infty\to\infty}\le C/\dstar\), or directly
from the restricted normal equation.

Let
\[
 R
 :=\widetilde{\mathcal A}^{\dagger}\widehat P
   (\widehat G-G)\widehat P .
\]
Lemma~\ref{lem:preconditioned-fisher-infty-app} and
\eqref{eq:tilde-inverse-derived-app} imply
\(\norm R_{\infty\to\infty}\le C C_A r_n=o(1)\).
On \(\widehat{\mathbb T}\),
\[
 \widehat A
 =\widetilde A(I+R),
 \qquad
 \widehat{\mathcal A}^{\dagger}
 =(I+R)^{-1}\widetilde{\mathcal A}^{\dagger}.
\]
Therefore \(\norm{(I+R)^{-1}}_{\infty\to\infty}\le2\) for all
sufficiently large \(n\), and
\[
 \frac{\norm{\widehat{\mathcal A}^{\dagger}}
             _{\infty\to\infty}}{\dstar}
 \vee
 \norm{\widehat{\mathcal A}^{\dagger}\widehat P G}
             _{\infty\to\infty}
 \le C C_A,
 \qquad
 \norm{\widehat{\mathcal A}^{\dagger}\widehat P
       (\widehat G-G)\widehat P}_{\infty\to\infty}
 \le 2\norm R_{\infty\to\infty}.
\]
This is~\eqref{eq:CA-def-app}.  The fixed-fold union bound finishes the
proof.
\end{proof}
%\end{revisionblock}

This closed form makes each remainder term in
Appendix~\ref{app:fdim-proof-decomp} computable in closed form.  In
particular, for a sparse score-gap contrast
\(\Gamma=e_t(e_m-e_{m'})^\top\), the projection
\(P_{\mathbb T}\Gamma\) is a low-rank row-centered matrix.  Incoherence
gives its upper envelope and Assumption~\ref{ass:alignment-app} supplies
the lower bound needed for studentization.

\subsection{One-step expansion and single-contrast remainder}\label{app:fdim-proof-decomp}

We give the argument for one fold and condition on its training samples.  The resulting
\(\widehat\Theta,\widehat{\mathbb T},\widehat G\), and
\(\widehat H_\Gamma\) are then fixed and are independent of the
evaluation observations.  Put
\(\Delta=\widehat\Theta-\Thetastar\),
\(\widehat s=s_\eta(Y,\ip{X}{\widehat\Theta})\), and
\(s^\star=s_\eta(Y,\ip{X}{\Thetastar})\).  If \(P\) denotes expectation
over a fresh observation and
\[
 g_\Gamma(W)
 :=\widehat s\ip{\widehat H_\Gamma}{X}
   -s^\star\ip{H_\Gamma^\star}{X},
\]
then \(P[s^\star\ip{H_\Gamma^\star}{X}]=0\), and the exact
decomposition is
\begin{align}
 \widehat\psi_\Gamma-\psi_\Gamma(\Thetastar)
 ={}&\Pn\!\left[s^\star\ip{H_\Gamma^\star}{X}\right]
    +(\Pn-P)g_\Gamma+B_\Gamma,                     \label{eq:fdim-decomp-app}\\
 B_\Gamma
 :={}&\ip{\Gamma}{\Delta}
       +P\!\left[\widehat s\ip{\widehat H_\Gamma}{X}\right].
 \nonumber
\end{align}
%\begin{revisionblock}
This form is the reason for cross-fitting: the direction is fixed
under the conditional law governing \((\Pn-P)g_\Gamma\).
%\end{revisionblock}

Let \(G\) be the population Fisher operator at \(\Thetastar\).
Conditional expectation over \(Y\) and Taylor's formula give
\begin{align}
 P\!\left[\widehat s\ip{\widehat H_\Gamma}{X}\right]
 ={}&-\ip{G\Delta}{\widehat H_\Gamma}+R_{2,\Gamma},
                                                        \label{eq:population-taylor-app}\\
 |R_{2,\Gamma}|
 \le{}& {L_I\over2}
 P\!\left[\ip{X}{\Delta}^2
            \left|\ip{\widehat H_\Gamma}{X}\right|\right],
 \nonumber
\end{align}
where \(L_I:=\sup_\eta|I'(\eta)|<\infty\).  Consequently,
\begin{align}
 B_\Gamma
 ={}&\ip{\Gamma-G\widehat H_\Gamma}
              {(I-\widehat P_{\mathbb T})\Delta}       \nonumber\\
 &+\ip{\widehat P_{\mathbb T}(\widehat G-G)
                 \widehat H_\Gamma}
             {\widehat P_{\mathbb T}\Delta}
   +R_{2,\Gamma}.                                      \label{eq:bias-decomp-app}
\end{align}
Indeed, the estimating equation
\(\widehat P_{\mathbb T}\widehat G\widehat H_\Gamma
=\widehat P_{\mathbb T}\Gamma\) replaces the tangent part of
\(\Gamma-G\widehat H_\Gamma\) by
\(\widehat P_{\mathbb T}(\widehat G-G)\widehat H_\Gamma\).

The exact-rank geometry makes the apparently linear normal term
second order.  If
\(\widehat\Theta=\widehat U\widehat\Sigma\widehat V^\top\) is the
rank-\(r\), row-centered fold estimate, then
\[
 (I-P_{\mathbb T})\Delta
 =(I-P_{U^\star})\widehat\Theta(I-P_{V^\star}).
\]
Likewise, because
\((I-P_{\widehat U})\widehat\Theta(I-P_{\widehat V})=0\),
\[
 (I-\widehat P_{\mathbb T})\Delta
 =-(I-P_{\widehat U})\Thetastar(I-P_{\widehat V}).
\]
Write each right-hand side as a product of its two projected
square-root singular factors and apply
\(|e_t^\top AB^\top e_m|\le\twoninfnorm A\twoninfnorm B\).
 Lemma~\ref{lem:foldwise-geometry-app} then gives
\begin{equation}\label{eq:normal-leakage-app}
 \infnorm{(I-P_{\mathbb T})\Delta}
 +\infnorm{(I-\widehat P_{\mathbb T})\Delta}
 \le C r_n^2 .
\end{equation}
Both sides contain one left and one right subspace error; this is the
source of the square.

We next control the two non-leading terms in
\eqref{eq:fdim-decomp-app}.  To compare operators living on different
tangent spaces, use the ambient orthogonal Procrustes map \(W\) from
Lemma~\ref{lem:foldwise-geometry-app}.  Put
\(E_A:=W^*\widehat A_{\widehat{\mathbb T}}W-A_{\mathbb T}\).
Adding and subtracting
\(W^*P_{\widehat{\mathbb T}}GP_{\widehat{\mathbb T}}W\), then using
\eqref{eq:foldwise-info-conc-app} and
\eqref{eq:foldwise-transport-app}, yields
\[
 \norm{E_A}_{\rm op}\le {C r_n\over\dstar},
\]
because \(\norm G_{\rm op}\le C/\dstar\) on the row-centered space.
The restricted-energy lower bound gives
\(\norm{A_{\mathbb T}^{-1}}_{\rm op}\vee
\norm{(W^*\widehat A_{\widehat{\mathbb T}}W)^{-1}}_{\rm op}
\le C\dstar\).  The common-space resolvent identity
\[
 (W^*\widehat A_{\widehat{\mathbb T}}W)^{-1}-A_{\mathbb T}^{-1}
 =-(W^*\widehat A_{\widehat{\mathbb T}}W)^{-1}E_AA_{\mathbb T}^{-1},
\]
together with the transported-contrast bound in
\eqref{eq:foldwise-transport-app}, gives
\[
 {\Fnorm{\widehat H_\Gamma-H_\Gamma^\star}
       \over\Fnorm{H_\Gamma^\star}}
 \le {C r_n\over\alpha_0}.
 \tag{D.1}
\]
Here we first obtain the bound for
\(W^*\widehat H_\Gamma-H_\Gamma^\star\); the identity
\(\widehat H_\Gamma-H_\Gamma^\star
=W(W^*\widehat H_\Gamma-H_\Gamma^\star)+(W-I)H_\Gamma^\star\)
and \eqref{eq:foldwise-transport-app} give the displayed
ambient Frobenius bound.
The lower alignment bound prevents division by a vanishing direction.
Thus the common \(A_{\mathbb T}^{-1}\) scale cancels; no inverse on the
full ambient space is used.

For completeness, the bounds needed for the population bias are as
follows.  Near-uniform comparison sampling gives, for every
row-centered matrix \(D\),
\[
 P|\ip{D}{X}|\le {C\over\dstar}\oneNorm D.
\]
The closed-form projector, incoherence, and its foldwise analogue give
\[
 \max_E\{\oneNorm{P_{\mathbb T}E},
          \oneNorm{\widehat P_{\mathbb T}E}\}\le C,
 \qquad
 \max_{\Gamma\in\mathcal F_{\rm gap}}
 \{\infnorm{P_{\mathbb T}\Gamma},
   \infnorm{\widehat P_{\mathbb T}\Gamma}\}
 \le C{\bar d\over\dstar}\oneNorm\Gamma,              \tag{D.2}
\]
where the maximum in the first display is over canonical matrix
atoms.  Indeed,
\(\oneNorm{P_Ue_t}\le\sqrt{\mu r}\),
\(\oneNorm{(I-P_U)e_t}\le1+\sqrt{\mu r}\),
\(\oneNorm{P_Ve_m}\le\sqrt{\mu r}\), and
\(\oneNorm{R_{\rm m}e_m}\le2\); substituting these four bounds into
\eqref{eq:tangent-projector-matrix} proves the first assertion, and
the same formula gives the second.  The estimated-space assertions
follow from the foldwise incoherence bounds.
Self-adjointness, the definition of the restricted inverse,
(D.2), and \eqref{eq:CA-def-app} therefore imply
\[
 \oneNorm{G\widehat H_\Gamma}
 \le C C_A\oneNorm\Gamma,
 \qquad
 \oneNorm{\widehat H_\Gamma}
 \le C C_A\dstar\oneNorm\Gamma.                       \tag{D.3}
\]
The first inequality uses
\(\|G\widehat P_{\mathbb T}\widehat A_{\widehat{\mathbb T}}^{-1}
\widehat P_{\mathbb T}\|_{1\to1}
=\|\widehat P_{\mathbb T}\widehat A_{\widehat{\mathbb T}}^{-1}
\widehat P_{\mathbb T}G\|_{\infty\to\infty}\); the second uses the
normalized inverse bound and the first part of (D.2).  Combining (D.3) with
\eqref{eq:normal-leakage-app} bounds the first term in
\eqref{eq:bias-decomp-app} by
\(C C_A\oneNorm\Gamma r_n^2\).

The projector bound in (D.2) gives
\(\infnorm{\widehat P_{\mathbb T}\Delta}\le Cr_n\).
By symmetry and the last composite bound in
\eqref{eq:CA-def-app}, the second term in
\eqref{eq:bias-decomp-app} is therefore at most
\[
 \oneNorm\Gamma\,
 \norm{\widehat A_{\widehat{\mathbb T}}^{-1}
       \widehat P_{\mathbb T}(\widehat G-G)
       \widehat P_{\mathbb T}}_{\infty\to\infty}
 \infnorm{\widehat P_{\mathbb T}\Delta}
 \le C C_A\oneNorm\Gamma r_n^2.
\]
Finally, \(\infnorm\Delta\le Cr_n\), (D.3), and
\eqref{eq:population-taylor-app} give the same bound for
\(R_{2,\Gamma}\).  Hence
\begin{equation}
 |B_\Gamma|\le C C_A\oneNorm\Gamma r_n^2.              \label{eq:bias-bound-app}
\end{equation}

It remains only to center \(g_\Gamma\).  Decompose it as
\[
 g_\Gamma
 =s^\star\ip{\widehat H_\Gamma-H_\Gamma^\star}{X}
  +(\widehat s-s^\star)\ip{\widehat H_\Gamma}{X}.
\]
The score is bounded,
\(|\widehat s-s^\star|\le C|\ip{X}{\Delta}|\le Cr_n\), and
(D.1)--(D.3), Fisher comparability, and the projected-atom envelope
control the conditional variance and envelope of the two summands.
Conditional Bernstein
therefore gives, for every fixed \(a>0\),
\begin{equation}
 |(\Pn-P)g_\Gamma|
 \le C C_A\oneNorm\Gamma
       \left\{r_n\sqrt{\frac{\bar d\log(n\bar d)}{n}}
                    +\frac{\bar d\log(n\bar d)}{n}\right\}
 \le C C_A\oneNorm\Gamma r_n^2                       \label{eq:centered-g-bound-app}
\end{equation}
with probability at least \(1-n^{-a}\).  %This conditioning is valid because neither training subsample contains an evaluation observation.
This proves the following theorem.

\begin{theorem}[Single-contrast remainder bound]\label{thm:single-contrast-remainder-app}
Under Assumptions~\ref{ass:score-app}--\ref{ass:sample-app},
\ref{ass:alignment-app}, and
{Assumption~\ref{ass:population-CA-app}}, for every admissible
score-gap contrast and every fixed \(a>0\), with probability at least
\(1-n^{-a}\),
\begin{equation}\label{eq:single-contrast-remainder-app}
 \left|
 \widehat\psi_\Gamma-\psi_\Gamma(\Thetastar)
 -\frac1n\sum_{i=1}^n\phi_\Gamma(W_i)
 \right|
 \le
 C C_A\oneNorm\Gamma\,
 {\,\bar d\log^c(n\bar d)\over n}.
\end{equation}
The same conclusion holds after averaging the fixed number of folds.
\end{theorem}

\begin{proof}
Apply the preceding expansion on every fold with
\(x=(a+C)\log(n\bar d)\).  The fold estimates are trained off their
evaluation folds, and {the information directions are deterministic
functions of those training estimates and the known design law}, so the stated conditional
Bernstein bounds apply.  Combine
\eqref{eq:fdim-decomp-app}, \eqref{eq:bias-bound-app}, and
\eqref{eq:centered-g-bound-app}.  Since
\(r_n^2=\bar d\log^c(n\bar d)/n\), the displayed bound follows.  A union
bound over the fixed number of folds preserves failure probability
\(n^{-a}\) after increasing \(C\).
\end{proof}
\subsection{Uniform single-contrast remainder over a contrast family}\label{app:fdim-proof-uniform}

The bound in
Theorem~\ref{thm:single-contrast-remainder-app} is per-contrast.
We extend it uniformly over a polynomial-size family by a careful union
bound.

\begin{theorem}[Uniform one-step remainder]\label{thm:uniform-remainder-app}
Let \(\mathcal F\subset\R^{\dt\times\dm}\) be any family of contrasts of
size \(|\mathcal F|\le(\dt\dm)^{C_F}\) for an absolute constant
\(C_F\), such that each \(\Gamma\in\mathcal F\) satisfies
Assumption~\ref{ass:alignment-app} with a uniform alignment constant
\(\alpha_{\rm min}>0\) and Assumption~\ref{ass:bounded-Gamma-app} with
uniform constants \(M,C_\psi\).  Fix any \(a>0\).  Then with probability
at least \(1-n^{-a}\),
\begin{equation}\label{eq:uniform-remainder-app}
    \max_{\Gamma\in\mathcal F}\,
    \frac{\sqrt n\,|R_n^\Gamma|}{\sigma_\Gamma}
    \;\le\;
    C\,C_A\,\sqrt{\frac{\bar d\,\mathrm{polylog}(n\bar d)}{n}},
\end{equation}
where \(\sigma_\Gamma^2=V_{\rm eff}(\Gamma)\).
The polylog absorbs both the logarithmic factor inherited from
Theorem~\ref{thm:single-contrast-remainder-app} and the
\(\log|\mathcal F|\le C_F\log(\bar d)\) factor from the union bound.
\end{theorem}

\begin{proof}
The proof is by union bound on top of
Theorem~\ref{thm:single-contrast-remainder-app}, with care taken so
that no \(\Gamma\)-dependent constants degrade.

\textbf{Step 1: per-contrast studentization.}
Fisher comparability and uniform alignment give
\[
 \sigma_\Gamma^2
 \asymp \dstar\Fnorm{P_{\mathbb T}\Gamma}^2
 \ge c\alpha_{\rm min}^2\bar d\Fnorm\Gamma^2.
\]
Since every contrast has at most \(M\) nonzero entries,
\(\oneNorm\Gamma\le\sqrt M\Fnorm\Gamma\), and hence
\begin{equation}\label{eq:gamma-over-sigma-app}
 {\oneNorm\Gamma\over\sigma_\Gamma}\le {C\over\sqrt{\bar d}}
 \qquad(\Gamma\in\mathcal F).
\end{equation}
Combining~\eqref{eq:gamma-over-sigma-app} with the raw remainder
bound~\eqref{eq:single-contrast-remainder-app}, with its free tail
parameter \(x\), gives
\[
 \frac{\sqrt n|R_n^\Gamma|}{\sigma_\Gamma}
 \le K C_A\sqrt{\frac{\bar d\,\log^{c'}(n\bar d)}{n}}
\]
outside an event of probability at most \(2e^{-x}\); the fixed exponent
\(c'\) absorbs the logarithms generated by the constituent Bernstein
bounds.

\textbf{Step 2: union bound over \(\mathcal F\).}
Set \(x=Ca\log(n\bar d)+C_F\log\bar d\) and apply the per-contrast bound
to each \(\Gamma\in\mathcal F\); a union bound over the
\(|\mathcal F|\le\bar d^{C_F}\) contrasts gives failure probability
at most \(2|\mathcal F|e^{-x}\le2n^{-a}\).

\textbf{Step 3: rate.}
Absorbing the tail and family-size logarithms into the fixed polylog
gives~\eqref{eq:uniform-remainder-app}.
\end{proof}

\subsection{Multivariate Berry--Esseen for the leading term}\label{app:fdim-proof-be}

We now establish the rate \(O(\sqrt{\bar d/n})\) directly for the leading
i.i.d.\ EIF average.
This scalar result is useful for pointwise inference; the fixed-\(q\)
rectangle result below is obtained separately by coordinate-standardized
CCK approximation, without polarization through potentially cancelling
contrasts.

\begin{theorem}[Berry--Esseen for the standardized leading term]\label{thm:fdim-be-app}
Fix a single contrast \(\Gamma\) satisfying
Assumptions~\ref{ass:bounded-Gamma-app}--\ref{ass:alignment-app}, and
let \(Z_i:=\phi_\Gamma(W_i)/\sigma_\Gamma\) be the standardized oracle
EIF coordinates with \(\sigma_\Gamma^2=V_{\rm eff}(\Gamma)\).  Then
\begin{equation}\label{eq:fdim-be-app}
    \rho_n
    :=
    \sup_{t\in\R}
    \Bigl|
    \Pr\Bigl(\frac{1}{\sqrt n}\sum_{i=1}^n Z_i\le t\Bigr)
    -
    \Phi(t)
    \Bigr|
    \;\le\;
    C\,\sqrt{\frac{\bar d}{n}}.
\end{equation}
\end{theorem}

\begin{proof}
We compute the second and third moments of \(Z_i\) explicitly and
apply the classical (univariate) Berry--Esseen theorem.

\textbf{Step 1: mean zero.}
\(\E^\star Z_i=\sigma_\Gamma^{-1}\E^\star\phi_\Gamma(W_i)=\sigma_\Gamma^{-1}\E^\star[s_\eta(Y,\eta^\star)\ip{H_\Gamma^\star}{X}]\),
which vanishes by the score-centering identity
\(\E^\star[s_\eta(Y,\eta^\star)\mid X]=0\).

\textbf{Step 2: second moment.}
By the definition of \(\sigma_\Gamma^2\) and Fisher comparability
(Lemma~\ref{lem:weighted-moment-app}),
\[
    \E^\star Z_i^2
    =
    \frac{1}{\sigma_\Gamma^2}\E^\star[s_\eta^2\ip{H_\Gamma^\star}{X}^2]
    =
    \frac{1}{\sigma_\Gamma^2}\ip{H_\Gamma^\star}{AH_\Gamma^\star}
    =
    \frac{V_{\rm eff}(\Gamma)}{\sigma_\Gamma^2}=1.
\]
By the Frobenius reduction (Lemma~\ref{lem:frob-reduction-app}) applied
to \(H_\Gamma^\star\),
\(\E^\star\ip{H_\Gamma^\star}{X}^2\asymp\Fnorm{H_\Gamma^\star}^2/\dstar\),
and Fisher comparability gives
\(\sigma_\Gamma^2\asymp\Fnorm{H_\Gamma^\star}^2/\dstar\), so
\begin{equation}\label{eq:fdim-be-step2}
    \frac{c_B}{\dstar}\Fnorm{H_\Gamma^\star}^2\le\sigma_\Gamma^2\le\frac{C_B}{\dstar}\Fnorm{H_\Gamma^\star}^2.
\end{equation}

\textbf{Step 3: third absolute moment.}
\[
    \E^\star|Z_i|^3
    =
    \frac{1}{\sigma_\Gamma^3}\E^\star[|s_\eta|^3|\ip{H_\Gamma^\star}{X}|^3]
    \le
    \frac{C_3}{\sigma_\Gamma^3}\E^\star|\ip{H_\Gamma^\star}{X}|^3
\]
where \(C_3:=\E^\star[|s_\eta(Y,\eta^\star)|^3\mid X]\le 1\) under
Assumption~\ref{ass:score-app}(iv) (since
\(|s_\eta|\le 1\)).  Now use the elementary inequality
\(\E|W|^3\le(\sup|W|)\E|W|^2\) with \(W=\ip{H_\Gamma^\star}{X}\):
\[
    \E^\star|\ip{H_\Gamma^\star}{X}|^3
    \le
    \Bigl(\sup_x|\ip{H_\Gamma^\star}{x}|\Bigr)\,
    \E^\star\ip{H_\Gamma^\star}{X}^2.
\]
Since \(H_\Gamma^\star\in\mathbb T\),
\(\ip{H_\Gamma^\star}{x}=\ip{H_\Gamma^\star}{P_{\mathbb T}x}\), so by
Cauchy--Schwarz and the tangent-projection envelope
\(\sup_x\Fnorm{P_{\mathbb T}x}\lesssim\sqrt{\bar d/\dstar}\) (which
follows from the basis-tensor projection bound and triangle inequality
for pairwise differences; see
Lemma~\ref{lem:tangent-envelope-app} below),
\begin{equation}\label{eq:fdim-be-tangent-envelope}
    \sup_{x\in\mathcal X}|\ip{H_\Gamma^\star}{x}|
    \le
    \Fnorm{H_\Gamma^\star}\sup_x\Fnorm{P_{\mathbb T}x}
    \;\lesssim\;
    \Fnorm{H_\Gamma^\star}\sqrt{\frac{\bar d}{\dstar}}.
\end{equation}
Combined with \(\E^\star\ip{H_\Gamma^\star}{X}^2\asymp\Fnorm{H_\Gamma^\star}^2/\dstar\),
\[
    \E^\star|\ip{H_\Gamma^\star}{X}|^3
    \;\lesssim\;
    \Fnorm{H_\Gamma^\star}\sqrt{\frac{\bar d}{\dstar}}\cdot\frac{\Fnorm{H_\Gamma^\star}^2}{\dstar}
    =
    \frac{\Fnorm{H_\Gamma^\star}^3\sqrt{\bar d}}{(\dstar)^{3/2}}.
\]
Substituting into the expression for \(\E^\star|Z_i|^3\), and using
\(\sigma_\Gamma^3\ge c_B^{3/2}\Fnorm{H_\Gamma^\star}^3/(\dstar)^{3/2}\)
from~\eqref{eq:fdim-be-step2},
\[
    \E^\star|Z_i|^3
    \;\le\;
    \frac{C_3}{\sigma_\Gamma^3}\cdot\frac{\Fnorm{H_\Gamma^\star}^3\sqrt{\bar d}}{(\dstar)^{3/2}}
    \;\le\;
    C_3 c_B^{-3/2}\,\sqrt{\bar d}.
\]

\textbf{Step 4: Berry--Esseen.}
By the classical univariate Berry--Esseen theorem (e.g.\ Shevtsova
2010 with constant \(C_{\rm BE}=0.4748\)),
\[
    \rho_n
    \le
    \frac{C_{\rm BE}}{\sqrt n}\E^\star|Z_i|^3
    \le
    C_3 c_B^{-3/2}\,C_{\rm BE}\,\sqrt{\frac{\bar d}{n}}.
\]
\end{proof}

\begin{lemma}[Pairwise tangent-projection envelope]\label{lem:tangent-envelope-app}
Under \(\mu\)-incoherence and the row-centering gauge, for every
admissible design tensor \(X=e_t(e_m-e_{m'})^\top\),
\(\Fnorm{P_{\mathbb T}X}\le C(\mu,r)\sqrt{\bar d/\dstar}\).  The same
bound holds for \(\widehat P_{\mathbb T}X\) under the estimated
incoherence guarantee.
\end{lemma}

\begin{proof}
Decompose \(X=E_{(t,m)}-E_{(t,m')}\) where \(E_\omega=e_te_m^\top\) is
the canonical basis tensor.  By the basis-tensor projection bound (a
direct consequence of the closed-form
projector~\eqref{eq:tangent-projector-matrix} together with
\(\mu\)-incoherence),
\(\Fnorm{P_{\mathbb T}E_\omega}\le C(\mu,r)\sqrt{\bar d/\dstar}\) for
every basis tensor.  By linearity and the triangle inequality,
\(\Fnorm{P_{\mathbb T}X}\le\Fnorm{P_{\mathbb T}E_{(t,m)}}+\Fnorm{P_{\mathbb T}E_{(t,m')}}\le 2C(\mu,r)\sqrt{\bar d/\dstar}\).
\end{proof}

\subsection{Multivariate rectangle CLT for fixed \texorpdfstring{$q$}{q}}\label{app:fdim-proof-clt}

We combine a coordinate-standardized rectangle Gaussian approximation
with the uniform remainder bound of
Theorem~\ref{thm:uniform-remainder-app} to obtain the rectangle CLT.
All objects in this subsection may depend on \(n\); subscripts \(n\) are
suppressed unless they are needed to emphasize the triangular-array limit.

Let \(\Phi_q(W_i):=(\phi_1(W_i),\ldots,\phi_q(W_i))^\top\),
\(S_n:=n^{-1/2}\sum_{i=1}^n\Phi_q(W_i)\),
\(T_n:=\sqrt n(\widehat\psi-\psi)\), and
\(\mathsf R_n:=\sqrt n(R_n^{\Gamma_1},\ldots,R_n^{\Gamma_q})^\top\),
so that \(T_n=S_n+\mathsf R_n\).

\begin{theorem}[Oracle rectangle Gaussian approximation]\label{thm:fdim-be-multi-app}
For any rectangle \(B\in\mathcal R_q\),
\[
    \rho_n^{\rm orac}
    :=
    \sup_{B\in\mathcal R_q}|\Pr(S_n\in B)-\Pr(Z_\Gamma\in B)|
    \;\le\;
    \delta_n^{\rm orac}
    :=C\left\{\frac{\bar d\log^7(nq)}n\right\}^{1/6},
\]
where \(Z_\Gamma\sim\mathcal N(0,\Sigma)\).  No lower bound on the
smallest eigenvalue of \(\Sigma\), or of its correlation matrix, is
required.
\end{theorem}

\begin{proof}
Let
\[
 D:=\operatorname{diag}(\sigma_1,\ldots,\sigma_q),\qquad
 \widetilde\Phi_q:=D^{-1}\Phi_q,
 \qquad R:=D^{-1}\Sigma D^{-1}.
\]
Uniform alignment implies \(\sigma_j>0\).  Each coordinate of
\(\widetilde\Phi_q\) is mean zero and has variance one.  Moreover, the
tangent-projection envelope and the Fisher lower bound used in
Theorem~\ref{thm:fdim-be-app} give, coordinatewise,
\[
 \left|\frac{\phi_j(W_i)}{\sigma_j}\right|^2\le C\bar d.
\]
Put \(B_n=C\sqrt{\bar d}\).  Since the standardized coordinates have
variance one and are bounded by \(B_n\), for \(k=1,2\),
\[
 \E\left|\frac{\phi_j}{\sigma_j}\right|^{2+k}
 \le B_n^k\E\left|\frac{\phi_j}{\sigma_j}\right|^2=B_n^k,
 \qquad
 \E\exp\!\left(\frac{|\phi_j|}{B_n\sigma_j}\right)\le e.
\]
Thus the moment and exponential-tail conditions of the bounded-coordinate
specialization of the rectangle Gaussian approximation theorem
\citep[Theorem~2.1]{chernozhukov2017central} hold, and it yields
\[
 \sup_{B\in\mathcal R_q}
 \left|
 \Pr\!\left\{\frac1{\sqrt n}\sum_{i=1}^n
                 \widetilde\Phi_q(W_i)\in B\right\}
 -\Pr\{N(0,R)\in B\}
 \right|
 \le C\left\{\frac{\bar d\log^7(nq)}n\right\}^{1/6}.
\]
This theorem requires only the coordinate variance lower bound, here
equal to one; it permits singular and arbitrarily nearly singular
correlation matrices.  Finally, multiplication by the positive diagonal
matrix \(D\) maps rectangles to rectangles, which gives the asserted
bound for \(S_n\) and \(Z_\Gamma\) without whitening or dividing by a
small eigenvalue of \(R\).
When \(q=1\), the same conclusion follows from the scalar Berry--Esseen
bound in Theorem~\ref{thm:fdim-be-app} (equivalently, one may duplicate the
single coordinate before applying the rectangle theorem).
\end{proof}

\begin{theorem}[Rectangle CLT for the feasible statistic; restatement of Theorem~\ref{thm:fdim-clt}]\label{thm:fdim-clt-app}
Under
Assumptions~\ref{ass:score-app}--\ref{ass:sample-app},
\[
    \sup_{B\in\mathcal R_q}|\Pr(T_n\in B)-\Pr(Z_\Gamma\in B)|
    \;\lesssim\;
    C_A\sqrt{\frac{\bar d\,\mathrm{polylog}(n\bar d)}{n}}
    +\delta_n^{\rm orac}.
\]
Consequently, if \(C_A\sqrt{\bar d\,\mathrm{polylog}(n\bar d)/n}\to 0\),
the rectangle distance to the triangular-array Gaussian law
\(N(0,\Sigma_n)\) tends to zero.  Here \(Z_\Gamma=Z_{\Gamma,n}\) and
\(\Sigma=\Sigma_n\) in the preceding display, and
Assumption~\ref{ass:sample-app}, with its fixed \(c_0\ge7\), also
implies \(\delta_n^{\rm orac}\to0\).
\end{theorem}

\begin{proof}
Let \(D:=\mathrm{diag}(\sqrt{\Sigma_{11}},\ldots,\sqrt{\Sigma_{qq}})\).
Uniform alignment makes every diagonal entry positive.  Define standardized vectors
\(\bar T_n:=D^{-1}T_n\), \(\bar S_n:=D^{-1}S_n\),
\(\bar Z:=D^{-1}Z_\Gamma\).  The standardized remainder event is
\[
    \mathcal E_r:=\bigl\{\infnorm{D^{-1}\mathsf R_n}\le\delta_n\bigr\},
    \qquad
    \delta_n:=C\,C_A\sqrt{\bar d\,\mathrm{polylog}(n\bar d)/n}.
\]
By Theorem~\ref{thm:uniform-remainder-app} applied to the fixed family
\(\{\Gamma_1,\ldots,\Gamma_q\}\),
\(\pi_n:=\Pr(\mathcal E_r^c)\le n^{-a}\) after relabelling.

For any rectangle \(B=\prod_{j=1}^q[\alpha_j,\beta_j]\), define enlarged
and shrunk rectangles
\(B^{\pm\delta_n}:=\prod[\alpha_j\mp\delta_n,\beta_j\pm\delta_n]\) (with
the convention that an empty interval makes
\(B^{-\delta_n}=\emptyset\)).
On \(\mathcal E_r\), \(\bar T_n\in B\) implies
\(\bar S_n\in B^{+\delta_n}\), and \(\bar S_n\in B^{-\delta_n}\) implies
\(\bar T_n\in B\).  Hence
\[
    \Pr(\bar S_n\in B^{-\delta_n})-\pi_n
    \le
    \Pr(\bar T_n\in B)
    \le
    \Pr(\bar S_n\in B^{+\delta_n})+\pi_n.
\]
Combining with the oracle Berry--Esseen
(Theorem~\ref{thm:fdim-be-multi-app}, applied to standardized
\(\bar S_n\)),
\[
    |\Pr(\bar T_n\in B)-\Pr(\bar Z\in B)|
    \le
    \rho_n^{\rm orac}+\pi_n+\Pr(\bar Z\in B^{+\delta_n}\setminus B^{-\delta_n}).
\]
The Gaussian boundary band is bounded by Gaussian anti-concentration:
each face contributes at most \(\delta_n\sqrt{2/\pi}\) (one-dimensional
standard normal density at most \(\sqrt{1/(2\pi)}\)), and there are
\(2q\) faces, so the band is at most \(C q\delta_n\).  Since
\(q=O(1)\), \(\pi_n\le n^{-a}\), and
\(\rho_n^{\rm orac}\le\delta_n^{\rm orac}\), the total is
\(\lesssim C_A\sqrt{\bar d\,\mathrm{polylog}(n\bar d)/n}
+\delta_n^{\rm orac}\).  Linear
coordinate rescaling maps rectangles to rectangles, so the bound
transfers from \(\bar T_n\) to \(T_n\) without loss.
\end{proof}

\subsection{Plug-in variance and covariance consistency}\label{app:fdim-proof-varcons}

Let
\(\overline\phi_\Gamma:=\Pn\widehat\phi_\Gamma\) and
\(\widehat V_\Gamma=\Pn[(\widehat\phi_\Gamma-\overline\phi_\Gamma)^2]\),
with foldwise quantities understood.
The transported-operator comparison (D.1) and
Fisher comparability imply, uniformly over the score-gap family,
\begin{equation}\label{eq:relative-direction-app}
 {\Fnorm{\widehat H_\Gamma-H_\Gamma^\star}
  \over\Fnorm{H_\Gamma^\star}}\le Cr_n/\alpha_0.
\end{equation}

\begin{proposition}[Relative variance consistency]\label{prop:fdim-var-cons-app}
For every fixed \(a>0\), with probability at least \(1-n^{-a}\),
\begin{equation}\label{eq:fdim-var-cons-app}
 \max_{\Gamma}
 \left|\frac{\widehat V_\Gamma}{V_\Gamma}-1\right|
 \le Cr_n,\qquad
 r_n=\sqrt{\frac{\bar d\log^c(n\bar d)}n}.
\end{equation}
No \(C_A\) factor is required.
\end{proposition}

\begin{proof}
It suffices to treat one evaluation fold; averaging the fixed number of
folds changes only constants.  Let
\[
\mathcal G_k:=\sigma\!\left({\mathcal D_{\mathcal I_k^c}},
                    \{X_i:i\text{ is in the evaluation fold}\}\right).
\]
Conditional on \(\mathcal G_k\), the fitted matrix and
\(\widehat H_\Gamma\) are fixed, while the evaluation responses are
independent.  Write
\[
 p_i^\star:=\sigma(\eta_i^\star),\quad
 \widehat p_i:=\sigma(\widehat\eta_i),\quad
 h_i:=\ip{\widehat H_\Gamma}{X_i},\quad
 h_i^\star:=\ip{H_\Gamma^\star}{X_i}.
\]
The conditional second moment of the plug-in score is
\begin{equation}\label{eq:plugin-score-cond-second-app}
 m_i:=\E(\widehat s_i^2\mid\mathcal G_k)
 =I(\eta_i^\star)+(p_i^\star-\widehat p_i)^2.
\end{equation}

\textbf{Step 1: empirical scale of the fitted direction.}
The bounded-logit event, the combined-tangent restricted-energy
condition, \eqref{eq:relative-direction-app}, and Fisher comparability
give
\begin{equation}\label{eq:Hhat-empirical-scale-app}
 \Pn h_i^2\le C V_\Gamma,\qquad
 \max_i\frac{h_i^2}{V_\Gamma}\le C\bar d.
\end{equation}
For the first inequality, use
\(\Pn h_i^2\le(c_B')^{-1}\Pn[I(\widehat\eta_i)h_i^2]
\le C\Fnorm{\widehat H_\Gamma}^2/\dstar\) and
\(\Fnorm{\widehat H_\Gamma}\asymp\Fnorm{H_\Gamma^\star}\).
For the second, use
\(h_i=\ip{\widehat H_\Gamma}{\widehat P_{\mathbb T}X_i}\), the estimated
tangent envelope, and
\(V_\Gamma\asymp\Fnorm{H_\Gamma^\star}^2/\dstar\).

\textbf{Step 2: response fluctuation conditional on all evaluation
covariates.}
By \eqref{eq:plugin-score-cond-second-app},
\[
 Q_{1,\Gamma}:=\Pn[(\widehat s_i^2-m_i)h_i^2]
\]
is a sum of conditionally independent, centered variables.  The two
bounds in \eqref{eq:Hhat-empirical-scale-app} give conditional envelope
\(C\bar d V_\Gamma\) and variance proxy \(C\bar d V_\Gamma^2\).
Conditional Bernstein, with
\(x=(a+C)\log(n\bar d)\), and a union bound over the polynomial-size
contrast family therefore yield
\begin{equation}\label{eq:response-square-fluctuation-app}
 \max_\Gamma\frac{|Q_{1,\Gamma}|}{V_\Gamma}
 \le C\left\{\sqrt{\frac{\bar d x}{n}}
              +\frac{\bar d x}{n}\right\}
 \le Cr_n
\end{equation}
with probability at least \(1-n^{-a}\).

\textbf{Step 3: conditional mean and the information equation.}
The Lipschitz properties of \(I\) and \(\sigma\), together with the
entrywise rate for \(\Delta=\widehat\Theta-\Thetastar\), imply
\[
 \max_i|m_i-I(\widehat\eta_i)|\le Cr_n.
\]
Consequently, by \eqref{eq:Hhat-empirical-scale-app},
\[
 \left|\Pn[m_i h_i^2]-\Pn[I(\widehat\eta_i)h_i^2]\right|
 \le Cr_nV_\Gamma.
\]
Let \(\widehat G_{\rm ev}\) be the evaluation-fold operator in
Lemma~\ref{lem:foldwise-geometry-app}.  The information-subsample
estimating equation and \eqref{eq:foldwise-info-conc-app} give
\begin{align*}
 \Pn[I(\widehat\eta_i)h_i^2]
 &=\ip{\widehat H_\Gamma}{\widehat G_{\rm ev}\widehat H_\Gamma},\\
 \left|\ip{\widehat H_\Gamma}
              {(\widehat G_{\rm ev}-\widehat G)\widehat H_\Gamma}\right|
 &\le {Cr_n\over\dstar}\Fnorm{\widehat H_\Gamma}^2
 \le Cr_nV_\Gamma,\\
 \ip{\widehat H_\Gamma}{\widehat G\widehat H_\Gamma}
 &=\ip{\widehat P_{\mathbb T}\Gamma}{\widehat H_\Gamma}.
\end{align*}
Put \(D_\Gamma=\widehat H_\Gamma-H_\Gamma^\star\).  Since
\(D_\Gamma\in\mathbb T+\widehat{\mathbb T}\), the local information
 bound \eqref{eq:foldwise-info-conc-app}, the relative
direction bound \eqref{eq:relative-direction-app}, and the upper Fisher
bound imply
\[
 \left|
 \ip{\widehat H_\Gamma}{\widehat G\widehat H_\Gamma}
 -\ip{H_\Gamma^\star}{G H_\Gamma^\star}
 \right|
 \le {Cr_n\over\dstar}\Fnorm{H_\Gamma^\star}^2
 \le Cr_nV_\Gamma.
\]
Because
\(\ip{H_\Gamma^\star}{G H_\Gamma^\star}=V_\Gamma\), Steps 2--3 show
\[
 \left|\Pn\widehat\phi_\Gamma^2-V_\Gamma\right|
 \le Cr_nV_\Gamma.
\]

\textbf{Step 4: empirical centering.}
Conditional on \(\mathcal G_k\),
\(\E(\widehat s_i h_i\mid\mathcal G_k)
=(p_i^\star-\widehat p_i)h_i\).  Cauchy--Schwarz and
\eqref{eq:Hhat-empirical-scale-app} bound this conditional mean by
\(Cr_n\sqrt{V_\Gamma}\).  Conditional Bernstein applied to
\(\Pn[(Y_i-p_i^\star)h_i]\) gives, uniformly,
\[
 \frac{|\Pn\widehat\phi_\Gamma|}{\sqrt{V_\Gamma}}
 \le C\left\{r_n+\sqrt{\frac{x}{n}}
                  +\frac{\sqrt{\bar d}\,x}{n}\right\}
 \le Cr_n.
\]
Thus \((\Pn\widehat\phi_\Gamma)^2\le Cr_n^2V_\Gamma\), which is absorbed
by the preceding \(Cr_nV_\Gamma\) bound and proves
\eqref{eq:fdim-var-cons-app}.

\textbf{Step 5: square-loss bound used for covariance estimation.}
The same event also gives the load-bearing comparison
\begin{equation}\label{eq:if-square-var-app}
 \Pn(\widehat\phi_\Gamma-\phi_\Gamma)^2
 \le Cr_n^2V_\Gamma.
\end{equation}
Indeed,
\[
 \widehat\phi_\Gamma-\phi_\Gamma
 =(\widehat s-s^\star)\ip{\widehat H_\Gamma}{X}
  +s^\star\ip{D_\Gamma}{X}.
\]
The first square is at most
\(Cr_n^2\Pn h_i^2\le Cr_n^2V_\Gamma\).  For the second,
 \(|s^\star|\le1\), the derived energy bound
 \eqref{eq:foldwise-energy-app}, and the
bounded-logit lower bound give
\[
 \Pn\ip{D_\Gamma}{X}^2
 \le {C\over\dstar}\Fnorm{D_\Gamma}^2
 \le {Cr_n^2\over\dstar}\Fnorm{H_\Gamma^\star}^2
 \le Cr_n^2V_\Gamma.
\]
\end{proof}

\subsection{Covariance consistency in correlation form}\label{app:fdim-proof-cov}

\begin{proposition}[Covariance consistency]\label{prop:fdim-cov-cons-app}
Let
\(\widehat\Sigma_{jk}=\Pn[(\widehat\phi_j-\Pn\widehat\phi_j)
(\widehat\phi_k-\Pn\widehat\phi_k)]\), where the juxtaposition in this
display denotes multiplication, i.e.
\(\Pn[(\widehat\phi_j-\Pn\widehat\phi_j)
      (\widehat\phi_k-\Pn\widehat\phi_k)]\).
With probability at least \(1-n^{-a}\),
\begin{equation}\label{eq:fdim-cov-abs-app}
 \max_{j,k}|\widehat\Sigma_{jk}-\Sigma_{jk}|
 \le C(\Sigma_{jj}+\Sigma_{kk})r_n.
\end{equation}
If the diagonal variances are comparable, then
\begin{equation}\label{eq:fdim-cov-rho-app}
 \max_{j,k}|\widehat\rho_{jk}-\rho_{jk}|\le Cr_n .
\end{equation}
\end{proposition}

\begin{proof}
We do not polarize through \(\Gamma_j+\Gamma_k\), since that contrast
can cancel and need not satisfy alignment.  Instead write
\[
 \widehat\phi_j\widehat\phi_k-\phi_j\phi_k
 =(\widehat\phi_j-\phi_j)\widehat\phi_k
  +\phi_j(\widehat\phi_k-\phi_k).
\]
Equation~\eqref{eq:if-square-var-app} and the uncentered second-moment
bound in Proposition~\ref{prop:fdim-var-cons-app} give
\[
 \Pn(\widehat\phi_j-\phi_j)^2\le C\Sigma_{jj}r_n^2,\quad
 \Pn\widehat\phi_j^2\le C\Sigma_{jj},\quad
 \Pn\phi_j^2\le C\Sigma_{jj},
\]
where the last inequality follows from the same oracle Bernstein bound.
Cauchy--Schwarz therefore bounds the plug-in part by
\(C\sqrt{\Sigma_{jj}\Sigma_{kk}}r_n\).
Bernstein applied directly to
\(\phi_j\phi_k-\E(\phi_j\phi_k)\) gives the same bound for the empirical
part, uniformly over \((j,k)\).  Since
\(2\sqrt{\Sigma_{jj}\Sigma_{kk}}\le\Sigma_{jj}+\Sigma_{kk}\),
the corresponding uncentered empirical covariance has the bound in
\eqref{eq:fdim-cov-abs-app}.  Passing to the centered covariance
subtracts \((\Pn\widehat\phi_j)(\Pn\widehat\phi_k)\).  The uniform
conditional mean bound in Step~4 of
Proposition~\ref{prop:fdim-var-cons-app}, followed by
\(2|xy|\le x^2+y^2\), shows that this additional term is
\(O\{(\Sigma_{jj}+\Sigma_{kk})r_n^2\}\), which is absorbed by
the right-hand side of~\eqref{eq:fdim-cov-abs-app}.

For correlations, Proposition~\ref{prop:fdim-var-cons-app} gives
\(\widehat\Sigma_{jj}/\Sigma_{jj}=1+O(r_n)\).  Use
\[
 |\sqrt{x}-1|=\frac{|x-1|}{\sqrt{x}+1}\le|x-1|,\qquad x>0,
\]
in both diagonal factors and combine with
\eqref{eq:fdim-cov-abs-app}.  Diagonal comparability yields
\eqref{eq:fdim-cov-rho-app}.
\end{proof}

\begin{remark}[Why no relative-error bound on \(\Sigma_{jk}\)]\label{rem:cov-no-relative-app}
Off-diagonal covariances may vanish, so the diagonal-normalized absolute
bound~\eqref{eq:fdim-cov-abs-app}, rather than relative error in
\(\Sigma_{jk}\), is the uniform statement used below.
\end{remark}
\subsection{Joint Loewner efficiency restatement}\label{app:fdim-proof-eff}

Combining
Proposition~\ref{prop:fdim-loewner-app} with
Theorem~\ref{thm:fdim-clt-app}, our one-step estimator attains the
Loewner-minimal asymptotic covariance \(\Sigma\), so it is jointly
semiparametrically efficient for the fixed finite contrast family
\(\{\Gamma_1,\ldots,\Gamma_q\}\).  This proves the efficiency claim
referenced in Section~\ref{sec:fdim}.

\subsection{Diagonal scale of \texorpdfstring{$\Sigma_{jj}$}{Sigma\_jj} for score-gap contrasts}\label{app:fdim-proof-Sigmajj}

For score-gap contrasts \(\Gamma=e_t(e_m-e_{m'})^\top\), we record the
explicit scaling of \(\Sigma_{jj}\) used in
Appendix~\ref{app:ranking-proof}.

\begin{lemma}[Diagonal scale of efficient variance for score gaps]\label{lem:Sigma-scale-app}
Under Assumptions~\ref{ass:score-app}--\ref{ass:alignment-app},
\(\Sigma_{jj}=V_{\rm eff}(\Gamma)\asymp \dt+\dm\) for every score-gap
contrast \(\Gamma=e_t(e_m-e_{m'})^\top\).  The standard error of the one-step
estimator is therefore \(\sigma_j/\sqrt n\asymp\sqrt{(\dt+\dm)/n}\), and the
simultaneous calibration over \(p\) score gaps gives a band width of
order \(\sqrt{d\log p/n}\).
\end{lemma}

\begin{proof}
Under near-uniform pairwise sampling and Fisher comparability,
\(\Sigma_{jj}=V_{\rm eff}(\Gamma)=
\ip{P_{\mathbb T}\Gamma}{A_{\mathbb T}^{-1}P_{\mathbb T}\Gamma}\).
Lemma~\ref{lem:weighted-moment-app} puts the eigenvalues of
\(A_{\mathbb T}^{-1}\) between constant multiples of \(\dstar\).
The tangent-projector upper envelope and
Assumption~\ref{ass:alignment-app} give
\(\Fnorm{P_{\mathbb T}\Gamma}^2\asymp1/\dt+1/\dm\).  Hence
\(\Sigma_{jj}\asymp\dstar(1/\dt+1/\dm)=\dt+\dm\).
\end{proof}

\section{Proof of Theorem~\ref{thm:rank-one-task}, Corollary~\ref{cor:rank-all-tasks}, and the top-\texorpdfstring{$K$}{K} extension}\label{app:ranking-proof}

This appendix proves the simultaneous ranking-inference results of
Section~\ref{sec:ranking}.  We use the Chernozhukov--Chetverikov--Kato
(CCK) high-dimensional approximate-means framework, which we state in
the form needed and then verify each constituent error term explicitly,
in order, in subsequent subsections.  We condition throughout on the
master good event \(\mathcal E_n\) of
Appendix~\ref{app:notation-probcalib}.

\subsection{Setup: contrast family and statistics}\label{app:ranking-proof-setup}

For a contrast family \(\mathcal J\) (indexed by score-gap contrasts as
in the three applications below), let
\(\Delta_j:=\psi_{\Gamma_j}(\Thetastar)\), and adopt the standardized
oracle and plug-in coordinates from
Appendix~\ref{app:fdim-proof-onestep}:
\[
    \phi_j(W_i)=s_\eta(Y_i,\eta_i^\star)\ip{H_j^\star}{X_i},
    \qquad
    Z_{ij}=\frac{\phi_j(W_i)}{\sigma_j},
    \qquad
    \widehat Z_{ij}=\frac{\widehat\phi_j(W_i)}{\widehat\sigma_j},
\]
with
\(\widehat\sigma_j^2:=\Pn[(\widehat\phi_j-\Pn\widehat\phi_j)^2]\),
\(\sigma_j^2=\Sigma_{jj}=V_{\rm eff}(\Gamma_j)\).  Define the cardinality
\(p:=|\mathcal J|\), which will be polynomial in \(\bar d\) for each
application below.  For the multiplier statistics, center both arrays:
\[
 Z_{ij}^{\rm c}:=Z_{ij}-\Pn Z_j,
 \qquad
 \widehat Z_{ij}^{\rm c}
 :=\frac{\widehat\phi_j(W_i)-\Pn\widehat\phi_j}{\widehat\sigma_j}.
\]

\textbf{Oracle and feasible test statistics.}
\begin{equation}\label{eq:T0-T-app}
    T_0:=\max_{j\in\mathcal J}\Bigl|\frac{1}{\sqrt n}\sum_{i=1}^n Z_{ij}\Bigr|,
    \qquad
    T:=\max_{j\in\mathcal J}\frac{|\sqrt n(\widehat\Delta_j-\Delta_j)|}{\widehat\sigma_j}.
\end{equation}

\textbf{Oracle and feasible multiplier-bootstrap statistics.}
With i.i.d.\ multipliers \(\xi_i\sim N(0,1)\) independent of the data,
\begin{equation}\label{eq:W0-Tstar-app}
    W_0^{\rm u}:=\max_{j\in\mathcal J}\Bigl|\frac{1}{\sqrt n}\sum_{i=1}^n\xi_i Z_{ij}\Bigr|,
    \qquad
    W_0^{\rm c}:=\max_{j\in\mathcal J}\Bigl|\frac{1}{\sqrt n}\sum_{i=1}^n\xi_i Z_{ij}^{\rm c}\Bigr|,
    \qquad
    T^\ast:=\max_{j\in\mathcal J}\Bigl|\frac{1}{\sqrt n}\sum_{i=1}^n\xi_i\widehat Z_{ij}^{\rm c}\Bigr|.
\end{equation}
Here \(W_0^{\rm u}\), rather than \(W_0^{\rm c}\), is the exact oracle
multiplier statistic to which the CCK theorem is applied; the centered
statistic is only an intermediate device for comparison with the implemented
bootstrap \(T^\ast\).
Let \(c_{1-\alpha}^\ast\) denote the conditional \((1-\alpha)\)-quantile
of \(T^\ast\) given the data.

\textbf{Reference Gaussian.}
Let \(\Sigma_Z:=(\Sigma_{Z,jk})_{j,k\in\mathcal J}\) with
\(\Sigma_{Z,jk}:=\E^\star[Z_jZ_k]=\Sigma_{jk}/(\sigma_j\sigma_k)\), and
let \(G\sim N(0,\Sigma_Z)\), \(Z_0:=\max_{j\in\mathcal J}|G_j|\).

\subsection{The CCK approximate-means theorem and the master decomposition}\label{app:ranking-proof-cck-master}

Apply the result to the \(2p\) coordinates \((Z_{ij},-Z_{ij})\), so
absolute maxima are ordinary coordinate maxima.  The relevant result is
Theorem~3.2 of \citet{chernozhukov2013gaussian}, not Theorem~3.1:
Theorem~3.1 treats exact averages, whereas both our statistic and our
feasible bootstrap are approximate averages.  In the notation of that
theorem we must verify numbers \(\zeta_1,\zeta_2\) such that
\begin{equation}\label{eq:cck-approx-conditions-app}
 \Pr(|T-T_0|>\zeta_1)<\zeta_2,\qquad
 \Pr\{\Pr_\xi(|T^\ast-W_0^{\rm u}|>\zeta_1)>\zeta_2\}<\zeta_2 .
\end{equation}

Let
\[
 a_n=\max_j\left|
 {\sqrt n(\widehat\Delta_j-\Delta_j)\over\sigma_j}
 -{1\over\sqrt n}\sum_iZ_{ij}\right|,\quad
 b_n=\max_j\left|{\widehat\sigma_j\over\sigma_j}-1\right|,
\]
and
\(c_n^2=\max_j\Pn(\widehat Z_j^{\rm c}-Z_j^{\rm c})^2\),
\(q_n=\max_j|\Pn Z_j|\).  The lemmas below provide deterministic
envelopes
\[
 \bar a_n=C C_A r_n,\qquad \bar b_n=C r_n,\qquad
 \bar c_n=C r_n,\qquad
 \bar q_n=C\left\{\sqrt{\frac{\log(2p/\zeta_2)}n}
       +{L_n\log(2p/\zeta_2)\over n}\right\},
\]
such that \(a_n\le\bar a_n\), \(b_n\le\bar b_n\),
\(c_n\le\bar c_n\), and \(q_n\le\bar q_n\) except on an event of
probability at most \(C_0n^{-a}\).  Increase the tail parameter in the
constituent lemmas so that each failure probability is at most
\(n^{-a}/C_0\), and set the deterministic number
\(\zeta_2=C_0n^{-a}\).

On \(b_n\le1/2\), algebra gives
\(|T-T_0|\le C(a_n+b_nT_0)\).  Coordinatewise Bernstein and a union
bound give, after calibrating this constituent event to have probability
at most \(\zeta_2/2\),
\[
 T_0\le M_n:=C\left\{\sqrt{\log(2p/\zeta_2)}
          +{L_n\log(2p/\zeta_2)\over\sqrt n}\right\}.
\]
Consequently \(|T-T_0|\le C(\bar a_n+\bar b_nM_n)\) on the stated
event.
Conditionally on the data,
\[
 \Var_\xi\!\left\{
 {1\over\sqrt n}\sum_i\xi_i(\widehat Z_{ij}^{\rm c}-Z_{ij}^{\rm c})\right\}
 =\Pn(\widehat Z_j^{\rm c}-Z_j^{\rm c})^2\le c_n^2.
\]
A union bound for the \(2p\) centered Gaussian coordinates therefore
gives, with conditional failure probability at most \(\zeta_2/2\),
\[
 |T^\ast-W_0^{\rm c}|\le C\bar c_n\sqrt{\log(4p/\zeta_2)}.
\]
Moreover,
\[
 |W_0^{\rm c}-W_0^{\rm u}|
 \le q_n\left|n^{-1/2}\sum_i\xi_i\right|
 \le \bar q_n\sqrt{2\log(4/\zeta_2)}
\]
with conditional probability at least \(1-\zeta_2/2\).  A conditional
union bound therefore makes the total multiplier-comparison failure
probability at most \(\zeta_2\).  Thus
\eqref{eq:cck-approx-conditions-app} holds with the deterministic numbers
\begin{equation}\label{eq:zeta-map-app}
 \zeta_1=C\max\{\bar a_n+\bar b_nM_n,\
       \bar c_n\sqrt{\log(4p/\zeta_2)}
       +\bar q_n\sqrt{\log(4/\zeta_2)}\},
 \qquad \zeta_2=C_0n^{-a}.
\end{equation}

Let
\(\Delta_n=\max_{j,k}|\Pn(Z_jZ_k)-\Pstar(Z_jZ_k)|\) and
\(\pi(\vartheta)=C\vartheta^{1/3}
\{1\vee\log(2p/\vartheta)\}^{2/3}\).
Theorem~3.2 of \citet{chernozhukov2013gaussian}, together with its
Gaussian comparison lemma, now yields
\begin{equation}\label{eq:cck-master-app}
 \sup_{\alpha\in(0,1)}
 \left|\Pr\{T\le c_{1-\alpha}^\ast\}-(1-\alpha)\right|
 \le
 3\rho_n+2\{\pi(\vartheta)+\Pr(\Delta_n>\vartheta)\}
 +C\zeta_1\sqrt{1\vee\log(2p/\zeta_1)}+5\zeta_2 ,
\end{equation}
where
\(\rho_n=\sup_z|\Pr(T_0\le z)-\Pr(Z_0\le z)|\).
For later reference denote the right side by
\begin{equation}\label{eq:En-master-app}
 \mathfrak E_n:=
 3\rho_n+2\{\pi(\vartheta)+\Pr(\Delta_n>\vartheta)\}
 +C\zeta_1\sqrt{1\vee\log(2p/\zeta_1)}+5\zeta_2 .
\end{equation}
The next subsections verify every term and show
\(\mathfrak E_n=o(1)\).
\subsection{Bounding the standardized envelope \(L_n\) (and ingredient I)}\label{app:ranking-proof-Ln}

We compute an explicit envelope for the standardized oracle coordinates.

\begin{lemma}[Envelope of standardized oracle coordinates]\label{lem:Ln-envelope-app}
Under Assumptions~\ref{ass:score-app}--\ref{ass:alignment-app},
for every \(j\in\mathcal J\) and every admissible design \(X_i\),
\begin{equation}\label{eq:Ln-envelope-app}
    |Z_{ij}|
    \;\le\;
    L_n,
    \qquad
    L_n^2
    \;\le\;
    C(\mu,r,\kappa,B,c_B,C_B)\,\bar d,
\end{equation}
Consequently \(L_n^2\log^7(pn)/n\le C\bar d\log^7(pn)/n\to 0\)
under Assumption~\ref{ass:sample-app} and the CLT condition
\(C_A\sqrt{\bar d\log^c(n\bar d)/n}\to 0\) of
Theorem~\ref{thm:fdim-clt}, and so the oracle CCK Gaussian
approximation error
\[
    \rho_n
    :=
    \sup_z|\Pr(T_0\le z)-\Pr(Z_0\le z)|
    \;\to\;0.
\]
\end{lemma}

\begin{proof}
The proof is a direct calculation.  Recall \(Z_{ij}=\phi_j(W_i)/\sigma_j\).

\textbf{Step 1: bound \(|\phi_j(W_i)|\) using the tangent envelope.}
Since \(s_\eta(Y_i,\eta_i^\star)\in[-1,1]\) almost surely
(Assumption~\ref{ass:score-app}(iv)),
\[
    |\phi_j(W_i)|
    =
    |s_\eta(Y_i,\eta_i^\star)|\cdot|\ip{H_j^\star}{X_i}|
    \le
    |\ip{H_j^\star}{X_i}|.
\]
Using \(H_j^\star\in\mathbb T\), \(\ip{H_j^\star}{X_i}=\ip{H_j^\star}{P_{\mathbb T}X_i}\),
and Cauchy--Schwarz with the tangent envelope
(Lemma~\ref{lem:tangent-envelope-app}),
\[
    |\ip{H_j^\star}{X_i}|
    \le
    \Fnorm{H_j^\star}\cdot\Fnorm{P_{\mathbb T}X_i}
    \;\le\;
    C(\mu,r)\,\Fnorm{H_j^\star}\sqrt{\bar d/\dstar}.
\]

\textbf{Step 2: bound \(\sigma_j^2\) from below using the Frobenius reduction.}
By Lemma~\ref{lem:weighted-moment-app},
\(\sigma_j^2=\E^\star[s_\eta^2\ip{H_j^\star}{X}^2]\ge c_B\E^\star\ip{H_j^\star}{X}^2\),
and by the Frobenius reduction (Lemma~\ref{lem:frob-reduction-app})
\(\E^\star\ip{H_j^\star}{X}^2\asymp\Fnorm{H_j^\star}^2/\dstar\), so
\[
    \sigma_j^2
    \;\ge\;
    \frac{c_B}{C}\,\frac{\Fnorm{H_j^\star}^2}{\dstar}.
\]

\textbf{Step 3: combine.}
\[
    Z_{ij}^2
    =
    \frac{\phi_j(W_i)^2}{\sigma_j^2}
    \;\le\;
    \frac{C(\mu,r)\,\Fnorm{H_j^\star}^2(\bar d/\dstar)}{(c_B/C)\,\Fnorm{H_j^\star}^2/\dstar}
    \;=\;
    C(\mu,r,c_B)\,\bar d.
\]
The factor \(\Fnorm{H_j^\star}^2\) appears in numerator and denominator
and cancels, and this cancellation is the same mechanism that drives the
empirical square-loss bound in
Appendix~\ref{app:ranking-proof-eif}.

\textbf{Step 4: oracle CCK approximation.}
The envelope implies, for \(k=1,2\),
\[
 \E^\star|Z_{ij}|^{2+k}\le L_n^k\E^\star Z_{ij}^2=L_n^k,
 \qquad \E^\star\exp(|Z_{ij}|/L_n)\le e.
\]
Hence the bounded-coordinate moment and tail conditions of the CCK
Gaussian approximation theorem
\citep[Theorem~2.1]{chernozhukov2017central} hold.  Applied to the
i.i.d.\ mean-zero coordinates \(Z_{ij}\), with
\(\E^\star Z_{ij}^2=1\), it shows that if
\(L_n^2\log^7(pn)/n\to 0\) then
\(\rho_n=o(1)\).  Substituting \(L_n^2\le C\bar d\) and
\(p\le\bar d^{O(1)}\) yields the displayed sufficient condition.
\end{proof}

\subsection{Bounding the empirical covariance error \(\Delta_n\) (ingredient II)}\label{app:ranking-proof-cov}

We now bound the maximum standardized covariance estimation error
\(\Delta_n\) by a Bernstein argument, with explicit calculation of the
variance proxy and envelope.

\begin{lemma}[Maximum standardized covariance error]\label{lem:Deltan-app}
Fix any \(a>0\).  With probability at least \(1-n^{-a}\),
\begin{equation}\label{eq:Deltan-app}
    \Delta_n
    :=
    \max_{j,k\in\mathcal J}\bigl|\Pn[Z_jZ_k]-\Pstar[Z_jZ_k]\bigr|
    \;\le\;
    C\Bigl(L_n\sqrt{\frac{\log(pn)}{n}}+L_n^2\frac{\log(pn)}{n}\Bigr)
    \;\lesssim\;
    \sqrt{\frac{\bar d\,\log(pn)}{n}}+\frac{\bar d\log(pn)}{n}.
\end{equation}
Consequently, with the same probability, the Gaussian comparison step
contributes
\[
    \pi(\vartheta_n)
    :=
    C\vartheta_n^{1/3}\{1\vee\log(p/\vartheta_n)\}^{2/3}
    =
    o(1)
\]
upon the choice, for a sufficiently large absolute constant \(C_\vartheta\),
\(\vartheta_n=C_\vartheta\{L_n\sqrt{\log(pn)/n}
+L_n^2\log(pn)/n\}\).
\end{lemma}

\begin{proof}
We first apply Bernstein's inequality to the uncentered product required
by the exact CCK oracle.  We also record the coordinate-mean bound,
which controls the difference between uncentered and centered feasible
covariances.

\textbf{Step 1: variance proxy.}
Each \(Z_{ij}Z_{ik}\) is bounded by \(L_n^2\) in absolute value
(Lemma~\ref{lem:Ln-envelope-app}), and has variance
\[
    \Var(Z_{ij}Z_{ik})
    \le
    \Pstar[Z_j^2 Z_k^2]
    \le
    L_n^2\,\Pstar[Z_j^2]
    =
    L_n^2,
\]
using \(\Pstar Z_j^2=1\) by definition.

\textbf{Step 2: per-pair Bernstein.}
For each pair \((j,k)\), Bernstein's inequality gives, for any \(x>0\),
\[
    \Pr\Bigl[\bigl|\Pn[Z_jZ_k]-\Pstar[Z_jZ_k]\bigr|\ge\sqrt{2 L_n^2\,x/n}+L_n^2 x/(3 n)\Bigr]
    \le 2 e^{-x}.
\]
This is the standard Bernstein bound with envelope \(L_n^2\) and
variance \(L_n^2\).

\textbf{Step 3: union bound.}
Set \(x=Ca\log(pn)\) with \(C\) sufficiently large, so the tail is at
most \(2\,p^2 e^{-Ca\log(pn)}\le 2 p^{-Ca+2}n^{-Ca}\le n^{-a-1}\) for
\(C\) large enough (using \(p\le\bar d^{O(1)}\)).  This gives
the stated bound for the uncentered product covariance.  The same
scalar Bernstein inequality gives, simultaneously in \(j\),
\[
 |\Pn Z_j|\le C\sqrt{\log(pn)/n}+CL_n\log(pn)/n.
\]
The product of two such mean bounds is absorbed by the right-hand side
of~\eqref{eq:Deltan-app}; hence the same order also holds for the
centered empirical covariance used by the implemented estimator.

\textbf{Step 4: \(\pi(\vartheta_n)=o(1)\).}
With the stated sufficiently-large-constant choice of \(\vartheta_n\), we have
\(\vartheta_n^{1/3}\le(L_n^2\log(pn)/n)^{1/6}+(L_n^2\log(pn)/n)^{1/3}\)
(treating \(L_n\sqrt{x}\) as the dominant term for small \(x\)), so
\(\pi(\vartheta_n)\lesssim(L_n^2\log(pn)/n)^{1/6}\log^{2/3}(p)\).  Under
\(L_n^2\log^7(pn)/n\to 0\), this is \(o(1)\), confirming
ingredient (II) is asymptotically negligible.
\end{proof}

\subsection{Bounding the one-step plug-in errors \(a_n\) and \(b_n\) (ingredients III, IV)}\label{app:ranking-proof-onestep}

We carry over the bounds from
Appendix~\ref{app:fdim-proof-uniform} and
Proposition~\ref{prop:fdim-var-cons-app}, applied uniformly over the
contrast family \(\mathcal J\).

\begin{lemma}[One-step transfer]\label{lem:an-app}
Under Theorem~\ref{thm:uniform-remainder-app} applied to
\(\mathcal F=\mathcal J\), with probability at least \(1-n^{-a}\),
\begin{equation}\label{eq:an-app}
    a_n
    :=
    \max_{j\in\mathcal J}\Bigl|\frac{\sqrt n(\widehat\Delta_j-\Delta_j)}{\sigma_j}-\frac{1}{\sqrt n}\sum_{i=1}^n Z_{ij}\Bigr|
    \;\le\;
    C\,C_A\,\sqrt{\frac{\bar d\,\mathrm{polylog}(n\bar d)}{n}}.
\end{equation}
\end{lemma}

\begin{proof}
By definition of \(R_n^{\Gamma_j}\) in
Appendix~\ref{app:fdim-proof-decomp},
\(\sqrt n(\widehat\Delta_j-\Delta_j)=\frac{1}{\sqrt n}\sum_i\phi_j(W_i)+\sqrt n R_n^{\Gamma_j}\),
and dividing by \(\sigma_j\),
\(\sqrt n(\widehat\Delta_j-\Delta_j)/\sigma_j=\frac{1}{\sqrt n}\sum_i Z_{ij}+\sqrt n R_n^{\Gamma_j}/\sigma_j\).
Hence \(a_n=\max_j|\sqrt n R_n^{\Gamma_j}/\sigma_j|\).
By the already studentized bound in
Theorem~\ref{thm:uniform-remainder-app},
\(a_n\le C\,C_A\sqrt{\bar d\,\mathrm{polylog}(n\bar d)/n}\) on the
uniform-remainder event.
\end{proof}

\begin{lemma}[Standard-error transfer]\label{lem:bn-app}
With probability at least \(1-n^{-a}\),
\begin{equation}\label{eq:bn-app}
    b_n
    :=
    \max_{j\in\mathcal J}\Bigl|\frac{\widehat\sigma_j}{\sigma_j}-1\Bigr|
    \;\le\;
    C\sqrt{\frac{\bar d\,\mathrm{polylog}(n\bar d)}{n}}.
\end{equation}
\end{lemma}

\begin{proof}
By Proposition~\ref{prop:fdim-var-cons-app} applied to each
\(\Gamma_j\), \(|\widehat V_j/V_j-1|\le C r_n\) with probability
\(1-n^{-a}\) per contrast, where \(V_j=\sigma_j^2\) and
\(\widehat V_j=\widehat\sigma_j^2\).  Take a union bound over the
\(p\le\bar d^{O(1)}\) contrasts at the cost of an extra
\(\sqrt{\log\bar d}\) factor absorbed into the polylog.  Taking square
roots and using the exact identity gives
\[
 \left|\frac{\widehat\sigma_j}{\sigma_j}-1\right|
 =\frac{|\widehat V_j/V_j-1|}{\sqrt{\widehat V_j/V_j}+1}
 \le |\widehat V_j/V_j-1|.
\]
\end{proof}

\subsection{Bounding the feasible-bootstrap error \(c_n\) (ingredient V)}\label{app:ranking-proof-eif}

We now prove the load-bearing bound on the empirical square-loss
between the standardized plug-in and oracle EIF coordinates.

\begin{proposition}[Estimated-EIF squared-loss bound]\label{prop:eif-sqr-bound}
Under the conditions of
Theorem~\ref{thm:fdim-clt-app}, with probability at least \(1-n^{-a}\),
\begin{equation}\label{eq:cn-app}
    c_n^2
    :=
    \max_{j\in\mathcal J}\Pn(\widehat Z_j^{\rm c}-Z_j^{\rm c})^2
    \;\le\;
    C\,r_n^2,
    \qquad
    r_n^2:=\frac{\bar d\,\mathrm{polylog}(n\bar d)}{n}.
\end{equation}
The bound has \emph{no} \(C_A\) factor.
\end{proposition}

\begin{proof}
The proof has four steps.
Since empirical centering is an orthogonal projection in
\(L_2(P_n)\),
\[
 \Pn\{(\widehat Z_j-Z_j)-\Pn(\widehat Z_j-Z_j)\}^2
 \le \Pn(\widehat Z_j-Z_j)^2.
\]
Thus it suffices to prove the uncentered upper bound used in the steps
below.

\textbf{Step 1: standardize away \(\widehat\sigma_j\).}
On the event \(b_n\le 1/2\) (Lemma~\ref{lem:bn-app}),
\[
    |\widehat Z_{ij}-Z_{ij}|
    =
    \Bigl|\frac{\widehat\phi_j(W_i)}{\widehat\sigma_j}-\frac{\phi_j(W_i)}{\sigma_j}\Bigr|
    \le
    \frac{2|\widehat\phi_j(W_i)-\phi_j(W_i)|}{\sigma_j}+2 b_n|Z_{ij}|.
\]
Squaring and averaging,
\[
    \Pn(\widehat Z_j-Z_j)^2
    \le
    \frac{8\Pn(\widehat\phi_j-\phi_j)^2}{\sigma_j^2}+8 b_n^2\Pn Z_j^2.
\]
By definition \(\Pn Z_j^2=1+(\Pn-\Pstar)Z_j^2\), and Bernstein with
envelope \(L_n^2\) and variance \(L_n^2\) gives
\((\Pn-\Pstar)Z_j^2\le L_n\sqrt{\log(pn)/n}+L_n^2\log(pn)/n=o(1)\); hence
\(\max_j\Pn Z_j^2\le 2\) on this event.  Combined with
\(b_n^2\le C r_n^2\) (Lemma~\ref{lem:bn-app}),
\[
    c_n^2\le \frac{8\,\max_j\Pn(\widehat\phi_j-\phi_j)^2/\sigma_j^2}{1}+16\,r_n^2.
\]
The second term is already at the \(r_n^2\) scale; the first term is
the main object to bound below.

\textbf{Step 2: decompose the EIF error.}
Write
\(\widehat\phi_j(W)-\phi_j(W)=\bigl(s_\eta(Y,\widehat\eta)-s_\eta(Y,\eta^\star)\bigr)\ip{\widehat H_j}{X}+s_\eta(Y,\eta^\star)\ip{\widehat H_j-H_j^\star}{X}\),
so
\(\Pn(\widehat\phi_j-\phi_j)^2/\sigma_j^2\le 2(T_{1j}+T_{2j})\) with
\[
    T_{1j}
    :=
    \frac{\Pn[(s_\eta(Y,\widehat\eta)-s_\eta(Y,\eta^\star))^2\ip{\widehat H_j}{X}^2]}{\sigma_j^2},
    \qquad
    T_{2j}
    :=
    \frac{\Pn[s_\eta(Y,\eta^\star)^2\ip{\widehat H_j-H_j^\star}{X}^2]}{\sigma_j^2}.
\]

\textbf{Step 3: score plug-in term \(T_{1j}\).}
By the Lipschitz bound on \(s_\eta\) in \(\eta\)
(Assumption~\ref{ass:score-app}(iii)),
\(|s_\eta(Y,\widehat\eta)-s_\eta(Y,\eta^\star)|\le|\widehat\eta-\eta^\star|\le 2\infnorm{\widehat\Theta-\Thetastar}\).
Hence
\[
    T_{1j}
    \le
    \frac{4\infnorm{\widehat\Theta-\Thetastar}^2\Pn\ip{\widehat H_j}{X}^2}{\sigma_j^2}.
\]
On the bounded-logit event,
\(I(\widehat\eta_i)\ge c_B'>0\).  The evaluation-fold restricted-energy
bound and the relative direction bound therefore give
\[
 \Pn\ip{\widehat H_j}{X}^2
 \le {C\over\dstar}\Fnorm{\widehat H_j}^2
 \le C\sigma_j^2 .
\]
By the entrywise theorem
\eqref{eq:init-entrywise-app},
\(\infnorm{\widehat\Theta-\Thetastar}^2\le C\bar d\,\mathrm{polylog}(n\bar d)/n=Cr_n^2\).
Combining,
\(\max_j T_{1j}\le C r_n^2\) on \(\mathcal E_n\).
\textbf{Step 4: Empirical Fluctuation $T_{2,j}$}
\begin{lemma}[Empirical square bound for estimated efficient directions]
\label{lem:no-ca-direction-square}
Let \(\mathcal J\) be a polynomial-size family of score-gap contrasts, and for each
\(j\in\mathcal J\) write
\[
    H_j^\star := \mathcal A^\dagger\Gamma_j,
    \qquad
    \widehat H_j := \widehat{\mathcal A}^{\dagger}\Gamma_j,
    \qquad
    D_j := \widehat H_j-H_j^\star .
\]
Let
\[
    \sigma_j^2
    :=
    P^\star\!\left[
        \{s_i^\star\}^2\langle H_j^\star,X_i\rangle^2
    \right],
    \qquad
    s_i^\star := s(Y_i,\eta_i^\star),
    \qquad
    \eta_i^\star := \langle X_i,\Theta^\star\rangle .
\]
Suppose that on the good event \(\mathcal E_n\), uniformly over \(j\in\mathcal J\),
\[
    \frac{\|D_j\|_F}{\|H_j^\star\|_F}
    \le r_n,
    \qquad
    r_n^2 :=
    \frac{\bar d\log^c(n\bar d)}{n},
\]
and that the tangent-projection envelope satisfies
\[
    \sup_X \|P_{\mathbb T}X\|_F
    \vee
    \sup_X \|\widehat P_{\mathbb T}X\|_F
    \le
    C\sqrt{\frac{\bar d}{\dstar}} .
\]
For the fixed exponent \(c\ge c_0\) in
Assumption~\ref{ass:sample-app}, with probability at
least \(1-n^{-a}\),
\[
    \max_{j\in\mathcal J}
    P_n
    \frac{
        \{s_i^\star\}^2
        \langle \widehat H_j-H_j^\star,X_i\rangle^2
    }{\sigma_j^2}
    \le
    C r_n^2 .
\]
In particular, this bound does not involve the inverse-information stability factor
\(C_A\).
\end{lemma}

\begin{proof}
The direction \(\widehat H_j\) is independent of the evaluation fold,
and both \(\widehat H_j\in\widehat{\mathbb T}\) and
\(H_j^\star\in\mathbb T\); hence
\(D_j\in\mathbb T+\widehat{\mathbb T}\).  Apply
Lemma~\ref{lem:foldwise-geometry-app} and the relative direction bound
to obtain
\begin{equation}\label{eq:empirical-direction-energy-app}
 \Pn\!\left[I(\widehat\eta)\ip{D_j}{X}^2\right]
 \le {C\over\dstar}\Fnorm{D_j}^2
 \le {Cr_n^2\over\dstar}\Fnorm{H_j^\star}^2
 \le Cr_n^2\sigma_j^2
 \qquad(j\in\mathcal J).
\end{equation}

The bounded-logit event gives
\(I(\widehat\eta_i)\ge c_B'>0\) and \(|s_i^\star|\le1\).
Consequently, \eqref{eq:empirical-direction-energy-app} yields
\[
 \max_{j\in\mathcal J}
 \Pn{(s_i^\star)^2\ip{D_j}{X_i}^2\over\sigma_j^2}
 \le {C\over c_B'}r_n^2.
\]
This bound is deterministic on the master event.
\end{proof}

\textbf{Conclusion.}
Combining steps 1--4,
\[
    c_n^2\le 8(T_{1j}+T_{2j})+16 r_n^2\le C r_n^2.
\]
The probability calibration to \(1-n^{-a}\) follows from
Appendix~\ref{app:notation-probcalib}. 
\end{proof}

\subsection{Aggregate CCK approximate-means error}\label{app:ranking-proof-aggregate}

Combining
Lemmas~\ref{lem:Ln-envelope-app}--\ref{lem:bn-app} and
Proposition~\ref{prop:eif-sqr-bound} with
the master decomposition~\eqref{eq:En-master-app}, we obtain the
aggregate CCK error.

\begin{theorem}[Aggregate CCK approximate-means error]\label{thm:cck-aggregate-app}
Under
Assumptions~\ref{ass:score-app}--\ref{ass:sample-app} and the CLT
condition \(C_A\sqrt{\bar d\log^c(n\bar d)/n}\to 0\) of
Theorem~\ref{thm:fdim-clt}, the aggregate CCK error satisfies
\[
    \mathfrak E_n
    \;=\;
    3\rho_n+2\{\pi(\vartheta_n)+\Pr(\Delta_n>\vartheta_n)\}
    +C\zeta_1\sqrt{1\vee\log(2p/\zeta_1)}+5\zeta_2
    \;=\;o(1),
\]
where \((\zeta_1,\zeta_2)\) is the deterministic map in
\eqref{eq:zeta-map-app} with \(\zeta_2=C_0n^{-a}\).
Consequently
\[
    \Pr(T\le c_{1-\alpha}^\ast)\;\ge\;1-\alpha-o(1),
\]
and the simultaneous score-gap bands
\(\widehat I_j:=[\widehat\Delta_j\pm c_{1-\alpha}^\ast\widehat\sigma_j/\sqrt n]\),
\(j\in\mathcal J\), satisfy
\(\Pr(\Delta_j\in\widehat I_j\;\forall j\in\mathcal J)\ge 1-\alpha-o(1)\).
\end{theorem}

\begin{proof}
Substitute the bounds:
\[
    \rho_n=o(1)\text{ (Lemma~\ref{lem:Ln-envelope-app})},
    \quad
    \pi(\vartheta_n)+\Pr(\Delta_n>\vartheta_n)=o(1)\text{ (Lemma~\ref{lem:Deltan-app})},
\]
\[
    \bar a_n\vee\bar b_n\vee\bar c_n
    \le C C_A\sqrt{\bar d\,\mathrm{polylog}(n\bar d)/n}=:\delta_n=o(1),
\]
where we used \(C_A\ge1\).  Since \(p\le\bar d^{O(1)}\) and
\(\zeta_2=C_0n^{-a}\), \eqref{eq:zeta-map-app} gives
\[
 \zeta_1\le C\delta_n\sqrt{\log(2pn^a)}.
\]
The fixed exponent \(c\ge c_0\) in Assumption~\ref{ass:sample-app}
dominates these displayed logarithmic factors, so the CLT condition yields
\[
 \zeta_1\sqrt{1\vee\log(2p/\zeta_1)}=o(1),
 \qquad \zeta_2=o(1).
\]
Together with the first two displayed bounds, every term in the exact
master expression is \(o(1)\).  Thus \(\mathfrak E_n=o(1)\), and
the master CCK theorem
\eqref{eq:cck-master-app} delivers the coverage statement.  Inverting
\(T\le c_{1-\alpha}^\ast\) gives the simultaneous-band coverage.
\end{proof}

\subsection{Application 1: rank confidence band for one task / one model}\label{app:ranking-proof-app1}

Fix a task \(t\in[\dt]\) and a model \(m\in[\dm]\).  Take
\(\mathcal J_{t,m}:=\{(t,\ell):\ell\ne m\}\), so the gaps are
\(\Delta_{t,\ell}^{(m)}=\Thetastar_{t,\ell}-\Thetastar_{t,m}\) for
\(\ell\ne m\) and \(p=|\mathcal J_{t,m}|=\dm-1\).  Let
\(\widehat I_{t,\ell}^{(m)}=[\widehat L_{t,\ell}^{(m)},\widehat U_{t,\ell}^{(m)}]\)
be the simultaneous bands of
Theorem~\ref{thm:cck-aggregate-app}.  Define
\[
    A_t(m):=|\{\ell\ne m:\widehat L_{t,\ell}^{(m)}>0\}|,
    \qquad
    B_t(m):=|\{\ell\ne m:\widehat U_{t,\ell}^{(m)}<0\}|,
\]
and
\(\widehat{\mathcal R}_t(m):=[1+A_t(m),\,\dm-B_t(m)]\).

\begin{theorem}[Rank confidence band for one task; restatement of Theorem~\ref{thm:rank-one-task}]\label{thm:rank-one-task-app}
Under the conditions of Theorem~\ref{thm:cck-aggregate-app},
\[
    \Pr\bigl\{\rk_t(m)\in\widehat{\mathcal R}_t(m)\bigr\}\;\ge\;1-\alpha-o(1).
\]
\end{theorem}

\begin{proof}
On the simultaneous coverage event of
Theorem~\ref{thm:cck-aggregate-app} applied to
\(\mathcal J_{t,m}\), every \(\ell\) with
\(\widehat L_{t,\ell}^{(m)}>0\) satisfies
\(\Delta_{t,\ell}^{(m)}>0\), i.e.\ \(\ell\) is certified above \(m\), so
\(\rk_t(m)\ge 1+A_t(m)\).  Symmetrically, every \(\ell\) with
\(\widehat U_{t,\ell}^{(m)}<0\) is certified below \(m\), so
\(\rk_t(m)\le\dm-B_t(m)\).  Combining, the simultaneous coverage event
implies \(\rk_t(m)\in\widehat{\mathcal R}_t(m)\).  The probability bound
follows.
\end{proof}

\subsection{Application 2: rank confidence band for one model, all tasks}\label{app:ranking-proof-app2}

For a fixed model \(m\), enlarge to
\(\mathcal J(m):=\{(t,\ell):t\in[\dt],\ell\ne m\}\), so
\(p=\dt(\dm-1)\le\bar d^2\).  This family is still polynomial in
\(\bar d\), so all conditions of
Appendix~\ref{app:ranking-proof-cck-master} are met with the same
scaling.

\begin{corollary}[Simultaneous taskwise rank inference; restatement of Corollary~\ref{cor:rank-all-tasks}]\label{cor:rank-all-tasks-app}
Under the conditions of
Theorem~\ref{thm:rank-one-task-app}, with the bootstrap maximum taken
over \(\mathcal J(m)\),
\[
    \Pr\Bigl\{\rk_t(m)\in\widehat{\mathcal R}_t(m)\;\forall t\in[\dt]\Bigr\}
    \;\ge\;
    1-\alpha-o(1).
\]
\end{corollary}

\begin{proof}
Apply Theorem~\ref{thm:cck-aggregate-app} to the larger family
\(\mathcal J(m)\); the inversion of
Theorem~\ref{thm:rank-one-task-app} now holds simultaneously for every
\(t\in[\dt]\) under the same simultaneous coverage event.  The price
relative to Application~1 is only the additional \(\sqrt{\log\bar d}\)
inflation in \(c_{1-\alpha}^\ast\), already absorbed in the polylog.
\end{proof}

\subsection{Application 3: simultaneous top-\texorpdfstring{$K$}{K} set inference}\label{app:topk-set}

For inference on the entire task-specific top-\(K\) set, enlarge to
\[
    \mathcal J_{\rm all}
    :=
    \{(t,m,\ell):t\in[\dt],\ m,\ell\in[\dm],\ \ell\ne m\},
    \qquad
    p=\dt\dm(\dm-1)\le\bar d^3.
\]
Each test in \(\mathcal J_{\rm all}\) corresponds to the studentized
score-gap statistic for the contrast
\(\Gamma_{t,m,\ell}=e_t(e_\ell-e_m)^\top\).  Define inner and outer
top-\(K\) sets by
\[
    \widehat{\TopK}_{\rm in}(t)
    :=
    \bigl\{m:|\{\ell\ne m:\widehat U_{t,m,\ell}<0\}|\ge\dm-K\bigr\},
\]
\[
    \widehat{\TopK}_{\rm out}(t)
    :=
    \bigl\{m:|\{\ell\ne m:\widehat L_{t,m,\ell}>0\}|<K\bigr\}.
\]

\begin{theorem}[Simultaneous top-\(K\) set inference]\label{thm:topk-set-app}
Under the conditions of
Theorem~\ref{thm:cck-aggregate-app} applied to \(\mathcal J_{\rm all}\),
\[
    \Pr\Bigl\{\widehat{\TopK}_{\rm in}(t)\subseteq\TopK(t)\subseteq\widehat{\TopK}_{\rm out}(t)\;\forall t\in[\dt]\Bigr\}
    \;\ge\;
    1-\alpha-o(1).
\]
\end{theorem}

\begin{proof}
On the simultaneous coverage event for \(\mathcal J_{\rm all}\),
the preceding rank-band argument gives
\[
  1+A_t(m)\le \rk_t(m)\le \dm-B_t(m),
\]
where
\(A_t(m)=|\{\ell\ne m:\widehat L_{t,m,\ell}>0\}|\) and
\(B_t(m)=|\{\ell\ne m:\widehat U_{t,m,\ell}<0\}|\).
If \(m\in\widehat{\TopK}_{\rm in}(t)\), then
\(B_t(m)\ge\dm-K\), so \(\rk_t(m)\le\dm-B_t(m)\le K\), proving
\(m\in\TopK(t)\).  Conversely, if \(m\in\TopK(t)\), then
\(1+A_t(m)\le\rk_t(m)\le K\), so \(A_t(m)<K\) and therefore
\(m\in\widehat{\TopK}_{\rm out}(t)\).
The simultaneous validity over \(t\) follows from the simultaneous
coverage of \(\mathcal J_{\rm all}\) at level \(1-\alpha-o(1)\).
\end{proof}

\subsection{Critical-value vs.\ dimension discussion}\label{app:ranking-proof-cn}

The bootstrap critical value \(c_{1-\alpha}^\ast\) is the
\((1-\alpha)\)-quantile of a maximum of \(p\) correlated approximately
standard-normal coordinates.  In the worst case (weakly dependent or
independent),
\(c_{1-\alpha}^\ast\asymp\sqrt{2\log(2p/\alpha)}\); under strong
correlation, it can be substantially smaller.  Since
\(p\in\{\dm-1,\dt(\dm-1),\dt\dm(\dm-1)\}\le\bar d^3\),
\(\log p\lesssim\log\bar d\) and therefore
\(c_{1-\alpha}^\ast=O(\sqrt{\log\bar d})\).  Combining with
\(\Sigma_{jj}\asymp\bar d\) and
\(\widehat\sigma_j^2/\Sigma_{jj}=1+o_{\Pr}(1)\)
(Lemma~\ref{lem:Sigma-scale-app}), the worst-case simultaneous band
width is of order \(c_{1-\alpha}^\ast\widehat\sigma_j/\sqrt n\asymp\sqrt{\bar d\log\bar d/n}\),
matching the rate predicted by the entrywise estimation theorem
(Theorem~\ref{thm:entrywise}) up to logarithmic factors.

\end{document}